\documentclass[a4paper,12pt]{article}
\usepackage{indentfirst}
\usepackage{graphicx}
\usepackage{epstopdf}
\usepackage{epsfig}
\usepackage{overpic}
\usepackage{overcite}
\usepackage{array}
\usepackage{color}
\usepackage{multirow}
\usepackage{float}
\usepackage{subfigure}
\usepackage{amsmath,amssymb, amsthm}
\usepackage[mathscr]{euscript}
\usepackage{latexsym,bm}
\usepackage{cite}
\usepackage{booktabs}
\usepackage{bbm}
\usepackage{diagbox}
 \usepackage[all]{xy}
 \usepackage{color}

\newtheorem{theorem}{Theorem}

\newtheorem{definition}{Definition}
\newtheorem{example}{\textup{\textbf{Example}}}

\newcommand\D{\textup{d}}

\newcommand\comment[1]{}

{\theoremstyle{remark} \newtheorem{remark}{\textup{\textbf{Remark}}}}

\def\dif{\mathrm{d}}
\def\mi{\mathbbm{i}}
\def\me{\mathbbm{e}}
\def\mone{\mathbbm{1}}

\def\fx{\mathcal{X}}

\def\fk{\mathcal{K}}

\begin{document}
\bibliographystyle{unsrt}

\title{{A computable branching random walk for the many-body Wigner quantum dynamics}}
\author{Sihong Shao\footnotemark[2], \and
Yunfeng Xiong\footnotemark[2]}
\renewcommand{\thefootnote}{\fnsymbol{footnote}}
\footnotetext[2]{LMAM and School of Mathematical Sciences, Peking University, Beijing 100871, China. Email addresses: {\tt sihong@math.pku.edu.cn} (S. Shao).}
%\footnotetext[1]{To
%whom correspondence should be addressed. Email:
%\texttt{sihong@math.pku.edu.cn}}
\date{\today}
\maketitle

% .......... the text...................

%{\no\bf Mathematics Subject Classifications:}\hspace{0.2cm} 65F10,
%65W05

\begin{abstract}
A branching random walk algorithm for the many-body Wigner
equation and its numerical applications for quantum dynamics in phase space 
are proposed and analyzed.
After introducing an auxiliary function,
the (truncated) Wigner equation is cast into
the integral formulation as well as its adjoint correspondence,
both of which can be reformulated into
the renewal-type equations and have transparent probabilistic interpretation. We prove that the first moment of a branching random walk happens to be the solution for the adjoint equation. More importantly,
we detail that such stochastic model,
associated with both importance sampling and resampling,
paves the way for a numerically tractable scheme,
within which the Wigner quantum dynamics is simulated in a time-marching manner and  the complexity can be controlled with the help of an (exact) estimator of the growth rate of particle number.
Typical numerical experiments on the Gaussian barrier scattering
and a Helium-like system validate our theoretical findings, as well as demonstrate the accuracy, the efficiency and thus the computability of the Wigner branching random walk algorithm.

\vspace*{4mm}

\end{abstract}

\section{Introduction}
\label{sec:intro}

Connections between partial differential equations (PDE) and stochastic processes are always heated topics in modern mathematics and provide powerful tools for both probability theory and analysis, especially for PDE of elliptic and parabolic type\cite{bk:Doob2001, bk:Dynkin2002}.
In the past few decades, their numerical applications have also burgeoned with a lot of developments, such as the ensemble Monte Carlo method for the Boltzmann transport equation\cite{NedjalkovVitanov1990,KosinaNedjalkovSelberherr2000,KosinaNedjalkovSelberherr2003,NedjalkovKosinaSelberherr2003},
the random walk method for the Laplace equation\cite{YanCaiZeng2013} and the diffusion Monte Carlo method for the Schr\"{o}dinger equation\cite{KosztinFaberSchulten1996,HairerWeare2014}. In particular, the diffusion Monte Carlo method allows us to go beyond the mean-field approximation and offer a reliable ground state solution to quantum many-body systems.
In this work, we focus on the probabilistic approach to the equivalent phase space formalism of quantum mechanics,
namely, the Wigner function approach\cite{Wigner1932},
which bears a close analogy to classical mechanics. In recent years, the Wigner equation has been drawing growing attention\cite{tatarskiui1983,Zurek1991,JacoboniBordone2004,DiasPrata2004}
and widely used
in nanoelectronics\cite{bk:MarkowichRinghoferSchmeiser1990,th:Biegel1997}, non-equilibrium statistical mechanics\cite{bk:Balescu1975}, quantum optics\cite{bk:Schleich2011}, and many-body quantum systems \cite{SellierNedjalkovDimov2015}.
Actually, a branch of experiment physics in the community of quantum tomography are devoting to reconstructing the Wigner function from measurements\cite{bk:Leonhardt1997,LeibfriedPfauMonroe1998}.
Moreover, the intriguing mathematical structure of the Weyl-Wigner correspondence has also been employed in the deformation quantization\cite{Zachos2002}.

In contrast to its great theoretical advantages, the Wigner equation is extremely difficult to be solved
because of the high dimensionality of the phase space as well as the highly oscillating structure
of the Wigner function due to the spatial coherence\cite{Zurek1991,LeibfriedPfauMonroe1998}.
Although several efficient deterministic solvers, e.g.,
the conservative spectral element method (SEM)\cite{ShaoLuCai2011} and
the third-order advective-spectral-mixed scheme (ASM)\cite{XiongChenShao2015}, have enabled
an accurate transient simulation in 2D and 4D phase space, they are still restricted by the limitation of
data storage and increasing computational complexity. One possible approach to solving the higher dimensional problems is the Wigner Monte Carlo (WMC) method, which displays $N^{-\frac12}$ convergence ($N$ is the number of samples), regardless of the dimensionality, and scales much better on the parallel computing platform\cite{bk:QuerliozDollfus2010, SellierNedjalkovDimov2015}.

The proposed work is motivated by a recently developed stochastic method, termed the signed particle Wigner Monte Carlo method (spWMC)\cite{NedjalkovKosinaSelberherrRinghoferFerry2004, NedjalkovSchwahaSelberherr2013,SellierNedjalkovDimov2014}.
This method utilizes the
branching of signed particles to capture the quantum coherence, and the numerical accuracy has been validated in 2D situations   \cite{ShaoSellier2015,MuscatoWagner2016,th:Ellinghaus2016}. Very recently, it has been also validated theoretically by exploiting the connection between a piecewise-deterministic Markov process and the weak formulation of the Wigner equation\cite{Wagner2016}.
In this work,
we use an alternative approach to
constructing the mathematical framework for spWMC from the viewpoint of computational mathematics,
say, we focus on the probabilistic interpretation of the mild solution of the (truncated) Wigner equation and its adjoint correspondence.
In particular, we would like to stress that
the resulting stochastic model, the importance sampling and the resampling are three cornerstones of a computable scheme for simulating the many-body Wigner quantum dynamics.

Our first purpose is to explore the inherent relation between the Wigner equation and a stochastic branching random walk model,
as sketched by the diagram below.
\begin{equation*}
\boxed{\small \text{Wigner equation}} \xrightarrow[\gamma(\bm{x})]{\textup{integral form}} \boxed{\small \text{Renewal-type equation}} \xleftarrow{\textup{moment}}
\boxed{\small \text{Branching random walk}}
\end{equation*}
With an auxiliary function $\gamma(\bm{x})$,
we can cast the Wigner equation (as well as its adjoint equation)
into a renewal-type integral equation and prove that
its solution is equivalent to the first moment of a stochastic branching random walk. In this manner, we arrive at the stochastic interpretation of the Wigner quantum dynamics, termed the {\sl Wigner branching random walk} (WBRW) in this paper. In particular, the $y$-truncated WBRW method recovers the popular spWMC method which needs a discretization of the momentum space beforehand.

Although the probabilistic interpretation of the Wigner equation naturally gives rises to a statistical method, in practice we have encountered two major problems.
First, such numerical method is point-wise in nature and not very efficient in general  unless we are only interested in the solution at specified points\cite{bk:Liu2001}.
Second, the number of particles in a branching system will grow exponentially in time\cite{bk:Harris1963},
indicating that the complexity increases dramatically for a long-time simulations.  Thus, our second purpose is to discuss how to overcome these two obstacles.
As for the first, we introduce the dual system of the Wigner equation and derive an equivalent form of the inner product problem, which allows us to draw weighted samples according to the initial Wigner distribution. Besides, by exploiting the principle of importance sampling, we can give a sound interpretation to several fundamental concepts in spWMC, such as particle sign and particle weight. For the second problem, we firstly derive the exact growth rate of branched particles, which reads $\me^{2 M \gamma_0 t}$ in time $t$, with $M$ pairs of potentials and a constant auxiliary function $\gamma(\bm{x})\equiv\gamma_0$ and then illustrate the basic idea of resampling to control the particle number within a reasonable size.
Roughly speaking, we make a histogram through the weighted particles and resample from it at the next step. Such a self-consistent scheme allows us to evolve the Wigner quantum dynamics in a time-marching manner and choose appropriate resampling frequencies to control the computational complexity.

The rest of this paper is organized as follows.
Section~\ref{sec:wigner} reviews briefly the Wigner formalism of
quantum mechanics. From both theoretical and numerical aspects,  it is more convenient to discuss
 the truncated Wigner equation, instead of the Wigner equation
itself. Thus in Section~\ref{sec:twigner}, we illustrate two typical
ways to truncate the Wigner equation, termed the $k$-truncated and
the $y$-truncated models. Section~\ref{sec:integral_form} manifests
the equivalence between the $k$-truncated Wigner model and a
renewal-type integral equation, where an auxiliary function
$\gamma(\bm{x})$ is used to introduce a probability measure.
Besides, the set of adjoint equation renders an
equivalent representation of the inner product problem. We will show
that such representation,
as well as the importance sampling, plays a vital role in WMC and serves as the motivation of WBRW.
In Section~\ref{sec:branching}, we will prove that the first moment of a branching
random walk is exactly the solution of the adjoint equation. This probabilistic
approach not only validates the branching process treatment, but also allows us to study the
mass conservation and exponential growth of particle number rigorously. After theoretical
analysis,  we turn to discuss the main
idea of the resampling procedure and present the
numerical challenges in high dimensional problems.
Section~\ref{sec:num_res} investigates the performance of WBRW by employing SEM or ASM as the reference. The
paper is concluded in Section~\ref{sec:con}.

\section{The Wigner equation}
\label{sec:wigner}

In this section, we briefly review the Wigner representation of quantum mechanics. The Wigner function $f(\bm{x}, \bm{k}, t)$ living in the phase space $(\bm{x},\bm{k})\in\mathbb{R}^{2d}$ for position $\bm{x}$ and wavevector $\bm{k}$,
\begin{equation}
f(\bm{x}, \bm{k}, t) = \int_{\mathbb{R}^{d}} \textup{d} \bm{y} ~ \me^{-\mi \bm{k} \cdot \bm{y}} \rho(\bm{x}+\frac{\bm{y}}{2}, \bm{x}-\frac{\bm{y}}{2}, t),
\end{equation}
is defined by the Weyl-Wigner transform of the density matrix
\begin{equation}
\rho(\bm{x}_{1},\bm{x}_{2},t) =\sum_{i}p_{i}\Psi_{i}(\bm{x}_{1},t)\Psi^{\dagger}_{i}(\bm{x}_{2},t),
\end{equation}
where $p_{i}$ gives the probability of occupying the $i$-th state,
$2d$ denotes the degree of freedom (2$\times$particle number$\times$dimensionality).
Although {\sl it possibly has negative values},
the Wigner function serves the role as a density function due to the following properties\cite{tatarskiui1983,JacoboniBordone2004}
\begin{itemize}
\item $f(\bm{x}, \bm{k}, t)$ is a real function.
\item $ \iint_{\mathbb{R}^{d} \times \mathbb{R}^{d}} f(\bm{x}, \bm{k}, t) \D \bm{x}\D \bm{k}= 1$.
\item The average of a quantum operator $\hat{A}$ can be written in a form
\begin{equation}\label{average_A_wf}
\langle \hat{A} \rangle_t= \iint_{\mathbb{R}^{d} \times \mathbb{R}^{d}} ~ A(\bm{x}, \bm{k}) f(\bm{x}, \bm{k}, t) \D \bm{x}\D \bm{k},
\end{equation}
with $A(\bm{x}, \bm{k})$ the corresponding classical function in phase space.
\end{itemize}

In particular, we can define the {\sl Wigner (quasi-) probability} $W_D$ on a domain $D$ by taking $A(\bm{x}, \bm{k})= \mone_{D}(\bm{x}, \bm{k})$
\begin{equation}\label{eq:wigner_prob}
W_D(t) = \iint_{D} f(\bm{x}, \bm{k}, t) \dif\bm{x}\dif\bm{k}.
\end{equation}

To derive the dynamics of the Wigner function, we evaluate its first
derivative through the Schr\"{o}dinger equation (or the quantum
Liouville equation)
\begin{equation}
\mi \hbar \frac{\partial}{\partial t}\Psi_i(\bm{x}, t)=-\frac{\hbar^{2}}{2m} \nabla^{2}_{\bm{x}}\Psi_i(\bm{x}, t)+V(\bm{x}, t)\Psi_i(\bm{x}, t),
\end{equation}
combine with the Fourier completeness relation
\begin{equation}
\label{Def::Fourier completeness relation}
\delta(\bm{k}-\bm{k}^{\prime})=\frac{1}{(2\pi)^d} \int_{\mathbb{R}^d} \D \bm{y} ~\me^{\mi (\bm{k}-\bm{k}^{\prime}) \cdot \bm{y}},
\end{equation}
and then obtain the Wigner equation
\begin{equation}
\frac{\partial }{\partial t}f(\bm{x}, \bm{k}, t)+\frac{\hbar \bm{k}}{m} \cdot \nabla_{\bm{x}} f(\bm{x},\bm{k}, t)=\Theta_{V}\left[f\right](\bm{x}, \bm{k}, t), \label{eq.Wigner}
\end{equation}
where
\begin{align}
\Theta_{V}\left[f\right](\bm{x}, \bm{k}, t)
&= \int_{\mathbb{R}^{d}} \textup{d} \bm{\bm{k}^{\prime}} f(\bm{x},\bm{k}^{\prime},t)V_{w}(\bm{x},\bm{k}-\bm{k}^{\prime},t),
\label{eq:pd_operator}\\
V_{w}(\bm{x},\bm{k},t)&=\frac{1}{\mi\hbar (2\pi)^{d}}\int_{\mathbb{R}^{d}} \text{d}\bm{y} \me^{-\mi\bm{k}\cdot \bm{y}}D_{V}(\bm{x}, \bm{y}, t), \label{Wigner_kernel} \\
D_{V}(\bm{x}, \bm{y}, t)&=V(\bm{x}+\frac{\bm{y}}{2}, t)-V(\bm{x}-\frac{\bm{y}}{2}, t). \label{Dv}
\end{align}
Here the nonlocal pseudo-differential term $\Theta_{V}[f](\bm{x}, \bm{k}, t)$ contains the quantum information,
 $D_{V}(\bm{x}, \bm{y}, t)$ denotes a central difference of the potential function $V(\bm{x}, t)$,
the Wigner kernel $V_{w}(\bm{x}, \bm{k}, t)$ is defined   through the Fourier transform of $D_{V}(\bm{x}, \bm{y}, t)$,
$\hbar$ is the reduced Planck constant and $m$ is the particle mass (for simplicity, we assume all particles have the same mass throughout this work).
Equivalently, we can first perform the integration in $\bm{k}^{\prime}$-space and arrive at another way to formulate the pseudo-differential term
\begin{align}
\Theta_{V}\left[f\right](\bm{x}, \bm{k}, t)&=\frac{1}{\mi \hbar}\int_{\mathbb{R}^{d}} \D \bm{y} D_{V}(\bm{x},\bm{y}, t) \widehat{f}(\bm{x}, \bm{y}, t) \me^{-\mi \bm{k}\cdot \bm{y}},
\label{eq:pd_operator_1}\\
\widehat{f}(\bm{x}, \bm{y}, t)&=\frac{1}{(2\pi)^{d}} \int_{\mathbb{R}^{d}} \D \bm{k}^{\prime} f(\bm{x}, \bm{k}^{\prime}, t) \me^{\mi \bm{k}^{\prime}\cdot \bm{y}}:=\mathcal{F}^{-1}\left[f\right](\bm{x}, \bm{y}, t).
\label{eq:fhat}
\end{align}
Actually, $\widehat{f}(\bm{x}, \bm{y}, t)$ is just another notation for $\rho(\bm{x}+\frac{\bm{y}}{2}, \bm{x}-\frac{\bm{y}}{2}, t)$.

One of the most important properties of the Wigner equation lies in the anti-symmetry of the Wigner kernel
\begin{equation}\label{vw_anti_symmetry}
V_{w}(\bm{x},\bm{k}, t)=-V_{w}(\bm{x},-\bm{k}, t),
\end{equation}
then a simple calculation yields
\begin{equation}\label{mass_conservation}
\int_{\mathbb{R}^{d}} \text{d} \bm{k} \int_{\mathbb{R}^{d}} \text{d} \bm{k^{\prime}} f(\bm{x},\bm{k}^{\prime},t)V_{w}(\bm{x},\bm{k}-\bm{k}^{\prime},t)=0,
\end{equation}
which corresponds to the conservation of the zeroth moment (i.e., total particle number or mass)
\begin{equation}\label{eq:mass}
\frac{\textup{d}}{\textup{dt}}\iint_{\mathbb{R}^{d} \times \mathbb{R}^{d}} f(\bm{x}, \bm{k},t) \textup{d}\bm{x} \textup{d}\bm{k}=0.
\end{equation}

Although the Wigner equation is completely equivalent to
the Schr\"{o}dinger equation in the full space, we would like to
point out that such an equivalence is not necessarily true for the
truncated Wigner equation (see, e.g.\cite{JiangCaiTsu2011}), since for
example the truncation of $\bm{y}$-domain may break the Fourier
completeness relation~\eqref{Def::Fourier completeness relation}.
Therefore, we must be more careful when doing benchmark tests for
stochastic Wigner simulations by adopting the Schr\"{o}dinger
wavefunction as the reference\cite{SellierNedjalkovDimov2014,MuscatoWagner2016,th:Ellinghaus2016},
because the underlying models may not
be the same. This also gives rise to the demanding for highly accurate deterministic
algorithms, such as SEM\cite{ShaoLuCai2011} and ASM\cite{XiongChenShao2015}, which can
be used to produce a reliable reference solution as already did
in\cite{ShaoSellier2015}.

\section{The truncated Wigner equation}
\label{sec:twigner}

In order to numerically solve the Wigner equation, we need to
discuss the truncated Wigner equation on a bounded domain.  It
should be noted that the double integrations with respect to
$\bm{k}^\prime$ and $\bm{y}$ in the pseudo-differential operator
(see Eq.~\eqref{eq:pd_operator} or \eqref{eq:pd_operator_1})
involves the infinite domain due to the Fourier transform nature,
posing a formidable challenge in seeking numerical approximations.
Intuitively, an feasible way is either truncating $\bm{k}$-space
first or truncating $\bm{y}$-space first, denoted below by the
$k$-truncated and $y$-truncated models, respectively. It is worth
noting that no matter what kind of truncation we choose, the mass
conservation~\eqref{eq:mass} should be maintained in the resulting
model as the physical requirement, which may yield additional
constraints.

\subsection{The $k$-truncated Wigner equation}
\label{sec:twigner:k}

A feasible way to formulate the Wigner equation in a bounded domain is to exploit the decay of the Wigner function when $\left|\bm{k}\right| \to \infty$. Thus we only need to evaluate the Wigner function $f(\bm{x}, \bm{k}, t)$  in a finite domain $\mathcal{K}=[-L_{1}, L_{1}]\times [-L_{2}, L_{2}]\cdots \times [-L_{d}, L_{d}]$ ($L_i>0$) and a simple nullification can be adopted outside $\mathcal{K}$, that yields the $k$-truncated Wigner equation
\begin{equation}
\begin{split}
\frac{\partial }{\partial t}f(\bm{x}, \bm{k}, t)+\frac{\hbar \bm{k}}{m} \cdot \nabla_{\bm{x}} f(\bm{x},\bm{k}, t)&=\int_{\mathcal{K}} \D \bm{\bm{k}^{\prime}} f(\bm{x},\bm{k}^{\prime},t)V_{w}(\bm{x},\bm{k}-\bm{k}^{\prime},t)\\
&=\int_{\mathcal{K}} \D \bm{\bm{k}^{\prime}} f(\bm{x},\bm{k}^{\prime},t)V^{T}_{w}(\bm{x},\bm{k}-\bm{k}^{\prime},t),\label{eq.k_truncated_Wigner}
\end{split}
\end{equation}
with the truncated Wigner kernel
\begin{equation}\label{Wigner_kernel_finite}
V^{T}_{w}(\bm{x}, \bm{k},t) =V_{w}(\bm{x}, \bm{k},t) \prod_{i=1}^{d} \textup{rect}(\frac{k_{i}}{4L_{i}}),
\end{equation}
where the rectangular function $\textup{rect}(k)$ is given by
\begin{equation}
\textup{rect}(k)=\left\{\begin{split}& 1, \quad \left|k\right|<\frac{1}{2},\\
&0, \quad \left|k\right| \geq \frac{1}{2}.\end{split}\right.
\end{equation}
The truncated Wigner kernel $V_{w}^{T}$ in Eq.~\eqref{Wigner_kernel_finite}
is used only in the case that the close form of $V_{w}$ is not available.
According to Eq.~\eqref{eq.k_truncated_Wigner},
it deserves to be mentioned that only a restriction of the Wigner kernel on a finite bandwidth $2\mathcal{K}=[-2L_{1}, 2L_{1}]\times [-2L_{2}, 2L_{2}]\cdots \times [-2L_{d}, 2L_{d}]$ (i.e., $\bm{k}-\bm{k}^{\prime}\in 2\fk$ when both $\bm{k}$ and $\bm{k}^\prime$ belong to $\fk$) is required.
Furthermore, it can be easily verified that
\begin{equation}\label{eq:mass1}
\int_{\mathbb{R}^{d}} V_{w}^{T}(\bm{x}, \bm{k}, t) \D \bm{k}=\int_{2\mathcal{K}} V_{w}(\bm{x}, \bm{k}, t) \D \bm{k}=0,
\end{equation}
and thus
\begin{equation}\label{eq:mass_k_truncated}
\frac{\textup{d}}{\textup{dt}}\int_{\mathbb{R}^{d}} \int_{\mathcal{K}} f(\bm{x}, \bm{k},t) \textup{d}\bm{x} \textup{d}\bm{k}=0.
\end{equation}
We expect that any reliable deterministic or stochastic method should preserve this property.

In general, the truncated Wigner kernel $V_{w}^{T}$ can be evaluated
by the Poisson summation formula
 \begin{equation}\label{Poisson_summation_truncated}
V^{T}_{w}(\bm{x}, \bm{k}, t)\approx\frac{1}{ \mi \hbar (2\pi)^{d}} \sum_{\bm{\mu}\in \mathbb{Z}^{d}} \left[(\prod_{i=1}^d \Delta y_{i} ) D_{V}(\bm{x}, \bm{y}_{\bm{\mu}},t)\me^{-\mi \bm{k} \cdot y_{\bm{\mu}}}\right],
\end{equation}
provided that $V_{w}(\bm{x}, \bm{k})$ decays for $|k_{i}|>2|\mathcal{K}_{i}|$ with $i=1,2, \cdots, d$.
Here $\bm{y}_{\bm{\mu}}=(\mu_1\Delta y_1, \mu_2\Delta y_2, \cdots, \mu_d\Delta y_d)$ with $y_i$ being the spacing and $\mu_i \in \mathbb{Z}$.
In this situation, we need to add the constraint\cite{XiongChenShao2015}
\begin{equation}\label{conservation_condition}
2L_{i}\Delta y_{i}=2\pi, \,\,\, i=1,2, \cdots, d,
\end{equation}
to both maintain the mass conservation \eqref{eq:mass_k_truncated}
and avoid the overlapping between $V^{T}_{w}$ and its adjacent image.

To sum up,
we would like to list several advantages of the $k$-truncated model.
\begin{itemize}
\item The $k$-truncated Wigner equation is defined over the continuous $\bm{k}$-space and thus a continuous momentum sampling can be allowed \cite{Wagner2016,MuscatoWagner2016}.
\item It preserves the definition of the Wigner kernel and avoids the artificial periodic extension of $D_{V}(\bm{x}, \bm{y}, t)$ in $\bm{y}$-space.
\item When the Weyl-Wigner transform of $V(\bm{x})$ has a close form, we can obtain the explicit formula of the Wigner kernel and avoid the artificial periodic extension of $V_{w}$ in $\bm{k}$-space.
\end{itemize}

The price to pay is that the sampling in the continuous
$\bm{k}$-space needs intricate techniques, such as a
rejection-acceptance method or the Markov chain Monte Carlo
strategies.

\subsection{The $y$-truncated Wigner equation}
\label{sec:twigner:y}

The other way, used in\cite{SellierNedjalkovDimov2014},
is based on the fact that the inverse Fourier transformed Wigner function $\widehat{f}(\bm{x}, \bm{y}, t)$ defined in Eq.~\eqref{eq:fhat}
decays when $|\bm{y}|\to \infty$.
Thus, we can focus on $\widehat{f}(\bm{x}, \bm{y}, t)$ on a bounded domain $\mathcal{Y}=[-L_{1}, L_{1}]\times [-L_{2}, L_{2}]\cdots \times [-L_{d}, L_{d}](L_i>0)$, and define the truncated pseudo-differential operator as
\begin{equation}\label{Def:truncated_y_wigner_potential}
\Theta_{V}^{T}\left[f\right](\bm{x}, \bm{k}, t)=\frac{1}{\mi \hbar}\int_{\mathcal{Y}} \D \bm{y} D_{V}(\bm{x},\bm{y}, t)\widehat{f}(\bm{x}, \bm{y}, t) \me^{-\mi \bm{k}\cdot \bm{y}}.
\end{equation}
With the assumption that it decays at $y_{i}>L_{i}$, we can evaluate $\widehat{f}(\bm{x}, \bm{y}, t)$ at a finite bandwidth through the Poisson summation formula
\begin{equation}\label{Poisson_sum_f}
\widehat{f}(\bm{x}, \bm{y}, t)\approx\frac{1}{(2\pi)^{d}} \sum_{\bm{m}\in \mathbb{Z}^{d}} \left[(\prod_{i=1}^d \Delta k_{i} ) f(\bm{x}, \bm{m} \Delta \bm{k}, t)\me^{\mi \bm{y} \cdot \bm{m} \Delta \bm{k}}\right], ~~ \bm{y}\in \mathcal{Y},
\end{equation}
where $\bm{m }\Delta \bm{k}=(m_1\Delta k_1, m_2\Delta k_2, \cdots, m_d\Delta k_d)$ with $\Delta k_i$ being the spacing, $m_i \in Z$, $i=1, 2, \cdots d$.

Substituting Eq.~\eqref{Poisson_sum_f} into Eq.~\eqref{Def:truncated_y_wigner_potential} leads to
\begin{align}
\Theta_{V}^{T}\left[f\right](\bm{x}, \bm{k}, t) &\approx \sum_{\bm{m} \in \mathbb{Z}^{d}} f(\bm{x}, \bm{m} \Delta \bm{k}, t)\tilde{V}_{w}(\bm{x}, \bm{k}-\bm{m}\Delta \bm{k}, t),
\label{eq:pd_operator_sum}\\
\tilde{V}_{w}(\bm{x}, \bm{k}, t)&=\frac{1}{\mi \hbar}\frac{1}{\left|\mathcal{Y}\right|} \int_{\mathcal{Y}}\D \bm{y} ~D_{V}(\bm{x}, \bm{y}, t) \me^{-\mi \bm{y} \cdot \bm{k}}.\label{eq:vwtilde}
\end{align}
Here we have let $\left|\mathcal{Y}\right|=2L_{1} \times 2L_{2} \cdots \times 2L_{d}$ and used the constraint
\begin{equation}\label{discrete_conservation_condition}
2L_{i}\Delta k_{i}=2\pi, \,\,\, i=1,2, \cdots, d,
\end{equation}
which serves as the sufficient and necessary condition
to establish the semi-discrete mass conversation
\begin{equation}
\label{eq:mass_y_truncated}
\frac{\textup{d}}{\textup{dt}}\int_{\mathbb{R}^{d}} \textup{d}\bm{x} \sum_{\bm{n}\in \mathbb{Z}^{d}} f(\bm{x}, \bm{n}\Delta \bm{k},t)\Delta \bm{k}  =0.
\end{equation}

Suppose the Wigner function at discrete samples
$\bm{k}=\bm{n}\Delta \bm{k}$ are wanted,
then we immediately arrive at
the $y$-truncated (or semi-discrete) Wigner equation \cite{ArnoldLangeZweifel2000,Goudon2002,GoudonLohrengel2002}
\begin{equation}\label{eq.discrete_Wigner}
\begin{split}
\frac{\partial }{\partial t}f(\bm{x}, \bm{n}\Delta \bm{k}, t)+&\frac{\hbar \bm{n}\Delta \bm{k}}{m} \cdot \nabla_{\bm{x}} f(\bm{x},\bm{n}\Delta \bm{k}, t)\\
&=\sum_{\bm{m}\in \mathbb{Z}^{d}} f(\bm{x},\bm{m}\Delta \bm{k},t)\tilde{V}_{w}(\bm{x},\bm{n}\Delta \bm{k}-\bm{m}\Delta \bm{k},t),
\end{split}
\end{equation}
which indeed provides a straightforward way for stochastic simulations
as used in the spWMC method, and
possesses the following properties.
\begin{itemize}
\item The modified Wigner kernel $\tilde{V}_{w}$ in Eq.~\eqref{eq:vwtilde} can be treated as the Fourier coefficients of  $D_{V}(\bm{x}, \bm{y}, t)$ (with a periodic extension), and can be recovered by the inverse Fourier transform (possibly by the inverse fast Fourier transform).
\item The continuous convolution is now replaced by a discrete convolution (see Eqs.~\eqref{eq:pd_operator} and \eqref{eq:pd_operator_sum}), so that the sampling in discrete $\bm{k}$-space can be simply realized in virtue of the cumulative distribution function.
\item The set of equidistant sampling in $\bm{k}$-space facilitates the data storage and the code implementation.
\end{itemize}

Although both truncated models approximate the original problem in some extent, their range of applicability  is different. In fact, the modified Wigner potential $\tilde{V}_{w}$ in $y$-truncated model is not a trivial approximation to the original Wigner potential (one can refer to the difference between the Fourier coefficients and continuous Fourier transformation). The convergence $\tilde{V}_{w} \to V_{w}$ is only valid when $|\mathcal{Y}| \to \infty$, or the potential $V(\bm{x}, t)$ decays rapidly at the boundary of the finite domain (but this condition is not satisfied for, e.g., the Coulomb-like potential, especially for the Coulomb interaction between two particles).  By contrast,  the $k$-truncated model is based on relatively milder assumption, and it is not necessary to change the definition of the Wigner kernel unless the Poisson summation formula is used. Thus we would like to stress that \emph{the $k$-truncated Wigner equation is more appropriate for simulating many-body quantum systems} and thus adopted hereafter.

\section{Renewal-type integral equations}
\label{sec:integral_form}

In order to establish the connection between the deterministic partial integro-differential equation~\eqref{eq.k_truncated_Wigner} and a stochastic process, we need to cast the deterministic equation into a renewal-type integral equation. For this purpose, the first crucial step is to  introduce an exponential distribution in its integral formulation via {\it an auxiliary function $\gamma(x)$}.
The second one is to split the Wigner kernel into several positive parts\cite{NedjalkovKosinaSelberherrRinghoferFerry2004}, such that each part can be endowed with a probabilistic interpretation.
More importantly,
to make the resulting branching random walk computable, we derive the adjoint equation of the Wigner equation and obtain an equivalent representation of the inner product~\eqref{average_A_wf}, which explicitly depends on the initial Wigner distribution. Therefore, it provides a much more efficient way to draw samples on the phase space, and naturally gives rise to several important features of spWMC, such as the particle sign and particle weight.

%Without loss of generality, in the subsequent discussion we focus on the $k$-truncated Wigner equation~\eqref{eq.k_truncated_Wigner}, and a similar idea can be straightforwardly generalized to the $y$-truncated equation~\eqref{eq.discrete_Wigner} as well as the original Wigner equation~\eqref{eq.Wigner}.

\subsection{Integral formulation with an auxiliary function}
\label{sec:branching:int}

The first step is to cast Eq.~\eqref{eq.k_truncated_Wigner} into a renewal-type equation. To this end, we can introduce
an {auxiliary function} $\gamma(\bm{x})$ and add the term $\gamma(\bm{x})f(\bm{x}, \bm{k}, t)$ in both sides of Eq.~\eqref{eq.k_truncated_Wigner},
yielding
\begin{equation}
\begin{split}
\frac{\partial }{\partial t}f(\bm{x}, \bm{k}, t)+&\frac{\hbar \bm{k}}{m} \cdot \nabla_{\bm{x}} f(\bm{x},\bm{k}, t)+\gamma(\bm{x})f(\bm{x}, \bm{k}, t)\\
&=\int_{\mathcal{K}} \D \bm{\bm{k}^{\prime}}~ f(\bm{x},\bm{k}^{\prime},t)\left[V_{w}(\bm{x},\bm{k}-\bm{k}^{\prime},t)+\gamma(\bm{x})\delta(\bm{k}-\bm{k}^{\prime})\right]. \label{eq.Wigner_equiv}
\end{split}
\end{equation}
At this stage, we only consider a nonnegative bounded
$\gamma(\bm{x})$, though a time-dependent $\gamma(\bm{x}, t)$ can be
also introduced if necessary and analyzed in a similar way. In
particular, \emph{we strongly recommend the readers to choose a constant
$\gamma(\bm{x})\equiv \gamma_0$} in real applications, for the convenience of both
theoretical analysis and numerical computation (vide post).
Formally, we can write down its integral formulation through the
variation-of-constant formula
\begin{equation}\label{abstract_integral_form}
\begin{split}
f(\bm{x}, \bm{k}, t)=&\me^{t\mathcal{A}}f(\bm{x}, \bm{k}, 0)+\int_{0}^{t} \me^{(t-t^{\prime}) \mathcal{A} }\left[\mathcal{B}(\bm{x}, \bm{k}, t^{\prime})+\gamma(\bm{x})\right] f(\bm{x}, \bm{k}, t^{\prime}) \textup{d}t^{\prime},
\end{split}
\end{equation}
where $\me^{t\mathcal{A}}$ denotes the semigroup generated by the operator
\begin{equation}
\mathcal{A}=-\hbar \bm{k}/m \cdot \nabla_{\bm{x}}-\gamma(\bm{x}),
\end{equation}
and
\begin{equation}
\mathcal{B}(\bm{x},\bm{k}, t)f(\bm{x},\bm{k}, t)=\int_{\mathcal{K}} \D \bm{\bm{k}^{\prime}}~ f(\bm{x},\bm{k}^{\prime},t)V_{w}(\bm{x},\bm{k}-\bm{k}^{\prime},t)
\end{equation}
is the convolution operator which is assumed to be a bounded operator throughout this work.

When $\gamma(\bm{x})$ is bounded, it only imposes a Lyapunov perturbation on a hyperbolic system, so that the operator $\me^{t\mathcal{A}}$ is also a $\textup{C}_{0}$-semigroup\cite{bk:Pazy1983}. To further determine how the operator $\me^{t\mathcal{A}}$ acts on a given function $u(\bm{x}, \bm{k}, t)\in C^{1}(L^{2}(\mathbb{R}^{2d}), [0,T])$, we need to solve the following evolution system
\begin{equation}
\frac{\partial }{\partial t}u(\bm{x}, \bm{k}, t)+\frac{\hbar \bm{k}}{m} \cdot \nabla_{\bm{x}} u(\bm{x},\bm{k}, t)+\gamma(\bm{x})u(\bm{x}, \bm{k}, t)=0.
\end{equation}
After performing the coordinate conversion\cite{DimovNedjalkovSellierSelberherr2015}
\begin{equation}
\left\{\begin{split}& \bm{x}^{\prime}=\bm{x}-\hbar \bm{k} t/m,\\
&\bm{k}^{\prime}= \bm{k},\\
&t^{\prime} = t, \end{split}\right.
\end{equation}
we obtain
\begin{equation}
\frac{\partial}{\partial t^{\prime}}u^\prime(\bm{x}^{\prime}, \bm{k}^{\prime}, t^{\prime})=-\gamma^\prime(\bm{x}^{\prime}, t^{\prime})u^\prime(\bm{x}^{\prime}, \bm{k}^{\prime}, t^{\prime}),
\end{equation}
where $u^\prime(\bm{x}^{\prime}, \bm{k}^{\prime}, t^{\prime}):=u(\bm{x}, \bm{k}, t)$
and $\gamma^\prime(\bm{x}^{\prime}, t^{\prime}):=\gamma(\bm{x}^{\prime}+\hbar \bm{k}^{\prime} t^{\prime} /m)=\gamma(\bm{x})$.
The solution to the above system reads
\begin{equation}\label{eq:tmp0}
{u}^\prime(\bm{x}^{\prime}, \bm{k}^{\prime}, t^{\prime})=\me^{-\int_{0}^{t^\prime} \gamma^\prime(\bm{x}^{\prime}, s) \D s} {u}^\prime(\bm{x}^{\prime}, \bm{k}^{\prime}, 0).
\end{equation}
Replacing $u^\prime, {\gamma}^\prime$ by $u, \gamma$
and making a shift $\bm{x}^{\prime} \to \bm{x} = \bm{x}^{\prime} + \hbar \bm{k}t/m$
in Eq.~\eqref{eq:tmp0} leads to
\begin{equation}
\me^{t\mathcal{A}} u(\bm{x}, \bm{k}, 0)=\me^{-\int_{0}^{t} \gamma(\bm{x}(t-s)) \D s}u(\bm{x}(t), \bm{k}, 0),
\end{equation}
where
\begin{equation}\label{Def:backward-in-time trajectory}
\bm{x}(\Delta t)=\bm{x}-{\hbar \bm{k}\Delta t}/{m}
\end{equation}
is termed the \emph{backward-in-time trajectory} of $(\bm{x}, \bm{k})$
with a positive time increment $\Delta t$.

After a simple variable substitution ($s+t^{\prime} \to s$),
the integral formulation of the Wigner equation becomes
\begin{equation}\label{Wigner_integral_form}
\begin{split}
f(\bm{x}, \bm{k}, t)=&\me^{-\int_{0}^{t} \gamma(\bm{x}(t-s)) \textup{d} s}f(\bm{x}(t), \bm{k}, 0)+\int_{0}^{t} \D t^{\prime}~  \me^{-\int_{t^{\prime}}^{t} \gamma(\bm{x}(t-s)) \D s}\\
 &\times \left[\mathcal{B}(\bm{x}(t-t^{\prime}), \bm{k}, t^{\prime})+\gamma(\bm{x}(t-t^{\prime}))\right]f(\bm{x}(t-t^{\prime}), \bm{k}, t^{\prime}).
\end{split}
\end{equation}
Let
\begin{equation}
\mathcal{H}(t^{\prime} ; \bm{x}, t )=\int_{t^{\prime}}^{t}  \gamma(\bm{x}(t-\tau)) \me^{-\int_{\tau}^{t} \gamma(\bm{x}(t-s)) \D s}~  \D \tau,
\end{equation}
and assume the auxiliary function satisfies
\begin{equation}\label{eq:gamma_condition}
\gamma(\bm{x})\geq 0, \quad \lim_{t^\prime\to-\infty}\int_{t^\prime}^t \gamma(\bm{x}(t-s))\dif s = +\infty, ~~
\forall\, \bm{x}\in\mathbb{R}^d,
\end{equation}
then we have
\begin{equation}
\dif \mathcal{H}(t^{\prime} ; \bm{x}, t ) \geq 0,
\quad \int_{-\infty}^t \dif \mathcal{H}(t^{\prime} ; \bm{x}, t )
= 1,
\end{equation}
implying that
$\mathcal{H}(t^{\prime} ; \bm{x}, t )$
is a probability measure with respect to $t^{\prime}$ for a given $(\bm{x}, t)$ on $t^{\prime} \leq t$,
characterized by the auxiliary function $\gamma(\bm{x})$.
Substituting this measure into Eq.~\eqref{Wigner_integral_form}
gives
\begin{equation}\label{Def:renewal_type_eq}
\begin{split}
f(\bm{x}, \bm{k}, t)&=\left[1-\mathcal{H}(0; \bm{x}, t)\right] f(\bm{x}(t), \bm{k}, 0)+ \int_{0}^{t} \D \mathcal{H}(t^{\prime}; \bm{x}, t) \times\\
& \int_{\mathcal{K}} \D \bm{k}^{\prime}~  f(\bm{x}(t-t^{\prime}), \bm{k}^{\prime}, t^{\prime})  \left\{\frac{V_{w}(\bm{x}(t-t^{\prime}), \bm{k}-\bm{k}^{\prime}, t^{\prime})}{\gamma(\bm{x}(t-t^{\prime}))}+\delta(\bm{k}-\bm{k}^{\prime})\right\},
\end{split}
\end{equation}
which can be regarded as {\it a kind of renewal-type equation} in the
renewal theory\cite{bk:Harris1963,bk:Kallenberg2002}.

Next we turn to consider the Wigner kernel $V_{w}$, that cannot be regarded as a transition kernel directly  due to possible negative values.
Nevertheless, we can regard it as the linear combination of positive semidefinite kernels. In general, the Wigner kernel $V_w$ is composed of $M$ parts
\begin{equation}
V_{w} = V_{w, 1}+V_{w, 2}+\cdots V_{w, M},
\end{equation}
that corresponds to the potential $V=V_1+V_2+\cdots+V_M$, then the Wigner kernel can be split into $M$ pairs
\begin{align}
V_{w} &= V_{w}^{+} - V_{w}^{-},\quad V_{w}^{\pm} = \sum_{m=1}^M V_{w, m}^{\pm},\\
V_{w, m}^{+}(\bm{x},\bm{k}, t) &=\frac{1}{2}\left|V_{w, m}(\bm{x},\bm{k}, t)\right|+\frac{1}{2}V_{w, m}(\bm{x},\bm{k}, t),\\
V_{w, m}^{-}(\bm{x},\bm{k}, t) &=\frac{1}{2}\left|V_{w, m}(\bm{x},\bm{k}, t)\right|-\frac{1}{2}V_{w, m}(\bm{x},\bm{k}, t).
\end{align}
Such splitting of $V_w$ is rather important in dealing with many-body systems since combining it with the Fourier completeness relation \eqref{Def::Fourier completeness relation} helps to reduce the Wigner interaction term \eqref{eq:pd_operator} into lower dimensional integrals.

Owing to the anti-symmetry of $V_{w}$ (see Eq.~\eqref{vw_anti_symmetry}),
it can be easily verified that
\begin{equation}\label{anti_symmetry}
V_{w, m}^{+}(\bm{x},\bm{k}, t)=V_{w, m}^{-}(\bm{x},-\bm{k}, t).
\end{equation}
Thus, it suffices to define a function $\Gamma$ on $t \ge t^{\prime}$,
composed of three terms
\begin{equation}
\begin{split}
\Gamma(\bm{x}(t-t^{\prime}), \bm{k}, t; \bm{x}^{\prime}, \bm{k}^{\prime}, t^{\prime})=&V^{+}_{w}(\bm{x}(t-t^{\prime}),\bm{k}-\bm{k}^{\prime}, t^{\prime}) \cdot \delta(\bm{x}(t-t^{\prime})-\bm{x}^{\prime})\\
&-V^{-}_{w}(\bm{x}(t-t^{\prime}),\bm{k}-\bm{k}^{\prime},t^{\prime})  \cdot \delta(\bm{x}(t-t^{\prime})-\bm{x}^{\prime})\\
&+\gamma(\bm{x}(t-t^{\prime}))\cdot \delta(\bm{k}-\bm{k}^{\prime})\cdot \delta(\bm{x}(t-t^{\prime})-\bm{x}^{\prime}).
\end{split}
\end{equation}

Finally, the $k$-truncated Wigner equation~\eqref{eq.k_truncated_Wigner}
can be cast into a Fredholm integral equation of the second kind
\begin{equation}\label{Wigner_Fredholm_integral_form}
f(\bm{x}, \bm{k}, t)= f_0(\bm{x}, \bm{k}, t)+ \mathcal{S} f(\bm{x}, \bm{k}, t),~ ~ ~0\leq t \leq T,
\end{equation}
where
\begin{align}
f_0(\bm{x}, \bm{k}, t) &=\me^{-\int_{0}^{t} \gamma(\bm{x}(t-s)) \textup{d} s}f(\bm{x}(t), \bm{k}, 0),\\
\mathcal{S} f(\bm{x}, \bm{k}, t)&=\int_{0}^{t} \D t^{\prime}   \int_{\mathbb{R}^d} \D\bm{x}^{\prime}  \int_{\mathcal{K}} \D\bm{k}^{\prime}~K(\bm{x}, \bm{k}, t; \bm{x}^{\prime}, \bm{k}^{\prime}, t^{\prime}) f(\bm{x}^{\prime}, \bm{k}^{\prime}, t^{\prime}),\\
K(\bm{x}, \bm{k}, t; \bm{x}^{\prime}, \bm{k}^{\prime}, t^{\prime})&=\me^{-\int_{t^{\prime}}^{t} \gamma(\bm{x}(t-s)) \D s}  \Gamma(\bm{x}(t-t^{\prime}), \bm{k}, t; \bm{x}^{\prime}, \bm{k}^{\prime}, t^{\prime}), ~~ t\ge t^{\prime}.
\end{align}

Before discussing the probabilistic approach to the integral equation~\eqref{Wigner_Fredholm_integral_form}, we would like first to derive its adjoint equation and attain an equivalent representation of $\langle A \rangle_{T}$, which serves as the cornerstone of WBRW.

\subsection{Dual system and adjoint equation}
\label{sec:branching:dual}

In quantum mechanics,  it's usually more important to study macroscopically observes $\langle \hat{A} \rangle_{t}$, such as the averaged position of particles, electron density, etc, than the Wigner function itself. In this regard, we turn to consider the inner product problem
\begin{equation}\label{Def:inner_product}
\langle g_0, f \rangle= \int_{0}^{T}  \D t  \int_{\mathbb{R}^d}  \D \bm{x}  \int_{\mathcal{K}} \D \bm{k}~g_0(\bm{x}, \bm{k}, t) f(\bm{x}, \bm{k} ,t),
\end{equation}
on the domain $\mathbb{R}^{d} \times \mathcal{K}$ and a finite time interval $[0,T]$.
For instance, to evaluate the average value $\langle \hat{A} \rangle_T$ at a given final time $T$,  we should take
\begin{equation}\label{Def:dual_system_g0}
g_0(\bm{x}, \bm{k},t)=A(\bm{x}, \bm{k})\delta(t-T),
\end{equation}
then
\begin{equation}\label{eq:AT_fT}
\langle \hat{A} \rangle_T = \langle g_0, f \rangle.
\end{equation}

The main goal of this section is to give the explicit formulation of the adjoint equation, starting from  Eq.~\eqref{Wigner_Fredholm_integral_form} and Eq.~\eqref{Def:dual_system_g0}. For brevity, we will assume that the
potential is time-independent,
and thus the kernels becomes
\begin{align}
K(\bm{x}, \bm{k}, t; \bm{x}^{\prime}, \bm{k}^{\prime}, t^{\prime}) &=\me^{-\int_{t^{\prime}}^{t} \gamma(\bm{x}(t-s)) \D s} \Gamma(\bm{x}(t-t^{\prime}), \bm{k}; \bm{x}^{\prime}, \bm{k}^{\prime}), ~ ~ ~t \ge t^{\prime},\\
\label{Def:kernel_function}
\Gamma(\bm{x}, \bm{k}; \bm{x}^{\prime}, \bm{k}^{\prime})&=\left[V^{+}_{w}(\bm{x},\bm{k}-\bm{k}^{\prime})-V^{-}_{w}(\bm{x},\bm{k}-\bm{k}^{\prime})+\gamma(\bm{x})\delta(\bm{k}-\bm{k}^{\prime})\right] \delta(\bm{x}-\bm{x}^{\prime}).
\end{align}

Suppose the kernel $K(\bm{x}, \bm{k}, t; \bm{x}^{\prime}, \bm{k}^{\prime}, t^{\prime})$ is bounded, then it is easy to verify that $\mathcal{S}$ is a bounded linear operator. Accordingly, we can define the adjoint operator $\mathcal{T}=\mathcal{S}^{\ast}$ by
\begin{equation}
\langle g, \mathcal{S}f\rangle=\langle \mathcal{S}^{\ast}g, f\rangle=\langle \mathcal{T}g, f\rangle,
\end{equation}
Applying Theorem 4.6 in\cite{bk:Kress2014} directly into
the Fredholm integral equation of the second kind \eqref{Wigner_Fredholm_integral_form} yields
\begin{equation}\label{eq:Toperator}
\mathcal{T}g(\bm{x}^{\prime}, \bm{k}^{\prime}, t^{\prime})=\int_{t^{\prime}}^{T} \D t \int_{\mathbb{R}^{d}}\D \bm{x} \int_{\mathcal{K}}\D \bm{k}~K(\bm{x}, \bm{k}, t; \bm{x}^{\prime}, \bm{k}^{\prime},t^{\prime}) g(\bm{x}, \bm{k}, t),~~~ t \ge t^{\prime}.
\end{equation}

Formally, it suffices to define
\begin{equation}\label{Def:dual_system}
g(\bm{x}^{\prime}, \bm{k}^{\prime}, t^{\prime})=\mathcal{T}g(\bm{x}^{\prime}, \bm{k}^{\prime}, t^{\prime}) +g_0(\bm{x}^{\prime}, \bm{k}^{\prime}, t^{\prime}),~ ~ ~ 0\leq t^{\prime} \leq T.
\end{equation}
Since
\begin{equation}
\langle g, f\rangle=\langle g, \mathcal{S}f+f_0\rangle=\langle \mathcal{T}g, f\rangle+\langle g, f_0\rangle= \langle g, f\rangle- \langle g_0, f\rangle+ \langle g, f_0\rangle,
\end{equation}
we have
\begin{equation}
\langle g_0, f\rangle = \langle g, f_0\rangle,
\end{equation}
namely
\begin{equation}\label{averaged_M}
\langle \hat{A} \rangle_{T}=\int_0^{T} \D t^{\prime} \int_{\mathbb{R}^d} \D \bm{x}^{\prime} \int_{\mathcal{K}}\D \bm{k}^{\prime} ~f(\bm{x}^{\prime}(t^{\prime}), \bm{k}^{\prime}, 0)\me^{-\int_{0}^{t^{\prime}} \gamma(\bm{x}^{\prime}(t^{\prime}-s)) \D s}g(\bm{x}^{\prime}, \bm{k}^{\prime}, t^{\prime}).
\end{equation}
Furthermore, we perform the coordinate conversion
\begin{equation}
\left\{
\begin{split}
&\bm{r}_0=\bm{x}^\prime(t^\prime) = \bm{x}^\prime-\hbar \bm{k}^\prime t^\prime/m,\\
&\bm{k}_0= \bm{k}^\prime,\\
&t_0= t^\prime,
\end{split}
\right.
\end{equation}
with which the Jacobian determinant
$\frac{\partial(\bm{r}_0,\bm{k}_0,t_0)}{\partial(\bm{x}^{\prime},\bm{k}^{\prime},t^{\prime})}=1$ implying the volume unit keeps unchanged (i.e., $\D \bm{r}_0\D \bm{k}_0 \D t=\D \bm{x}^{\prime} \D \bm{k}^{\prime}\D t^{\prime}$),
and thus Eq.~\eqref{averaged_M} becomes
\begin{equation}\label{averaged_M_2}
\langle \hat{A} \rangle_{T}=\int_0^{T} \D t_0 \int_{\mathbb{R}^d} \D \bm{r}_0 \int_{\mathcal{K}}\D \bm{k}_0 ~f(\bm{r}_0, \bm{k}_0, 0)\me^{-\int_{0}^{t_0} \gamma(\bm{r}_0(s)) \D s}g(\bm{r}_0(t_0), \bm{k}_0, t_0),
\end{equation}
where we have introduced a \emph{forward-in-time trajectory} (in
contrast to the backward-in-time trajectory $\bm{x}(\Delta t)$ given
in Eq.~\eqref{Def:backward-in-time trajectory}) as follows
\begin{equation}\label{Def:forward-in-time trajectory}
\bm{r}_0(\Delta t)=\bm{r}_0+{\hbar \bm{k}_0\Delta t}/{m},
\end{equation}
with $\Delta t \ge 0$ being the time increment.
Actually, Eq.~\eqref{averaged_M_2} motivates us to combine the exponential factor with $g$ and define a new function $\varphi(\bm{r}, \bm{k}, t)$ as
\begin{equation}\label{Def:varphi}
\varphi(\bm{r}, \bm{k}, t)=\int_{t}^{T}\D t^{\prime} \me^{-\int_{t}^{t^{\prime}} \gamma(\bm{r}(s-t)) \D s} g(\bm{r}(t^{\prime}-t), \bm{k}, t^{\prime}).
\end{equation}
Please keep in mind that, it is required $t^\prime\geq t$ for convenience in the definition~\eqref{Def:varphi}, before which $t^\prime\leq t$ is always assumed, for example, see Eq.~\eqref{eq:Toperator}. Consequently,
from Eq.~\eqref{averaged_M_2},
the inner product~\eqref{eq:AT_fT}
can be determined {\sl only by the `initial' data},
as stated in the following theorem.

\begin{theorem}
The average value $\langle \hat{A} \rangle_T$ of a macroscopic quantity $A(\bm{x}, \bm{k})$ at a given final time $T$
can be evaluated by
\begin{equation}\label{eq:important}
\langle \hat{A} \rangle_{T}=  \int_{\mathbb{R}^d} \D \bm{r} \int_{\mathcal{K}}\D \bm{k} ~f(\bm{r}, \bm{k}, 0) \varphi(\bm{r}, \bm{k}, 0),
\end{equation}
where $\varphi$ is defined in Eq.~\eqref{Def:varphi}.
\end{theorem}

According to Eq.~\eqref{eq:important}, in order to evaluate $\langle \hat{A} \rangle_{T}$,
the remaining task is to calculate $\varphi(\bm{r}_0, \bm{k}_0, 0)$.
To this end, we need first to obtain the expression of $g(\bm{r}_0(t_0), \bm{k}_0, t_0)$ from the dual system \eqref{Def:dual_system}.

Replacing $(\bm{x}^{\prime},\bm{k}^{\prime},t^\prime)$ by $(\bm{r}_0(t_0),\bm{k}_0, t_0)$ and performing the coordinate conversion $\bm{x}(t-t^{\prime}) \to   \bm{r}_1$, $\bm{k} \to \bm{k}_1$, $t \to t_1$
in Eq.~\eqref{Def:dual_system} yield
\begin{equation}\label{eq:g(0)}
\begin{split}
g(\bm{r}_0(t_0), \bm{k}_0, t_0)=&g_0(\bm{r}_0(t_0), \bm{k}_0, t_0)+\int_{t_0}^{T} \D t_1 \int_{\mathbb{R}^d} \D \bm{r}_1 \int_{\mathcal{K}} \D \bm{k}_1 ~\me^{-\int_{t_0}^{t_1} \gamma(\bm{r}_1(s-t_0)) \D s} \\
& \times \Gamma(\bm{r}_1, \bm{k}_1; \bm{r}_0(t_0), \bm{k}_0) g(\bm{r}_1(t_1-t_0), \bm{k}_1, t_1),
\end{split}
\end{equation}
where the trajectory $\bm{r}_{1}(\Delta t)$ reads
\begin{equation}
\bm{r}_1(\Delta t)=\bm{r}_1+\hbar \bm{k}_1\Delta t/m, \quad
\Delta t\geq 0,
\end{equation}
which is \emph{not} the same as $\bm{r}_0(\Delta t)$ given in
Eq.~\eqref{Def:forward-in-time trajectory}
since the underlying wavevectors $\bm{k}_0$ and $\bm{k}_1$ are different!
Substituting Eq.~\eqref{eq:g(0)} into Eq.~\eqref{averaged_M_2} leads to
\begin{equation}\label{particle_model}
\begin{split}
\langle &\hat{A} \rangle_{T} = \langle \hat{A} \rangle_{T, 0}+\int_0^T \D t_0 \int_{\mathbb{R}^d} \D \bm{r}_0  \int_{\mathcal{K}}\D \bm{k}_0  ~f(\bm{r}_0, \bm{k}_0, 0)\me^{-\int_{0}^{t_0} \gamma(\bm{r}_0(s)) \D s}\int_{t_0}^{T} \D t_1    \\
&\times \int_{\mathbb{R}^d} \D \bm{r}_1  \int_{\mathcal{K}}\D \bm{k}_1 \me^{-\int^{t_1}_{t_0} \gamma(\bm{r}_1(s-t_0)) \D s} {\Gamma(\bm{r}_1, \bm{k}_1; \bm{r}_0(t_0), \bm{k}_0)} g(\bm{r}_1(t_1-t_0), \bm{k}_1, t_1),
\end{split}
\end{equation}
where
\begin{equation}
\langle  \hat{A} \rangle_{T, 0}=\int_{\mathbb{R}^d} \D \bm{r}_0  \int_{\mathcal{K}}\D \bm{k}_0 ~f(\bm{r}_0, \bm{k}_0, 0)\me^{-\int_{0}^{T} \gamma(\bm{r}_0(s)) \D s}A(\bm{r}_0(T), \bm{k}_0).
\end{equation}
From the dual system \eqref{Def:dual_system},
we can also obtain a similar expression to Eq.~\eqref{eq:g(0)} for $ g(\bm{r}_1(t_1-t_0), \bm{k}_1, t_1)$,
and then corresponding time integration with respect to $t_1$ in Eq.~\eqref{particle_model} becomes
\begin{align}
&\int_{t_0}^{T} \D t_1 ~ \me^{-\int^{t_1}_{t_0} \gamma(\bm{r}_1(s-t_0)) \D s} g(\bm{r}_1(t_1-t_0), \bm{k}_1, t_1)=\me^{-\int^{T}_{t_0} \gamma(\bm{r}_1(s-t_0)) \D s} A(\bm{r}_1(T-t_0), \bm{k}_1)\nonumber\\
&+\int_{t_0}^{T} \D t_1~\me^{-\int^{t_{1}}_{t_0} \gamma(\bm{r}_1(s-t_0)) \D s}  \int_{t_1}^{T} \D t_2  \int_{\mathbb{R}^d} \D \bm{r}_2 \int_{\mathcal{K}} \D \bm{k}_2 ~  g(\bm{r}_2(t_2-t_1), \bm{k}_2, t_2) \label{second_expansion_formula}
\\
&\times \me^{-\int^{t_2}_{t_1} \gamma(\bm{r}_2(s-t_1)) \D s} \Gamma(\bm{r}_2, \bm{k}_2; \bm{r}_1(t_1-t_0), \bm{k}_1),\nonumber
\end{align}
where
\begin{equation}
\bm{r}_2(\Delta t)=\bm{r}_2+{\hbar \bm{k}_2\Delta t}/{m},
\quad \Delta t\geq 0.
\end{equation}

Combining Eq.~\eqref{second_expansion_formula} and the definition~\eqref{Def:varphi} directly gives the adjoint equation for $\varphi$ as stated in Theorem~\ref{th:adjoint}.

\begin{theorem}[Adjoint equation]
\label{th:adjoint}
The function $\varphi(\bm{r}, \bm{k}, t)$ defined in Eq.~\eqref{Def:varphi}  satisfies the following integral equation
\begin{equation}\label{Adjoint_renewal_equation}
\begin{split}
\varphi(\bm{r}, \bm{k}, t)=&\me^{-\int^{T}_{t} \gamma(\bm{r}(s-t)) \D s} A(\bm{r}(T-t), \bm{k})+\int_{t}^{T} \D t^{\prime} \int_{\mathbb{R}^d} \D \bm{r}^{\prime} \int_{\mathcal{K}} \D \bm{k}^{\prime}\\
& \times \varphi(\bm{r}^{\prime}, \bm{k}^{\prime}, t^{\prime}) \me^{-\int^{t^{\prime}}_{t} \gamma(\bm{r}(s-t)) \D s} \Gamma(\bm{r}^{\prime}, \bm{k}^{\prime}; \bm{r}(t^{\prime}-t), \bm{k}).
\end{split}
\end{equation}
\end{theorem}

We call Eq.~\eqref{Adjoint_renewal_equation} the adjoint equation of
Eq.~\eqref{Wigner_integral_form} is mainly because $f(\bm{r},
\bm{k}, 0)$ and $\varphi(\bm{r}, \bm{k}, 0)$ constitute a dual
system in the bilinear form~\eqref{eq:important} (denoted by
$\langle \cdot, \cdot \rangle_0$) for determining $\langle \hat{A}
\rangle_{T}$. Combining the formal solution $f(\bm{r}, \bm{k}, T)=
\me^{T(\mathcal{A+B})}f(\bm{r}, \bm{k}, 0)$ of the Wigner
equation~\eqref{eq.k_truncated_Wigner} as well as
Eqs.~\eqref{average_A_wf} and \eqref{eq:important} directly yields
\begin{align}
\langle \hat{A} \rangle_{T} &=
\langle f(\bm{r}, \bm{k}, 0), \varphi(\bm{r}, \bm{k}, 0) \rangle_0
=\langle f(\bm{r}, \bm{k}, T), A(\bm{r}, \bm{k})\rangle_0 \nonumber \\
&=\langle\me^{T(\mathcal{A+B})}f(\bm{r}, \bm{k}, 0), A(\bm{r}, \bm{k})\rangle_0=\langle f(\bm{r}, \bm{k}, 0), \me^{-T(\mathcal{A+B})}A(\bm{r}, \bm{k})\rangle_0.
\end{align}
In consequence, we formally obtain $\varphi(\bm{r}, \bm{k}, 0)=\me^{-T(\mathcal{A+B})}A(\bm{r}, \bm{k})$ with $\me^{-T(\mathcal{A+B})}$ being the adjoint operator of $\me^{T(\mathcal{A+B})}$, indicating that Eq.~\eqref{Adjoint_renewal_equation}, in some sense, can be treated as an inverse problem of Eq.~\eqref{Wigner_Fredholm_integral_form}, which produces a quantity $\varphi(\bm{r}, \bm{k}, 0)$ from the observation $A(\bm{r}, \bm{k})$ at the ending time $T$.

Moreover, for given $(\bm{r}, t)$ on $t^{\prime} \geq t$,  we can similarly
introduce a probability measure with respect to $t^{\prime}$ like
\begin{equation}\label{eq:measure2}
\mathcal{G}(t^{\prime}; \bm{r}, t)=\int_{t}^{t^{\prime}} \gamma(\bm{r}(\tau-t)) \me^{-\int^{\tau}_{t} \gamma(\bm{r}(s-t)) \D s} ~ \D \tau,
\end{equation}
because of
\begin{equation}
\dif \mathcal{G}(t^{\prime} ; \bm{r}, t ) \geq 0,
\quad \int_{t}^{+\infty} \dif \mathcal{G}(t^{\prime} ; \bm{r}, t )
= 1,
\end{equation}
under the assumption that
the auxiliary function satisfies
\begin{equation}\label{eq:gamma_condition_1}
\forall \bm{r}\in\mathbb{R}^d, \quad
\gamma(\bm{r})\geq 0, \quad \lim_{t^\prime\to+\infty}\int_{t}^{t^\prime} \gamma(\bm{r}(t-s))\dif s = +\infty.
\end{equation}
Substituting the measure~\eqref{eq:measure2} into Eq.~\eqref{Adjoint_renewal_equation} also yields
a renewal-type equation
\begin{equation}\label{Adjoint_renewal_type_equation}
\begin{split}
\varphi(\bm{r}, \bm{k}, t)=&\left[1-\mathcal{G}(T; \bm{r}, t)\right] A(\bm{r}(T-t), \bm{k})\\
&+\int_{t}^{T} \D \mathcal{G}(t^{\prime}; \bm{r}, t)  \int_{\mathbb{R}^d} \D \bm{r}^{\prime} \int_{\mathcal{K}} \D \bm{k}^{\prime}  \frac{\Gamma(\bm{r}^{\prime}, \bm{k}^{\prime}; \bm{r}(t^{\prime}-t), \bm{k})}{\gamma(\bm{r}(t^{\prime}-t))} \varphi(\bm{r}^{\prime}, \bm{k}^{\prime}, t^{\prime}).
\end{split}
\end{equation}

\subsection{Importance sampling}
\label{sec:importance_sample}

%Before getting into the details of {\cf WBRW},
%we would like to emphasize the central role of Eq.~\eqref{eq:important} in the {\cf stochastic computation}. In fact, combining the probabilistic approach with the principle of importance sampling,
%the inner product~\eqref{eq:AT_fT} can be evaluated by a  purely particle-based scheme in which every super-particle, carrying a weight and a sign, is moving according to several specific rules in the branching particle system. More importantly, such {\cf a setting allows us to solve the Wigner quantum dynamics in a time-marching manner and resample from the branching particle system, thereby suppressing} the exponential growth of particle number in the branching particle system and successfully saving the efficiency.

Before launching into the details of probabilistic interpretation, we would like to emphasize the central role of Eq.~\eqref{eq:important} in computation. Hereto we have shown in Eq.~\eqref{eq:important} that the average $\langle \hat{A} \rangle_{T}$ can be evaluated by sampling from the initial Wigner function $f(\bm{r}, \bm{k}, 0)$ and solving the adjoint equation \eqref{Adjoint_renewal_equation} for $\varphi(\bm{r}, \bm{k}, 0)$, instead of the direct calculation based on $f(\bm{r}, \bm{k}, T)$ as shown in Eq.~\eqref{eq:AT_fT}. Actually, the bilinear form~\eqref{eq:important} for determining
$\langle \hat{A} \rangle_{T}$ serves as the foundation of WBRW in which the importance sampling plays a key role.

Regarding of the fact that the Wigner function may take negative value, we have to introduce an {\sl instrumental probability distribution}  $f_I$ as follows
\begin{equation}\label{instrumental_distribution}
f_I(\bm{r}, \bm{k}, t)=\frac{1}{H(t)} \Big| f(\bm{r}, \bm{k}, t) \Big|,
\end{equation}
where $H(t)$ is the normalizing factor (we assume $f \in L^1(\mathbb{R}^d \times \mathcal{K})$)
\begin{equation}\label{eq:pq}
H(t)= \iint_{\mathbb{R}^d\times\mathcal{K}} \Big| f(\bm{r}, \bm{k}, t) \Big| \D \bm{x} \D \bm{k}.
\end{equation}

Now the inner product problem can be evaluated through the importance sampling. Owing to the Markovian property of the linear evolution system, it suffices to divide the time interval  $[0, T]$ into $n$ steps, denoted by $t_{l}$ with $l=0, 1, \cdots, n$, and set $f(\bm{x}, \bm{k}, t_l)$ as the initial condition. Then we have
\begin{align}
\langle A \rangle_{t_{l+1}}&=\iint_{\mathbb{R}^d \times \mathcal{K}} \varphi(\bm{r}, \bm{k}, t_{l}) \cdot \frac{f(\bm{r}, \bm{k}, t_l)}{f_I(\bm{r}, \bm{k}, t_l)}  \cdot  f_I(\bm{r}, \bm{k},t_l) \D \bm{r} \D \bm{k} \nonumber \\
&\approx {\sum_{\alpha} \varphi(\bm{r}_\alpha, \bm{k}_\alpha, t_l)  \cdot w_\alpha(t_l)},\label{eq:A_approx2}
\end{align}
where $N_\alpha$ discrete samples $\{(\bm{r}_\alpha, \bm{k}_\alpha)\}_{\alpha=1}^{N_\alpha}$ generated from the instrumental probability distribution $f_I(\bm{r}, \bm{k}, t_l)$,
and the `weight' $w_\alpha(t)$ reads
\begin{equation}
w_\alpha(t) = \frac{s_\alpha(t)}{\sum_{\alpha} s_\alpha(t)}
\end{equation}
with
\begin{equation}\label{eq:s}
s_\alpha(t) = \frac{f(\bm{r}_\alpha, \bm{k}_\alpha, t)}{f_I(\bm{r}_\alpha, \bm{k}_\alpha, t)H(t)}
\end{equation}
being either $-1$ or $1$ due to Eq.~\eqref{instrumental_distribution}.
In fact, the estimator adopted in Eq.~\eqref{eq:A_approx2} implicitly utilizes
the strong law of large number
\begin{equation}\label{eq:strong}
\frac{1}{N_\alpha}\sum_{\alpha} H(t) s_\alpha(t) \to 1
\,\, \text{as} \,\, N_\alpha\to+\infty, ~~ a.s.
\end{equation}
It must be emphasized here that
the sign function $s_\alpha(t)$ indicates that every super-particle must be endowed with a sign, either positive or negative, for resolving the possible negative part of the Wigner function. Actually, the concept of signed particles, which emerges naturally here via the importance sampling,
is the intrinsic feature of spWMC, that is never seen in classical Vlasov or Boltzmann simulations.

In particular, the Wigner probability on arbitrary domain $D$ can be estimated by
\begin{equation}\label{eq:ft_approx}
W_D(t_{l+1}) = \iint_{D} f(\bm{r}, \bm{k}, t_{l+1})\approx {\sum_{\alpha} \varphi(\bm{r}_\alpha, \bm{k}_\alpha, t_l) \cdot w_\alpha(t_l)},
\end{equation}
with $\varphi(\bm{x}, \bm{k}, t_{l+1}) = \mone_{D}(\bm{x}, \bm{k})$. This motivates us to estimate the Wigner function at $t_{l+1}$ by the
piecewise constant function
\begin{align}
&f(\bm{x}, \bm{k}, t_{l+1}) \approx \sum_{j=1}^{J}  d_{j}(t_{l+1}) \cdot \mone_{D_j} (\bm{x}, \bm{k}),\\
&f_I(\bm{x}, \bm{k}, t_{l+1}) \approx \frac{1}{H(t)}\sum_{j=1}^{J} |d_j(t_{l+1})| \cdot \mone_{D_j} (\bm{x}, \bm{k}),
\label{eq:histogram}
\end{align}
where $D_{j}, j=1,\cdots J$ gives a partition of $\mathbb{R}^d \times \mathcal{K}$, and $|D_j|$ denotes the volume of $D_j$ and $d_{j}(t) = W_{D_j}(t)/|D_j|$.  In fact, Eq.~\eqref{eq:histogram} is the basis of the resampling procedure, which will be discussed in  Section 5.3. Combining Eqs.~\eqref{eq:A_approx2} and \eqref{eq:histogram} allows us to solve the Wigner equation through a time-marching scheme.

The remaining problem is how to estimate $\varphi(\bm{r}_\alpha, \bm{k}_\alpha, t_l)$, which can be replaced by an additive functional of $A(\bm{x}, \bm{k})$,
as shown in the next section. Hence $\langle \hat{A} \rangle$ is evaluated by a purely particle-based scheme in which every super-particle, carrying a weight and a sign, is moving according to several specific rules in the branching particle system.

%Most of the deductions can be straightforwardly generalized to the Wigner equation due to the strong similarities between Eqs.~\eqref{Def:renewal_type_eq} and \eqref{Adjoint_renewal_type_equation}.

%More importantly, such a time marching scheme overcomes the
%intrinsic weakness of the branching process treatment, say, the
%exponential growth of particle number (see Theorem~\ref{th:exp}).
%Within each step, we only need to evolve the particle system for a
%relatively small $\Delta t$ and update the instrumental
%distribution, from which we are able to draw new samples for the
%next loop. By exploiting the cancelation of weights of opposite
%sign, the resampling procedure can significantly reduce the particle
%number and save the computational efficiency dramatically in
%practice (See Section~\ref{sec:num_res}). Just because of such
%sampling reduction, the resampling procedure is also called
%\emph{annihilation} in previous works, e.g.,
%see\cite{SellierNedjalkovDimov2015}.

%The rest of this section is devoted to a computable, albeit not the unique, stochastic method for solving the adjoint equation based on the stochastic interpretation of its solution. In the subsequent sections, we will focus on the probabilistic approach for the adjoint equation~\eqref{Adjoint_renewal_equation}, while all the deductions can be straightforwardly generalized to the Wigner equation due to the strong similarities between
%Eqs.~\eqref{Def:renewal_type_eq} and \eqref{Adjoint_renewal_type_equation}.

\section{{The Wigner branching random walk}}
\label{sec:branching}

This section is devoted to the probabilistic interpretation of the adjoint equation \eqref{Adjoint_renewal_type_equation}. Owing to the fact that all the moments of a branching process satisfy the renewal-type equations\cite{bk:Harris1963}, it motivates us to construct a  stochastic branching random walk such that its expectation is equal to the unique solution of Eq.~\eqref{Adjoint_renewal_type_equation}.
In addition to validating WBRW, we also analyze the inherent  mass conservation property and derive the exact growth rate of particle number. Most of  deductions can be straightforwardly generalized to the ($k$-truncated) Wigner equation due to the strong similarities between Eqs.~\eqref{Def:renewal_type_eq} and \eqref{Adjoint_renewal_type_equation}.

After the theoretical analysis, we turn to some numerical aspects and discuss the idea of resampling, which is aimed at suppressing the exponential growth of particle number. It is closely linked to the non-parameter density estimation and presents several challenges in high dimensional cases, as also found in the statistical learning and classification.
Finally, we will outline the procedures of the statistical algorithm.

\subsection{A branching particle system }
\label{sec:branching:model}
To illustrate the main theorem more clearly, we first introduce a probabilistic model, a branching particle system associated with an exit system, to describe the WBRW in a picturesque language.  The exit system means that a particle in the branching system will be frozen when its life-length exceeds the final time $T$. All related rigorous analysis is left for the next subsection.

Consider a system of particles moving in $\mathbb{R}^d \times \mathcal{K} \times [t, T]$ according to the following rules. Without loss of generality, the particle, starting at  time $t$ at state $(\bm{r}, \bm{k})$,  having a random life-length $\tau$ and carrying a weight $\phi$,
is marked. The chosen initial data corresponds to those adopted in the renewal-type equation~\eqref{Adjoint_renewal_type_equation}.

\begin{description}
\item[Rule 1] The motion of each particle is described by a {right continuous Markov process}.

\item[Rule 2]  The particle at $(\bm{r}, \bm{k})$ dies in the age time interval $(t, t^{\prime}) $  with probability $\mathcal{G}(t^{\prime}; \bm{r}, t)$, which depends on its position $\bm{r}$ and the time $t$ (see Eq.~\eqref{eq:measure2}). In particular, when using the constant auxiliary function $\gamma(x)\equiv \gamma_0$, the particle dies during time interval $(t ,t^{\prime})$ with probability $1-\me^{-\gamma_0(t^{\prime}-t)}\approx\gamma_0 (t^{\prime}-t)$ for small $t^{\prime}-t$,
which is totally independent of both its position and age.

\item[Rule 3]  If $t+\tau < T$, the particle dies at age $t^{\prime}=t+\tau$ at state $(\bm{r}(\tau), \bm{k})$,  and produces $2M+1$ new particles at  states $(\bm{r}^{\prime}_{(1)}, \bm{k}^{\prime}_{(1)})$,  $(\bm{r}^{\prime}_{(2)}, \bm{k}^{\prime}_{(2)})$,  $\cdots$, $(\bm{r}^{\prime}_{(2M+1)}, \bm{k}^{\prime}_{(2M+1)})$,
endowed with updated weights $\phi^\prime_{(1)}$, $\phi^\prime_{(2)}$, $\cdots$, $\phi^\prime_{(2M+1)}$, respectively. All these parameters can be determined by the kernel function in Eq.~\eqref{Adjoint_renewal_type_equation}:
\begin{equation}
\begin{split}
\frac{\Gamma(\bm{r}^{\prime}, \bm{k}^{\prime}; \bm{r}(\tau), \bm{k})}{\gamma(\bm{r}(\tau))}&=\sum_{m=1}^{M}
\frac{\xi_m(\bm{r}^{\prime})}{\gamma(\bm{r}(\tau))} \cdot \frac{V^{-}_{w, m}(\bm{r}^{\prime},\bm{k}-\bm{k}^{\prime})}{\xi_m(\bm{r}^{\prime})}\cdot \delta(\bm{r}(\tau)-\bm{r}^{\prime})\\
&-\sum_{m=1}^{M} \frac{\xi_m(\bm{r}^{\prime})}{\gamma(\bm{r}(\tau))} \cdot \frac{V^{+}_{w, m}(\bm{r}^{\prime},\bm{k}-\bm{k}^{\prime})}{\xi_m(\bm{r}^{\prime})}\cdot \delta(\bm{r}(\tau)-\bm{r}^{\prime})\\
&+ 1\cdot \delta(\bm{k}-\bm{k}^{\prime})\cdot \delta(\bm{r}(\tau)-\bm{r}^{\prime}),
\end{split}
\end{equation}
and thus for $1\le m \le M$
\begin{align}
\bm{r}^{\prime}_{(1)} = \bm{r}^{\prime}_{(2)} = \cdots = \bm{r}^{\prime}_{(2M+1)} = \bm{r}(\tau),\\
\bm{k}-\bm{k}^{\prime}_{(2m-1)} \propto \frac{V^{-}_{w,m}(\bm{r}(\tau), \bm{k})}{\xi_m(\bm{r}(\tau)) },\quad
\bm{k}-\bm{k}^{\prime}_{(2m)} \propto \frac{V^{+}_{w, m}(\bm{r}(\tau), \bm{k})}{\xi_m(\bm{r}(\tau)) }, \label{eq:ksample1} \\
\quad \bm{k}^{\prime}_{(2M+1)}=\bm{k}, \label{eq:ksample}\\
\phi^\prime_{(2m-1)}={\zeta_{2m-1}(\bm{r}(\tau))}\cdot \mone_{ \{ \bm{k}^\prime_{2m-1} \in \mathcal{K} \} } \cdot \phi, \;\;
\phi^\prime_{(2m)}={\zeta_{2m}(\bm{r}(\tau))}\cdot \mone_{ \{ \bm{k}^\prime_{2m} \in \mathcal{K} \} } \cdot \phi, \\
\quad \phi^\prime_{(2M+1)}= {\zeta_{2M+1}(\bm{r}(\tau))}  \cdot \phi = 1 \cdot \phi,
\end{align}
where the function $\xi_{m}(\bm{r})$ is the normalizing factor for both
$V^{+}_{w,m}$ and  $V^{-}_{w,m}$, i.e.,
\begin{equation}\label{eq:normalizing}
\xi_{m}(\bm{r})=\int_{2\mathcal{K}} V^{+}_{w, m}(\bm{r}, \bm{k}) \D \bm{k} =  \int_{2\mathcal{K}} V^{-}_{w, m}(\bm{r}, \bm{k}) \D \bm{k},
\end{equation}
because of the mass conservation~\eqref{eq:mass1},
and
\begin{equation}\label{eq:zeta}
\zeta_{2m-1}(\bm{r}) = \frac{\xi_m(\bm{r})}{\gamma(\bm{r})}, \quad \zeta_{2m}(\bm{r}) =  -\frac{\xi_m(\bm{r})}{\gamma(\bm{r})}, \quad \zeta_{2M+1}(\bm{r})=1.
\end{equation}
%Actually, besides Eq.~\eqref{eq:ksample},
%it is required that both
%$\bm{k}^{\prime}_{(2m-1)} \in \mathcal{K}$ and $ \bm{k}^{\prime}_{(2m)} \in \mathcal{K}$ hold.

\item[Rule 4] If $t+\tau \ge T$, say, the life-length of the particle exceeds $T-t$, so it will immigrate to the state $(\bm{r}(T-t), \bm{k})$ and be frozen. This rule corresponds to the first right-hand-side term of Eq.~\eqref{Adjoint_renewal_type_equation},
and the related probability is
\begin{equation}
\Pr(\tau \ge T-t)=1-\mathcal{G}(T; \bm{r}, t)=\me^{-\int_{t}^{T} \gamma(\bm{r}(s-t)) \D s}.
\end{equation}

\item[Rule 5] The only interaction between the particles is that the birth time and state of offsprings coincide with the death time and state of their parent.
\end{description}

\begin{remark}
In Eqs.~\eqref{eq:ksample1} and \eqref{eq:normalizing}, we require that $V_{w}^{+}$ and $V_{w}^{-}$ can be normalized, which is not necessarily true for $|\mathcal{K}| \to \infty$. Nevertheless, when the convolution operator $\mathcal{B}$ is bounded and $f(\bm{x}, \bm{k}, t)$ decays in $\mathcal{K}$-space sufficiently rapidly, we can approximate the solution of the Wigner equation by  the $k$-truncated correspondence because of the boundedness of the semigroup $\me^{t\mathcal{A}}$ and
\begin{equation}
\Vert \int_{\mathbb{R}^d \setminus \mathcal{K}} V_{w}(\bm{x}, \bm{k}-\bm{k}^{\prime}, t) f(\bm{x}, \bm{k}^{\prime}, t) \D \bm{k}^{\prime}\Vert \to 0, ~~\textup{as}~~|\mathcal{K}| \to \infty.
\end{equation}
Thus it suffices to restrict our discussion on the $k$-truncated Wigner equation and its adjoint counterpart.
\end{remark}

Now we present the main result.  Let $\mathcal{E}_\alpha$ be the index set of all frozen particles with
the same ancestor initially at time $t=0$ at state $(\bm{r}_{\alpha}, \bm{k}_{\alpha})$ carrying the weight $\phi_0=1$,
$\{ (\bm{r}_{i, \alpha}, \bm{k}_{i, \alpha}), i\in \mathcal{E}_{\alpha}\}$ denote the collection of corresponding frozen states,
and $\phi_{i, \alpha}$ the updated weight of the $i$-th particle. Accordingly, the adjoint equation is solved by
\begin{equation}\label{eq:adjoint_probabilisitic}
\varphi(\bm{r}_\alpha, \bm{k}_\alpha, 0)
= \Pi_{0, \bm{r}_\alpha, \bm{k}_\alpha} \left( \sum_{i\in \mathcal{E}_\alpha} \phi_{i,\alpha} \cdot  A(\bm{r}_{i, \alpha}, \bm{k}_{i, \alpha}) \right),
\end{equation}
where $\Pi_{0, \bm{r}_\alpha, \bm{k}_\alpha}(\cdot)$ means the expectation with respect to the probability law defined by the above five rules (its definition is left in Section \ref{sec:branching:analysis}). Furthermore, from Eq.~\eqref{eq:A_approx2},
the quantity $\langle \hat{A} \rangle_T$ is solved by
\begin{equation}\label{eq:inner_product_probabilisitic}
\langle \hat{A} \rangle_{T} =  \mathbb{E}_{f_I}\left[ \Pi_{0, \bm{r}_\alpha, \bm{k}_\alpha} \left( s_\alpha(0) \cdot H(0) \cdot \sum_{i\in \mathcal{E}_\alpha} \phi_{i,\alpha} \cdot  A(\bm{r}_{i, \alpha}, \bm{k}_{i, \alpha})
 \right)\right],
\end{equation}
where $\mathbb{E}_{f_I}$ means the expectation with respect to $f_I$. The proof of  Eqs.~\eqref{eq:adjoint_probabilisitic} and \eqref{eq:inner_product_probabilisitic}
is left for Section \ref{sec:branching:analysis}.

\subsection{Stochastic interpretation}
\label{sec:branching:analysis}

In the theory of branching process, all the moments of an age-dependent branching processes satisfy renewal-type integral equations\cite{bk:Harris1963}, where the term ``age-dependent'' means the probability that a particle, living at $t$, dies at $(t, t+\Delta t)$ might not be a constant function of $t$.
In this regard, it suffices to define a stochastic branching Markov process (continuous in time parameter), corresponding to the branching particle system as described earlier.

The random variable of a branching particle system is the family history, a denumerable random sequence corresponding to a unique family tree. Firstly, we need a sequence to identify the objects in a family. Beginning with an ancestor, denoted by $\langle 0 \rangle$, and we can denote its $m$-th children by  $\langle m \rangle$. Similarly, we can denote the $j$-th child of $i$-the child by $\langle i j \rangle$, and thus $ \langle  i_1 i_2\cdots i_n \rangle$ means $i_n$-th child of  $i_{n-1}$-th child of $\cdots$ of the $i_2$-child of the $i_1$-th child, with $i_n \in \left\{1, 2, \cdots, 2M+1\right\}$. The ancestor $\langle 0 \rangle$ is omitted here and hereafter for brevity.

Our branching particle system involves three basic elements: the position $\bm{r}$ (or $\bm{x}$), the wavevector $\bm{k}$ and the life-length $\tau$, and each particle will either immigrate to $\bm{r}(\tau)= \bm{r}+ \hbar \bm{k} \tau/m$, then be killed and produce three offsprings, or  be frozen when hitting the first exit time $T$. Now we can give the definition of a family history, starting from one particle at age $t$ at state $(\bm{r}, \bm{k})$. In the subsequent discussion we let $(\bm{r}_0, \bm{k}_0)=(\bm{r}, \bm{k})$.

\begin{definition}
A family history $\omega$ stands for a random sequence
\begin{equation}\label{Def:family_history}
\omega=((\tau_0, \bm{r}_0, \bm{k}_0); (\tau_{1}, \bm{r}_{1}, \bm{k}_{1});  (\tau_{2}, \bm{r}_{2}, \bm{k}_{2}); (\tau_{3} ,\bm{r}_{3}, \bm{k}_{3}); (\tau_{11}, \bm{r}_{11}, \bm{k}_{11}); \cdots),
\end{equation}
where the tuple $Q_{i}=(\tau_{i}, \bm{r}_{i}, \bm{k}_{i})$ appears in a definite order of enumeration. $\tau_{i}$, $\bm{r}_i$, $\bm{k}_i$ denote the life-length, starting position and wavevector of the $i$-th particle, respectively. The exact order of $Q_{i}$ is immaterial but is supposed to be fixed. The collection of all family histories is denoted by $\Omega$.
\end{definition}

At this stage, the initial time $t$ and the initial state $(\bm{r}_0, \bm{k}_0)$ of the ancestor particle $\langle 0 \rangle$ are assumed to be non-stochastic.

\begin{definition}
For each $\omega=(Q_{0}; Q_{1}; Q_{2}; Q_{3}; Q_{11}\cdots)$, the subfamily $\omega_{i}$ is the family history of $\langle i \rangle$ and its descendants, defined by $\omega_{i}=(Q_{i}; Q_{i1}; Q_{i2}, Q_{i3}; \cdots)$.  The collection of $\omega_i$ is denoted by $\Omega_i$.
\end{definition}

Equivalently, we can also use the time parameter to identity the path of a particle and all its ancestors in the family history, and denote $\eta_s$ its state (i.e., starting position and wavevector) at time $s \ge t$. Taking the particle $i=\langle  i_1 i_2\cdots \rangle$ as an example, we have
$Q_0 = (\tau_0, \eta_t)$ with $\eta_t=(\bm{r}_0, \bm{k}_0)$,
$Q_{i_1}=(\tau_{i_1}, \eta_{t_{i_1}})$ with $t_{i_1}=t+\tau_0$
and $\eta_{t_{i_1}}=(\bm{r}_{i_1},\bm{k}_{i_1})$,  $Q_{i_1 i_2}=(\tau_{i_1 i_2}, \eta_{t_{i_1 i_2}})$ with $t_{i_1 i_2}=t_{i_1}+\tau_{i_1}$ and $\eta_{t_{i_1 i_2}}=(\bm{r}_{i_1 i_2},\bm{k}_{i_1 i_2})$, $\cdots$.
To characterize the freezing behavior of the particle, we denote the first exit time by $T$, as boundary conditions are not specified.

\begin{definition}\label{def:frozen}
Suppose the family history starts at time $t$. Then a particle $\langle  i_{1}i_{2}\cdots i_{n}\rangle$ is said to be frozen at $T$ if the following conditions hold
\begin{align}
t+\tau_0+\tau_{i_1}+\tau_{i_1 i_2}+\cdots+\tau_{i_1 i_2 \cdots i_{n-1}} &< T,\\
t+\tau_0+\tau_{i_1}+\tau_{i_1 i_2}+\cdots+\tau_{i_1 i_2 \cdots i_{n-1}}+\tau_{i_1 i_2 \cdots i_{n}} &\ge T.
\end{align}
In particular, when $t+\tau_0\ge T$, the ancestor particle $\langle 0 \rangle$ is frozen. Sometimes the particle $\langle i_{1}i_{2}\cdots i_{n}\rangle$ is also called alive in the time interval $[t, T]$. The collection of frozen particles is denoted by $\mathcal{E}(\omega)$.
\end{definition}

\begin{remark}
The first exit time $\tau_e(O)$ from an open set $O$ is defined by
\begin{equation*}
\tau_e(O)=\inf \left\{ s \ge t : (s, \eta_s) \notin O\right\}.
\end{equation*}
Since a boundary condition is not applied yet, we have $O=(-\infty, T) \times \mathbb{R}^d \times \mathcal{K^{\prime}}$
with $\mathcal{K^{\prime}}$ being an open cover of $\mathcal{K}$,
and thus $\tau_e(O)=T$.
\end{remark}

\begin{figure}[h]
\centering
\subfigure{\includegraphics[width=0.6\textwidth,height=0.4\textwidth]{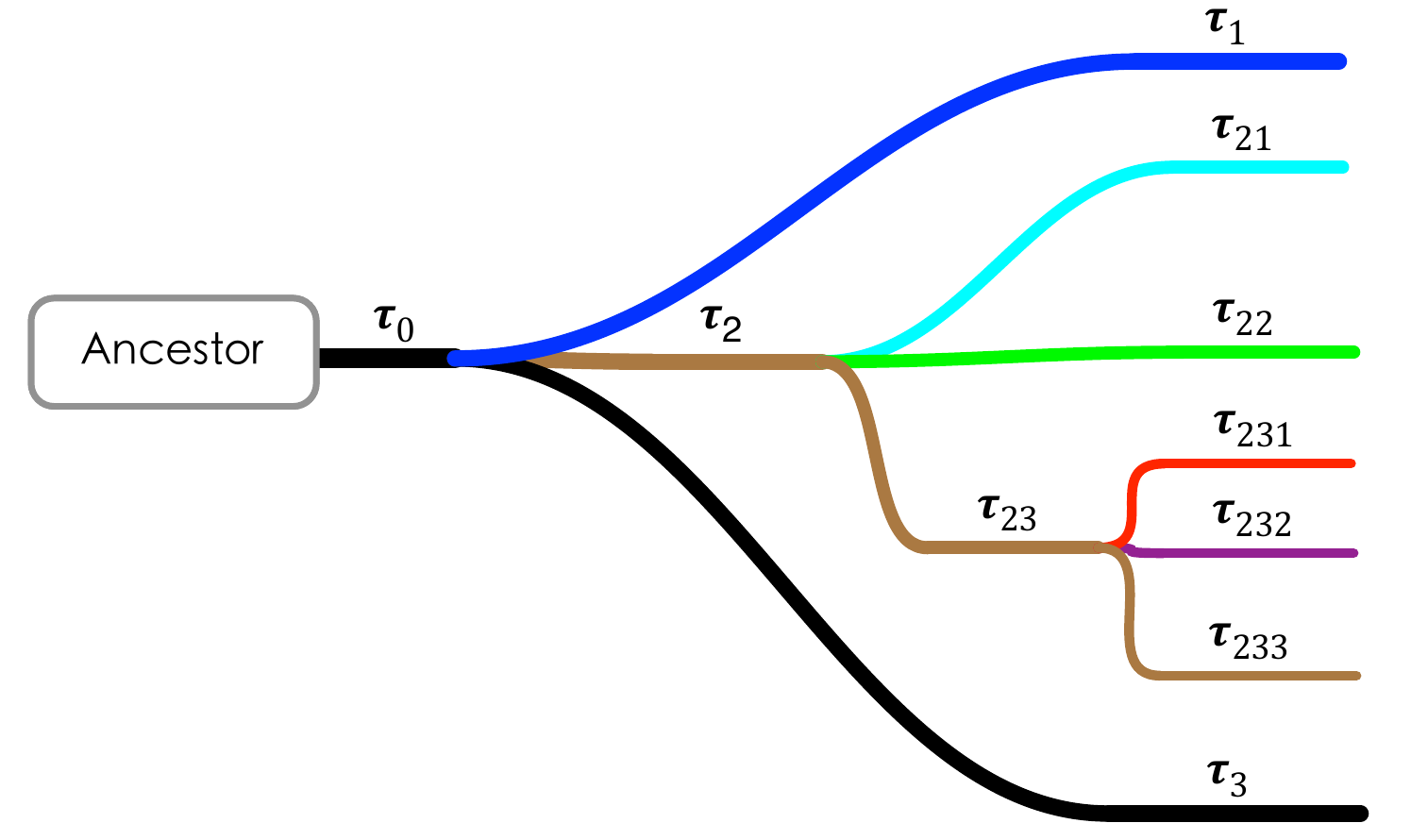}}
\caption{\small An example of family history tree.}
\label{fig:branching_tree}
\end{figure}

\begin{example}\rm
$\omega=(Q_{0}; Q_{1};  Q_{2}; Q_{3}; Q_{21}; Q_{22}; Q_{23};   Q_{231}; Q_{232}; Q_{233})$ uniquely determines a family history tree, as shown in Fig.~\ref{fig:branching_tree}. $\omega_2=(Q_{2};  Q_{21}; Q_{22}; Q_{23}; Q_{231}; Q_{232}; Q_{233})$ is a subfamily history that describes the family history of $Q_{2}$ and its descendants, and we have $\omega=(Q_{0};\omega_{1}; \omega_{2}; \omega_{3})$. The collection of frozen particles is $\mathcal{E}(\omega)=\{\langle 1 \rangle, \langle 21 \rangle, \langle 22 \rangle, \langle 231 \rangle, \langle 232 \rangle, \langle 233 \rangle, \langle 3 \rangle\}$.
\end{example}

Hereafter we assume that all particles in the branching particle system will move until reaching the frozen state, and still use $\Omega$ to denote the collection of the family history of all frozen particles. Now we need to define a probability measure $\Pi_{t, \bm{r}, \bm{k}}$ on $\Omega$,
corresponding to the branching process started from state $(\bm{r}, \bm{k})$ at time $t$.

For the Borel sets $T_i \subset [0,+\infty)~(i=0,1,\cdots,n)$ , $R_i \subset \mathbb{R}^d, K_i \subset \mathcal{K} ~(i=1,2,\cdots, n)$ on $\Omega$,  let $E =\{ \tau_0 \in T_0, (\tau_{i_1},\eta_{t_{i_1}}) \in T_1 \times R_1 \times K_1,  \cdots,
(\tau_{i_1 i_2 \cdots i_{n}}, \eta_{t_{i_1 i_2\cdots i_{n}}})
\in T_{n}\times R_{n}\times K_{n} \}$, then the probability of the event $E$ is
\begin{align}
\Pr(E)&=\int_{T_0} \D \tau_0 \cdots  \int_{T_{n}} \D \tau_{i_1 i_2\cdots i_{n}}  \int_{R_1} \D \bm{r}_{i_1} \int_{K_1} \D \bm{k}_{i_1} \cdots \int_{R_n} \D \bm{r}_{i_1 i_2\cdots i_n} \int_{K_n} \D\bm{k}_{i_1 i_2\cdots i_n} \nonumber\\
&\times p_{i_1}(t, \bm{r}_0, \bm{k}_0; t_{i_1}, \bm{r}_{i_1}, \bm{k}_{i_1}) \times  p_{i_2}(t_{i_1}, \bm{r}_{i_1}, \bm{k}_{i_1}; t_{i_1 i_2}, \bm{r}_{i_1 i_2}, \bm{k}_{i_1 i_2}) \times \cdots \nonumber \\
&\times p_{i_{n}}(t_{i_1 i_2\cdots i_{n-1}}, \bm{r}_{i_1 i_2\cdots i_{n-1}}, \bm{k}_{i_1 i_2\cdots i_{n-1}}; t_{i_1 i_2\cdots i_{n}}, \bm{r}_{i_1 i_2\cdots i_{n}}, \bm{k}_{i_1 i_2\cdots i_{n}}) \nonumber\\
&\times p(t_{i_1 i_2 \cdots i_n}, \bm{r}_{i_1 i_2 \cdots i_n} ; t_{i_1 i_2 \cdots i_n}+\tau_{i_1 i_2 \cdots i_n})
\label{Def:probability_measure}
\end{align}
with $i_{l} \in \{1,2,\cdots, 2M+1\}$ ($l=1,2,\cdots,n$).
Here the transition densities
$p_{i_l}$ and $p$ are given by ($1\le m \le M$)
\begin{align}
p_{2m-1}(t, \bm{r}, \bm{k}; t^{\prime}, \bm{r}^{\prime}, \bm{k}^{\prime})&= p(t, \bm{r} ; t^{\prime}) \cdot \frac{V_{w,m}^{-}(\bm{r}^\prime, \bm{k}-\bm{k}^{\prime}) }{\xi_m(\bm{r}^\prime)}  \cdot \delta(\bm{r}^{\prime}- \bm{r}(\tau)), \label{eq:p1}\\
p_{2m}(t, \bm{r}, \bm{k}; t^{\prime}, \bm{r}^{\prime}, \bm{k}^{\prime})&= p(t, \bm{r} ; t^{\prime}) \cdot \frac{V_{w,m}^{+}(\bm{r}^\prime, \bm{k}-\bm{k}^{\prime}) }{\xi_m(\bm{r}^\prime)}  \cdot \delta(\bm{r}^{\prime}- \bm{r}(\tau)), \\
p_{2M+1}(t, \bm{r}, \bm{k}; t^{\prime}, \bm{r}^{\prime}, \bm{k}^{\prime})&=p(t, \bm{r} ; t^{\prime}) \cdot
\delta(\bm{k}-\bm{k}^{\prime}) \cdot \delta(\bm{r}^{\prime}- \bm{r}(\tau)), \\
p(t, \bm{r} ; t^{\prime})&=\frac{\D \mathcal{G}(t^{\prime}; \bm{r}, t)}{\D t^{\prime}}\Big |_{t^{\prime}=t+\tau},
\end{align}
where $\mathcal{G}(t^{\prime}; \bm{r}, t)$ has been defined in Eq.~\eqref{eq:measure2} and the random life length $\tau=t^\prime - t$ satisfies $\tau \propto p(t, \bm{r} ; t^{\prime})$.  $(\bm{r}, \bm{k}) \to (\bm{r}^{\prime}, \bm{k}^{\prime})$ corresponds to a random walk. Combining with the independence assumption in Rule 5, we are able to define a probability measure $\Pi_{t, \bm{r}, \bm{k}}$ on $\Omega$ as follows
\begin{equation}\label{eq:Pi}
\Pi_{t, \bm{r}, \bm{k}}(\mone_E)= \int_{\Omega} \mone_E(\omega) \Pi_{t, \bm{r}, \bm{k}}(\D \omega) =\Pr(E),
\end{equation}
as well as a stochastic branching process $(\Omega, \Pi_{t, \bm{r}, \bm{k}})$. Moreover, from Eq.~\eqref{Def:probability_measure}, we can easily verify the following Markov property of the stochastic process $(\Omega, \Pi_{t, \bm{r}, \bm{k}})$
\begin{equation}\label{Def:Markov_property}
\Pi_{t, \bm{r}, \bm{k}} (XY) =\int_{\Omega_{X}}  X \Pi_{t+\tau, \eta_{t+\tau}}(Y) \Pi_{t, \bm{r}, \bm{k}}(\D \omega)
\end{equation}
{for any function $X$ in a measurable space $(\Omega_X, \mathcal{T}_{[t, t+\tau]} \otimes \mathcal{F}_{[t, t+\tau]})$ and any function $Y$ in $(\Omega_{Y}, \mathcal{T}_{[t+\tau, {+\infty})} \otimes \mathcal{F}_{[t+\tau, T]})$}, where $\mathcal{T}_{\mathcal{I}} $ is the $\sigma$-algebra on $\mathcal{I} \subset [0,+\infty)$, $\mathcal{F}_\mathcal{I}$ is the $\sigma$-algebra generated by $\eta_s$ with $s\in \mathcal{I}$,
and $\Omega=\Omega_X\times\Omega_Y$. That is, {\sl events observable before and after time $t+\tau$ are conditionally independent for given state  $\eta_{t+\tau}$}.

\begin{remark}
Since $\omega\in \Omega$ corresponds to a denumerable random sequence, the basic theorem of Kolmogorov (see Theorem 6.16, \cite{bk:Kallenberg2002}) ensures the existence and uniqueness of such a random process $(\Omega, \mathscr{B}_\Omega, \Pi_{t, \bm{r}, \bm{k}})$ with the probability measure $\Pi_{t, \bm{r}, \bm{k}}$ defined on the Borel extension $\mathscr{B}_\Omega$ of the cylinder sets on $\Omega$\cite{bk:Harris1963}.
\end{remark}

Next we need to define a signed measure valued function
$\mu: (\mathbb{R}^d \times \mathcal{K}, \mathscr{B}) \to \mathbb{R}$ through the particle weights and the frozen states, where $\mathscr{B} \in\mathscr{B}_{\Omega}$ and $\mathscr{B}_{\Omega}$ stands for the Borel cylinder sets on $\Omega$. According to Rule 4, the frozen state of a particle $\langle i_1 i_2 \cdots i_n \rangle$ is $(\bm{r}_{i_1 i_2 \cdots i_n}(T-t_{i_1 i_2 \cdots i_n}), \bm{k}_{i_1 i_2 \cdots i_n})$ with $t_{i_1 i_2 \cdots i_n}=t+\tau_0+ \cdots +\tau_{i_1 i_2 \cdots i_{n-1}}$, and the frozen state of $\langle 0 \rangle$ is $(\bm{r}_0(T-t), \bm{k}_0)$.

\begin{definition}\label{eq:exit}
Suppose $(\bm{r}_i, \bm{k}_i)$ is the starting state of a frozen particle $i$ in a given family history $\omega$, and let $\delta_{(\bm{r}, \bm{k})}$ mean the unit measure concentrated at  state $(\bm{r}, \bm{k})$.
Then we define the exit measure as follows
\begin{equation}
\mu=\sum_{i \in \mathcal{E}(\omega)} \phi_{i} \cdot \delta_{(\bm{r}_{i}(T-t_i), \bm{k}_{i})},
\end{equation}
where $\phi_i$ is the cumulative weight of particle $i$. For an object $i=\langle i_{1}i_{2}\cdots i_{n} \rangle$, $\phi_i$ is given by
\begin{equation}\label{eq:phii}
\phi_i=\phi_0 \cdot \zeta_{i_1}(\bm{r}_{i_{1}}) \cdot  \zeta_{i_2}(\bm{r}_{i_{1}i_{2}})\cdots \zeta_{i_{n-1}}(\bm{r}_{i_{1}i_{2}\cdots i_{n-1}}) \cdot \mone_{ \{ \bm{k}_{i_1} \in \mathcal{K}, \cdots, \bm{k}_{i_1 i_2 \cdots i_{n-1}} \in \mathcal{K} \}},
\end{equation}
where $\phi_0=1$ is
the initial weight of the ancestor
and the function $\zeta(\bm{r})$ has been defined in Eq.~\eqref{eq:zeta}. Moreover, for given $\omega$ and function $A(\bm{r}, \bm{k})$, we can further define a random integral on the point distribution
\begin{equation}\label{eq:muA}
\mu_A(\omega) = \int A(\bm{r}, \bm{k}) \mu(\D \bm{r} \times \D \bm{k}, \omega)=\sum_{i \in \mathcal{E}(\omega)} \phi_i \cdot A(\bm{r}_{i}(T-t_i), \bm{k}_{i}).
\end{equation}
To ensure a bounded weight, we require $\left|\zeta_{i}(\bm{r})\right|\leq 1$ for any $\bm{r}\in\mathbb{R}^d$ and $1\le i \le 2M+1$. The first moment of random function $\mu_A(\omega)$ is denoted by $\psi(\bm{r}, \bm{k}, t)$ which reads
\begin{equation}\label{eq:psi}
\psi(\bm{r}, \bm{k}, t)=\Pi_{t, \bm{r}, \bm{k}} (\mu_A)=\int_{\Omega} \mu_A(\omega) \Pi_{t, \bm{r}, \bm{k}}(\D \omega).
\end{equation}
\end{definition}

\begin{example}
Suppose the ancestor starts at $t=0$ carrying the initial weight $\phi=1$.
For the family history $\omega$ displayed in Fig.~\ref{fig:branching_tree},
the random integral ${\mu_A}(\omega)$ is
\begin{equation}
\begin{split}
{\mu_A}(\omega)&=\zeta(\bm{r}_{1})A(\bm{r}_1(T-\tau_0), \bm{k}_1)+\zeta(\bm{r}_{2})\zeta(\bm{r}_{21})A(\bm{r}_{21}(T-\tau_0-\tau_2), \bm{k}_{21})\\
&+\zeta(\bm{r}_{2})\zeta(\bm{r}_{22})A(\bm{r}_{22}(T-\tau_0-\tau_2), \bm{k}_{22})\\
&+\zeta(\bm{r}_{2})\zeta(\bm{r}_{23})\zeta(\bm{r}_{231})A(\bm{r}_{231}(T-\tau_0-\tau_2-\tau_{23}), \bm{k}_{231})\\
&+\zeta(\bm{r}_{2})\zeta(\bm{r}_{23})\zeta(\bm{r}_{232})A(\bm{r}_{232}(T-\tau_0-\tau_2-\tau_{23}), \bm{k}_{232})\\
&+\zeta(\bm{r}_{2})\zeta(\bm{r}_{23})\zeta(\bm{r}_{233})A(\bm{r}_{233}(T-\tau_0-\tau_2-\tau_{23}), \bm{k}_{233})\\
&+\zeta(\bm{r}_{3})A(\bm{r}_{3}(T-\tau_0), \bm{k}_3).
\end{split}
\end{equation}
\end{example}

In order to study the particle number in the branching
particle system with the family history $\omega$ starting from time $t$, we use a random function $Z(\omega, T-t)$ to stand for the total number of frozen particles at the final instant $T$.
In consequence,
the first moment of $Z(\omega, T-t)$ is
\begin{equation}
\mathbb{E} Z_{T-t}=\int_{\Omega} Z(\omega, T-t)\Pi_{t, \bm{r}, \bm{k}}(\D \omega),
\end{equation}
which also gives the expectation of the total number of alive
particles in time interval $[t, T]$, and should be finite (see
Theorem~\ref{th:finite}). This further means that $Z(\omega, T-t)$
is finite almost surely. As the easier case, the finiteness of
$\mathbb{E}Z_{T-t}$ for the constant auxiliary function is directly implied
from Theorem 13.1 and its corollary of Chapter VI
in~\cite{bk:Harris1963}.

\begin{theorem}
\label{th:finite}
Suppose the family history $\omega$ starts at time $t$ at state $(\bm{r}, \bm{k})$, and ends at $T$. Then $\mathbb{E}Z_{T-t}<\infty$ and as a consequence $\Pr(\{Z(\omega, T-t)< \infty\})=1$.
\end{theorem}

\begin{proof}
. We define a random function $\mone_{i_1 i_2 \cdots i_n}(\omega)=1$ when the particle $\langle i_{1}i_{2}\cdots i_{n}\rangle$ appears in the family history $\omega$, otherwise $\mone_{i_1 i_2 \cdots i_n}(\omega)=0$. From Eq.~\eqref{eq:Pi} and Definition~\ref{def:frozen}, we have
\begin{equation}\label{eq:pi_1}
\Pi_{t, \bm{r}, \bm{k}}(1_{i_1 i_2 \cdots i_n}) = \int_{\Omega} \mone_{i_1 i_2 \cdots i_n}(\omega)  \Pi_{t, \bm{r}, \bm{k}}(\D \omega)=\Pr(\{t+\tau_0+\cdots+\tau_{i_1 i_2 \cdots i_{n-1}} < T \}).
\end{equation}

Let
\begin{equation}\label{eq:zbar0}
\bar{Z}(\omega, T-t)=1+\sum_{n=1}^{\infty} \sum_{i_1, \cdots, i_{n}=1}^{2M+1} 1_{i_1 i_2 \cdots i_n}(\omega),
\end{equation}
that corresponds to the number of particles born up to the final time $T$.
It is obvious that
\begin{equation}\label{eq:Zbar}
Z(\omega, T-t)\le \bar{Z}(\omega, T-t).
\end{equation}

For constant $\gamma_0$,
we introduce an exponential distribution
\begin{equation}
G(t^{\prime})=1-\me^{-\gamma_0 (t^{\prime}-t)},\quad
t^\prime\geq t,
\end{equation}
and define its $n$-th convolution by
\begin{equation}
G_0(t^\prime)=G(t^\prime), \quad G_{n}(t^{\prime})=\int_{0}^{t^\prime} G_{n-1}(t^{\prime}-u) \D G(u).
\end{equation}
It can be readily verified that
\begin{equation}
\frac{\D \mathcal{G}(t^{\prime}; \bm{r}, u)}{\D t^{\prime}} \le \frac{k}{ 2M+1} \cdot \frac{\D G(t^{\prime})}{\D t^{\prime}}, ~~\forall \, t^{\prime} \in [u, T], ~ \forall \,\bm{r}\in\mathbb{R}^d, ~ \forall \, u\in [0,T],
\end{equation}
holds for a sufficiently large integer $k$, e.g.,   $k>(2M+1)\me^{\gamma_0 T}$.

We first show by the mathematical induction that there exists a sufficient large integer $k$ and a sufficient large constant $\gamma_0>0$
such that
\begin{equation}\label{Eq:estimate}
\Pr(\{t+\tau_0+\cdots+\tau_{i_1 i_2 \cdots i_{n-1}}<T-u \}) \le (\frac{k}{2M+1})^{n-1} G_{n-1}(T-u), ~~~\forall u \in [0, T-t].
\end{equation}

For $n=1$, we only need $\gamma_0 \ge \max\{\gamma(\bm{r})\}$ and then have
\begin{equation}
\begin{split}
\Pr(\{t+\tau_0<T-u\})&=\frac{\D \mathcal{G}(t^{\prime}; \bm{r}, t)}{\D t^{\prime}}\Big |_{t^\prime=T-u}=1-\me^{-\int_{t}^{T-u} \gamma(\bm{r}(s-t)) \D s} \\
&\le 1- \me^{-\gamma_0(T-u-t)}=G_0(T-u).
\end{split}
\end{equation}

Assume Eq.~\eqref{Eq:estimate} is true for $n$.
Direct calculation shows
\begin{equation}
\begin{split}
&\Pr(\{t+\tau_0+\cdots+\tau_{i_1 i_2 \cdots i_{n}}<T-u\})=\Pi_{t, \bm{r}, \bm{k}}(\mone_{\{t+\tau_0+\cdots+\tau_{i_1 i_2 \cdots i_{n}}<T-u\} })\\
&=\int_{0}^{T-t-u}  \Pi_{t, \bm{r}, \bm{k}} (\mone_{\{t+\tau_0+\cdots+\tau_{i_1 i_2 \cdots i_{n-1}}<T-u-v\} } \mone_{\{\tau_{i_1 i_2 \cdots i_{n}}<v\}  }) \D v\\
&=\int_{0}^{T-t-u}  \Pi_{t, \bm{r}, \bm{k}} (\mone_{\{t+\tau_0+\cdots+\tau_{i_1 i_2 \cdots i_{n-1}}<T-u-v\} } \cdot \Pi_{ \sigma , \eta_{\sigma}}( \mone_{\{\tau_{i_1 i_2 \cdots i_{n}}<v\}  })) \D v\\
&\le \int_{0}^{T-t-u}  (\frac{k}{ 2M+1})^{n-1} G_{n-1}(T-u-v) \cdot \frac{k}{2M+1} \D G(v) = (\frac{k}{2M+1})^{n} G_{n}(T-u),
\end{split}
\end{equation}
which implies that Eq.~\eqref{Eq:estimate} holds for $n+1$,
where $\sigma$ is short for $t+\tau_0+\cdots+\tau_{i_1 i_2 \cdots i_{n-1}}$.

Finally, using Eqs.~\eqref{eq:pi_1} and \eqref{Eq:estimate} yields
\begin{equation}
\begin{split}
\int_{\Omega} \bar{Z}(\omega, T-t) \Pi_{t, \bm{r}, \bm{k}}(\D \omega) &= 1+\sum_{n=1}^{\infty} \sum_{i_1, \cdots, i_{n}=1}^{2M+1} \Pi_{t, \bm{r}, \bm{k}}(1_{i_1 i_2 \cdots i_n}) \\
&\le 1+(2M+1)\sum_{n=1}^{\infty} k^{n-1} G_{n-1}(T),
\end{split}
\end{equation}
and implies $\int_{\Omega} \bar{Z}(\omega, T-t) \Pi_{t, \bm{r}, \bm{k}}(\D \omega)$ is bounded for the infinite series is convergent (see Lemma 1 of the Appendix to Chapter VI in~\cite{bk:Harris1963}). Hence the proof is completed according to Eq.~\eqref{eq:Zbar}.
\end{proof}

Moreover, according to Definition~\ref{eq:exit} and
Theorem~\ref{th:finite}, we can directly show that
$\mu_A$ is integrable, say,
\begin{equation}\label{eq:mu_A_L1}
\int_{\Omega} |\mu_A(\omega)| \Pi_{t, \bm{r}, \bm{k}}(\D \omega) < \infty,
\end{equation}
provided that $A(\bm{r},\bm{k})$ is essentially bounded.
That is, both $\mu_A$ in Eq.~\eqref{eq:muA} and $\psi$ in Eq.~\eqref{eq:psi} are well defined.

With the above preparations, we begin to prove Eqs.~\eqref{eq:adjoint_probabilisitic} and \eqref{eq:inner_product_probabilisitic}.

\begin{theorem}
\label{th:psi}
The first moment $\psi(\bm{r}, \bm{k}, t)$ defined in Eq.~\eqref{eq:psi} equals to the solution of the adjoint equation~\eqref{Adjoint_renewal_equation}.
 \end{theorem}

\begin{proof}
Let $E=\left\{\tau_0: t+\tau_0\ge T\right\}\cap\Omega$ correspond to the case in which the particle travels to $(\bm{r}(T-t), \bm{k})$ and then is frozen. The probability of such event is $1-\mathcal{G}(T; \bm{r}, t)$ by Rule 4. Then the remaining case is denoted by $E^{c}=\left\{\tau_0: t+\tau_0<T\right\}\cap\Omega$. Accordingly, from Eq.~\eqref{eq:psi},
we have
\begin{equation}\label{Eq_75}
\psi (\bm{r}, \bm{k}, t)=\int_{E} \mu_A(\omega) \Pi_{t, \bm{r}, \bm{k}}(\D \omega)+\int_{E^{c}} \mu_A(\omega) \Pi_{t, \bm{r}, \bm{k}}(\D \omega),
\end{equation}
and direct calculation gives
\begin{equation}
\int_{E} \mu_A (\omega) \Pi_{t, \bm{r}, \bm{k}}(\D \omega)=\me^{-\int_{t}^{T} \gamma(\bm{r}(s-t)) \D s} A(\bm{r}(T-t), \bm{k}),
\end{equation}
which recovers the first right-hand-side term of Eq.~\eqref{Adjoint_renewal_equation}.
When event $E^{c}$ occurs,
it indicates that $2M+1$ offsprings are generated.
Notice that $\omega=(Q_{0}; \omega_{1}; \omega_{2}; \cdots; \omega_{2M+1})$
and thus we have
%\begin{align}
%\mu_A(\omega)&=\phi_1 \cdot \mu_A(\omega_{1})+\phi_2 \cdot \mu_A(\omega_{2})+\phi_3  \cdot \mu_A(\omega_{3}), \nonumber\\
%&=\zeta(\bm{r}(\tau_0))\cdot \mu_A(\omega_{1})-\zeta (\bm{r}(\tau_0))\cdot \mu_A(\omega_{2})+\mu_A(\omega_{3}),\label{Eq_79}
%\end{align}
\begin{align}\label{Eq_79}
\mu_A(\omega)&=\sum_{i=1}^{2M+1} \phi_i \cdot \mu_A(\omega_{i})= \sum_{i=1}^{2M+1} \zeta_{i}(\bm{r}(\tau_0))\cdot \mone_{ \{ \bm{k}_i \in \mathcal{K} \} } \cdot \mu_A(\omega_{i})
\end{align}
where we have applied Rule 3.

Substitute Eq.~\eqref{Eq_79} into the second right-hand-side term of Eq.~\eqref{Eq_75} leads to
%\begin{align}
%\int_{E^{c}} \mu_A(\omega) \Pi_{t, \bm{r}, \bm{k}}(\D \omega)&=\int_{E^{c}}\zeta(\bm{r}(\tau_0))  \left\{ \int_{{\Omega_1}} \mu_A(\omega_1) \Pi_{t+\tau_0, \bm{r}_1, \bm{k}_1}(\D \omega_1)\right\}\Pi_{t, \bm{r}, \bm{k}}(\D \omega)\nonumber\\
%&-\int_{E^{c}}\zeta(\bm{r}(\tau_0))\left\{ \int_{{\Omega_2}} \mu_A(\omega_2) \Pi_{t+\tau_0, \bm{r}_2, \bm{k}_2}(\D \omega_2)\right\} \Pi_{t, \bm{r}, \bm{k}}(\D \omega) \label{eq:Ec}\\
%&+\int_{E^{c}} \left\{ \int_{{\Omega_3}} \mu_A(\omega_3) \Pi_{t+\tau_0, \bm{r}_3, \bm{k}_3}(\D \omega_3)\right\}\Pi_{t, \bm{r}, \bm{k}}(\D \omega),\nonumber
%\end{align}
\begin{equation}\label{eq:Ec}
\int_{E^{c}} \mu_A(\omega) \Pi_{t, \bm{r}, \bm{k}}(\D \omega)=\sum_{i=1}^{2M+1} \int_{E^{c} \cap \{ \bm{k}_i \in \mathcal{K} \} }\zeta_i(\bm{r}(\tau_0))  \left\{ \int_{{\Omega_i}} \mu_A(\omega_i) \Pi_{t+\tau_0, \bm{r}_i, \bm{k}_i}(\D \omega_i)\right\}\Pi_{t, \bm{r}, \bm{k}}(\D \omega).
\end{equation}
where we have used the Markov property~\eqref{Def:Markov_property}
as well as the mutual independence among the subfamilies inherited in Rule 5.

Finally, by the definition \eqref{eq:psi}, we have
\begin{equation}
\int_{{\Omega_i}} \mu_A(\omega_{i}) \Pi_{t+\tau_0, \bm{r}_{i}, \bm{k}_{i}}(\D \omega_{i})=\psi(\bm{r}_{i}, \bm{k}_{i}, t+\tau_0), \quad i=1,\cdots, 2M+1.
\end{equation}
Then the first right-hand-side term of Eq.~\eqref{eq:Ec} becomes
\begin{equation}
\begin{split}
&\int_{E^{c} \cap \{ \bm{k}_1 \in \mathcal{K} \} } \zeta_1(\bm{r}(\tau_0))  \psi(\bm{r}_1, \bm{k}_1, t+\tau_0) \Pi_{t, \bm{r}, \bm{k}}(\D \omega)=\int_{t}^{T} \D t_1 \me^{-\int_{t}^{t_1} \gamma(\bm{r}(s-t)) \D s} \\
& \times\int_{\mathbb{R}^d} \D \bm{r}_1 \int_{\mathcal{K}}\D \bm{k}_1\left\{V_{w}^{-}(\bm{r}(t_1-t), \bm{k}-\bm{k}_{1})\cdot \delta(\bm{r}(t_1-t)-\bm{r}_1)\right\}\psi(\bm{r}_{1}, \bm{k}_1, t_1),
\end{split}
\end{equation}
where we have let $t+\tau_0 \to t_1$ and used Eqs.~\eqref{eq:zeta} and \eqref{eq:p1}.
The remaining right-hand-side term of Eq.~\eqref{eq:Ec} terms can be treated in a similar way,
and putting them together recovers the second right-hand-side term of Eq.~\eqref{Adjoint_renewal_equation}.
We complete the proof.
\end{proof}

So far we have proven the existence of the solution of the adjoint equation~\eqref{Adjoint_renewal_equation},
while its uniqueness can be deduced by the Fredholm alternative. It remains to validate Eq.~\eqref{eq:inner_product_probabilisitic}. To this end, we let
$\nu$ be a probability measure on the Borel sets of $\mathbb{R}^d \times \mathcal{K}$ with the density $f_I(\bm{r}, \bm{k}, 0)$, then it yields a unique product measure $\nu \otimes \Pi_{t, \bm{r}, \bm{k}}$.
Consequently, from Eq.~\eqref{eq:A_approx2}, we obtain
\begin{equation}
\begin{split}
\langle \hat{A} \rangle_{T} = &\iint_{\mathbb{R}^d \times \mathcal{K}} \varphi(\bm{r}, \bm{k}, 0) \cdot \frac{f(\bm{r}, \bm{k}, 0)}{f_I(\bm{r}, \bm{k}, 0)} \cdot  f_I(\bm{r}, \bm{k},0) ~\D \bm{r} \D \bm{k}\\
=&\iint_{\mathbb{R}^d \times \mathcal{K}}
\left[\int_\Omega \mu_A(\omega) \Pi_{t_l, \bm{r}, \bm{k}} (\D \omega)\right]
\cdot \frac{f(\bm{r}, \bm{k}, 0)}{f_I(\bm{r}, \bm{k}, 0)} ~  \nu(\D \bm{r}\times \D \bm{k}) \\
=&\iint_{\mathbb{R}^d \times \mathcal{K} \times \Omega} \mu_A (\omega)  \cdot \frac{f(\bm{r}, \bm{k}, 0)}{f_I(\bm{r}, \bm{k}, 0)} ~ \nu\otimes\Pi_{0, \bm{r}, \bm{k}}(\D \bm{r} \times \D \bm{k} \times \D \omega)  \\
=&\mathbb{E}_{f_I}\left[ \Pi_{t_l, \bm{r}_\alpha, \bm{k}_\alpha} \left ( s_\alpha(0) \cdot H(0) \cdot \sum_{i\in\mathcal{E}(\omega_\alpha)} \phi_{i,\alpha} \cdot A(\bm{r}_{i,\alpha},\bm{k}_{i,\alpha}) \right) \right],
\end{split}
\end{equation}
and thus fully recover Eq.~\eqref{eq:inner_product_probabilisitic} (noting that $\mathcal{E}_\alpha$ and $\mathcal{E}(\omega_\alpha)$ denote the same set).

In particular, we set $t=0$, $\phi_0 = 1$, $A(\bm{r}, \bm{k})\equiv 1$,
and then it's easy to verify that $\varphi (\bm{r}, \bm{k}, t) \equiv 1$ is the unique solution of Eq.~\eqref{Adjoint_renewal_equation}
using the mass conservation~\eqref{eq:mass1}.
By Theorem~\ref{th:psi}, the first moment of $\mu_A$ also equals to 1, i.e.,
\begin{equation}\label{stochastic_mass_conservation}
1\equiv \varphi (\bm{r}, \bm{k}, 0)=\int_\Omega \mu_A(\omega)
~ \Pi_{0,\bm{r},\bm{k}}(\D\omega)=\int_\Omega
(\sum_{i\in\mathcal{E}(\omega)}\phi_i) ~ \Pi_{0,\bm{r},\bm{k}}(\D\omega),
\end{equation}
and then it further implies
\begin{equation}\label{eq:mass_AT}
\int_{\mathbb{R}^d} \D \bm{r} \int_{\mathcal{K}} \D \bm{k} ~ f(\bm{r}, \bm{k}, T) =  \int_{\mathbb{R}^d} \D \bm{r} \int_{\mathcal{K}} \D \bm{k} ~ f(\bm{r}, \bm{k}, 0), \quad \forall\,T\geq 0,
\end{equation}
due to Eqs.~\eqref{eq:AT_fT} and \eqref{eq:important},
which is nothing but the mass conservation law~\eqref{eq:mass_k_truncated}.
Furthermore, for such special case,
we can show in Theorem \ref{Thm_3} that $\sum_{i\in\mathcal{E}(\omega)}\phi_i=1$
is almost sure for any family history $\omega$,
not just the first moment as shown in Eq.~\eqref{stochastic_mass_conservation}.
It implies that any estimator
using finite number of super-particles is still able to preserve the mass conservation with probability $1$.
\begin{equation}
\int_{\mathcal{R}^d \times \mathcal{K}} f(\bm{r}, \bm{k}, T) \D \bm{r} \D \bm{k} \approx \sum_{\alpha} \mu_A(\omega_\alpha) \cdot w_{\alpha}({0})=\sum_{\alpha} w_{\alpha}({0})=1,
~ a.s.
\end{equation}

\begin{theorem}[Mass conservation]\label{Thm_3}
Suppose that the ancestor particle starts at $t=0$ and carries a weight $\phi_0=1$.
Then we have
\begin{equation}
\Pr (E)=\int_{{\Omega}}  \mone_{E}(\omega) \Pi_{0, \bm{r}, \bm{k}} (\D \omega)=1,
\end{equation}
where the event $E$ is given by
\begin{equation}
E=\{\omega \in \Omega: \sum_{i \in \mathcal{E}(\omega)} \phi_i = 1 \}.
\end{equation}
\end{theorem}

\begin{proof}
Since $Z(\omega, T-t)$ only takes odd values,
it suffices to take
\begin{align}
E_{n}&=\{\omega \in \Omega: \sum_{i \in \mathcal{E}(\omega)} \phi_{i}= 1; Z(\omega, T-t) \leq 2n+1\}, \\
E^{\ast}_{n}&=\{\omega \in \Omega: Z(\omega, T-t) \leq 2n+1\},
\end{align}
and it is easy to see $\Pr(E_0)=\Pr(E_0^{\ast})$.
In the remaining part of the proof, we will omit $\omega \in \Omega$ for brevity.  Assume that the statement that $\Pr(E_k)=\Pr( E^{\ast}_k)$ is true
for $0\leq k\leq n$, $\forall t\in [0, T]$. We show below by the mathematical induction that it still holds for $k=n+1$.

From $\omega=(Q_{0};\omega_{1}; \cdots; \omega_{2M+1})$,
it can be easily verified that
\begin{equation}
Z(\omega, T-t)=\sum_{l=1}^{2M+1} Z(\omega_l, T-t-\tau_0)
\end{equation}
thus we have
\begin{equation}
Z(\omega_l, T-t-\tau_0)< Z(\omega, T-t), ~~ \forall \, l\in\{1,2,\cdots, 2M+1\},
\end{equation}
implying
\begin{equation}\label{condition_1}
E^{\ast}_{n+1} \subset \{Z(\omega_l, T-t-\tau_0) \le 2n+1\},
~~ \forall \, l\in\{1,2,\cdots, 2M+1\}.
\end{equation}
Furthermore, since $ \zeta_{2m-1}(\bm{r}) =  -\zeta_{2m}(\bm{r})$, it yields
\begin{equation}
\sum_{i \in \mathcal{E}(\omega)} \phi_{i}=  \sum_{m=1}^{M} \zeta_{2m-1}(\bm{r}(\tau_0))  \cdot [ \sum_{i \in \mathcal{E}(\omega_{2m-1})} \phi_{i}- \sum_{i \in \mathcal{E}(\omega_{2m})} \phi_{i} ]+\sum_{i \in \mathcal{E}(\omega_{2M+1})} \phi_{i},
\end{equation}
we obtain
\begin{equation}\label{condition_2}
\bigcap_{l=1}^{2M+1} \{\sum_{i \in \mathcal{E}(\omega_l)} \phi_{i}=1 \}   \subset \{\sum_{i \in \mathcal{E}(\omega)} \phi_{i}=1\}.
\end{equation}

Combining Eqs.~\eqref{condition_1} and \eqref{condition_2} with
the conditionally independence of $\omega_l$ yields
\begin{align*}
\frac{\Pr( E_{n+1})}{\Pr(E^{\ast}_{n+1})} &\ge \Pr(\bigcap_{l=1}^{2M+1} \{\sum_{i \in \mathcal{E}(\omega_l)} \phi_{i}=1\} \big | E^{\ast}_{n+1} )=\prod_{l=1}^{2M+1} \Pr( \{\sum_{i \in \mathcal{E}(\omega_l)}  \phi_{i}=1 \} \big | E^{\ast}_{n+1} )\\
& =\prod_{l=1}^{2M+1} \Pr(  \{\sum_{i \in \mathcal{E}(\omega_l)}  \phi_{i}=1, Z(\omega, T-t) \le 2n+3 \} ) / \Pr(E_{n+1}^{\ast})=1,
\end{align*}
where the induction hypothesis is applied in the last line,
and thus $\Pr(E_{n+1}) \ge \Pr(E^{\ast}_{n+1})$.
Accordingly, we have $\Pr(E_{n+1}) = \Pr(E^{\ast}_{n+1})$
for it is obvious that $\Pr(E_{n+1}) \le \Pr(E^{\ast}_{n+1})$.

Finally, according to the fact that  $E_0 \subset E_1 \subset \cdots E_n \subset E_{n+1} \cdots \subset E$, we have
\begin{equation}
\Pr(E)=\lim_{n \to +\infty} \Pr(E_{n}) =1,
\end{equation}
due to the monotone convergence theorem. Hence we complete the proof by setting $t=0$.
\end{proof}

In the proof of Theorem~\ref{Thm_3},
the latent assumption, namely,
\emph{two particles carrying the same weight but opposite sign must
be generated in pair},
is required to conserve numerically the mass.

Now we turn to estimate the growth rate of particles in the branching system.
For the constant auxiliary function $\gamma(\bm{r}) \equiv \gamma_0$, the random life-length $\tau$ of a particle starting at time $t$ is characterized by an exponential distribution
\begin{equation}\label{Def:exponential_distribution_fun}
G(t^{\prime}) = \Pr(\tau<t^{\prime}-t)=1-\me^{-\gamma_0 (t^{\prime}-t)}, ~~ t^{\prime} \ge t.
\end{equation}
In this case, the growth of particle number has been thoroughly
studied in the literature\cite{bk:Harris1963,KosinaNedjalkovSelberherr2004} and we are able to obtain a simple
calculation formula of the particle number as shown in
Theorem~\ref{th:exp}.

\begin{theorem}\label{th:exp}
Suppose the family history $\omega$ starts at $t=0$ and the constant
auxiliary function $\gamma(\bm{r}) \equiv \gamma_0$ is adopted. Then
the expectation of the total number of frozen particles in time
interval $[0, T]$ is
\begin{equation}\label{particle_number_formula}
\mathbb{E}Z_T=\me^{2M\gamma_0 T}.
\end{equation}
\end{theorem}

\begin{proof}
According to Theorem 15.1 of Chapter VI in~\cite{bk:Harris1963}, the expectation $\mathbb{E}Z_{t^{\prime}}$ satisfies the following renewal integral equation
\begin{equation}\label{renewal_equation_N}
\mathbb{E}Z_{t^{\prime}}=1-G(t^{\prime})+(2M+1) \int_{0}^{t^{\prime}} \mathbb{E}Z_{t^{\prime}-u} \D G(u), ~~ \mathbb{E}Z_0=1.
\end{equation}
We substitute Eq.~\eqref{Def:exponential_distribution_fun} into
Eq.~\eqref{renewal_equation_N} and then can easily verify that $\mathbb{E}Z_{t^{\prime}}=\me^{2M\gamma_0 t^{\prime}}$  is the solution.
The proof is finished.
\end{proof}

Considering the fact that randomly generated $\bm{k}^{\prime}$ may be rejected  in numerical application according to the indicator function in Eq.~\eqref{eq:phii}, we can modify Eq.~\eqref{renewal_equation_N} by replacing $2M+1$ with $2\alpha_0 M+1$, with $\alpha_0$ the average acceptance ratio of $\bm{k}^{\prime}$. In consequence, the modified expectation of  total particle number is $\mathbb{E}Z_{T} = \me^{2\alpha_0 M \gamma_0 T}$.

Theorem~\ref{th:exp} also provides an upper bound of $\mathbb{E}Z_T$ when the variable auxiliary function  satisfies $\gamma(\bm{r}) \le \gamma_0$. Whatever the auxiliary function is, it is clear that the particle number will grow exponentially.  \emph{To suppress the particle number, we can either decrease the parameter $\gamma_0$ or choose a smaller final time $T$}.

Finally, suppose we would like to evolve the branching particle system until the final time $T$. Usually, there are two ways. One is to evolve the system until each particle is frozen at the final time $T$ in a single step. The other is to divide $T$ into
$0=t_0<t_{1}<\cdots <t_{n-1}<t_{n}=T$ with $t_n=n\Delta t$,
and then we evolve the system successively in $n$ steps.
However, the following theorem tells us that both produce the same $\mathbb{E}Z_T$.

\begin{theorem}[Theorem 11.1 of Chapter VI in~\cite{bk:Harris1963}] \label{th:G(t)}
Suppose $G(t)=1-\me^{-\gamma_0 t}$. Then $Z(\omega, t)$ is a Markov branching process. In addition, $Z(\omega, \Delta t), Z(\omega, 2\Delta t), \cdots $ is a Galton-Watson process.
\end{theorem}

We recall that for a Galton-Watson model, if $\mathbb{E}Z_{\Delta t}=\beta$, then $\mathbb{E}Z_{n\Delta t}=\beta^{n}$. From Eq.~\eqref{particle_number_formula}, we know that $\beta= \me^{2M \gamma_0 \Delta t}$, so that  $\mathbb{E}Z_{n\Delta t}=\me^{2M \gamma_0  T}$. Therefore, we cannot expect to reduce the particle number by simply dividing $T$ into several steps and evolve the particle system successively, which also manifests the indispensability of resampling.

\subsection{Resampling}

As illustrated in Section \ref{sec:importance_sample}, it suffices to set the initial and final time to be $t_l$ and $t_{l+1}$, respectively. Thus from the integrability of $\mu_A$ in Eq.~\eqref{eq:mu_A_L1} and the strong law of large number in Eq.~\eqref{eq:strong}, we can use the following estimator to calculate Eq.~\eqref{eq:inner_product_probabilisitic}.
\begin{equation}\label{inner_product_estimator}
\langle \hat{A} \rangle_{t_{l+1}} \approx {\sum_{\alpha=1}^{N_\alpha} \mu_A(\omega_\alpha)  \cdot w_\alpha(t_l)} = \sum_{\alpha=1}^{N_\alpha} \sum_{i\in\mathcal{E}(\omega_\alpha)} \phi_{i,\alpha} \cdot A(\bm{r}_{i,\alpha},\bm{k}_{i,\alpha})\cdot w_\alpha(t_l),
\end{equation}
with ancestor particles $(\bm{r}_{0, \alpha} , \bm{k}_{0,\alpha})$ drawn from $f_{I}(\bm{r}, \bm{k}, t_l)$, which
converges almost surely when $N_\alpha \to \infty$. According to Eq.~\eqref{inner_product_estimator},
we would like to point out three important features below.
\begin{description}
\item[(1)] It is unnecessary to know the normalizing factor $H(t)$ in Eq.~\eqref{eq:inner_product_probabilisitic}, since it has been absorbed in the sign function $s_{\alpha}(t_l)$ in Eq.~\eqref{eq:s}.

\item[(2)] It is unnecessary to take multiple replicas of branching particle system starting from the same ancestor because we only need to evaluate the expectation of $\mu_A(\omega_\alpha)\cdot w_{\alpha}(t_l)$ with respect to the product measure $\nu \otimes \Pi_{t_l, \bm{r}_\alpha, \bm{k}_\alpha}$.

\item[(3)] Theorem \ref{Thm_3} ensures the mass conservation property. It must be mentioned that the conserved quantity is the summation of particle sign function $\sum_{\alpha=1}^{N_\alpha} s_{\alpha}(t_l)$, instead of total particle number $N_\alpha$. In fact, the total particle number may increase in order to capture the negative values of Wigner function.
\end{description}

Unfortunately, Theorem \ref{th:exp} has presented an unpleasant
property of such estimator, namely, the exponentially increasing
complexity. In this regard, a resampling procedure, which is based on the statistical properties of the Wigner function and density estimation method, must be introduced to save the efficiency, say, to
reduce the particle number from $\mathcal{O}(N_{\alpha}
\me^{2\alpha_0 M\gamma_0 \Delta t)})$ to $\mathcal{O}(N_{\alpha})$.

%Such a special kind of resampling is sometimes referred to as the
%subsampling in chemistry. Here our resampling technique is based on
%the statistical properties of the Wigner function and
%Eq.~\eqref{eq:ft_approx} and uses two steps.

%, without storing all the particles in the branching system.
%As the branching particle system grows, the variance of particle weight $\phi_{i, \alpha}$ will increase, and many particle will carry very small weights because of the condition $\left|\zeta_{i}(\bm{r})\right|\leq 1$. The idea of resampling is to replace several particles carrying small weights with a single particle carrying a large weight in order to reduce the variance. More importantly, we expect this procedure can Actually the resampling techniques are widely used in statistical computing, such as pruning, enrichment and bootstrap\cite{bk:Liu2001,bk:RobertCasella2004}.

The first step is to use the non-parameter density estimation method
(the histogram) to evaluate $f(\bm{r}, \bm{k}, t_{l+1})$ through the
branched particles on a given suitable partition of the phase space
$\mathbb{R}^d \times \mathcal{K} = \bigcup_{j=1}^{J}D_j$. The
instrumental density $f_I(\bm{r}, \bm{k}, t_{l+1})$ can be simply
estimated by Eq.~\eqref{eq:histogram}. The successive step is to
draw new samples according to the resulting piecewise constant
density $f_I(\bm{r}, \bm{k}, t_{l+1})$. The main problem is how to
determine the phase space partition. The simplest way, as suggested
in \cite{SellierNedjalkovDimov2014}, is using the uniformly
distributed cells in phase space: $\mathbb{R}^d \times \mathcal{K} =
\bigcup_{j_1=1}^{J_1} \mathcal{X}_{j_1} \times \bigcup_{j_2=1}^{J_2}
\mathcal{K}_{j_2} $, then $f$ is estimated by a piecewise constant
function
\begin{equation}
f(\bm{r}, \bm{k}, t_{l}) \approx \sum_{j_1 = 1}^{J_1}\sum_{j_2=1}^{J_2} d_{j_1, j_2}(t_l) \cdot \mone_{\mathcal{X}_{j_1} \times \mathcal{K}_{j_2}}(\bm{r}, \bm{k}),
\end{equation}
with $d_{j_1, j_2}(t) = {W_{\mathcal{X}_{j_1} \times \mathcal{K}_{j_2}}(t)}/{(|\mathcal{X}_{j_1}| \cdot |\mathcal{K}_{j_2}|)}$. Then the number of particles allocated in each cell is determined by $W_{\mathcal{X}_{j_1} \times \mathcal{K}_{j_2}}(t_{l})$ and the sign by $W_{\mathcal{X}_{j_1} \times \mathcal{K}_{j_2}}(t_{l})/|W_{\mathcal{X}_{j_1} \times \mathcal{K}_{j_2}}(t_{l})|$. The position and wave vector are assumed to be uniformly distributed in each cell.  This approach, usually termed {\sl annihilation} in previous work, e.g. \cite{SellierNedjalkovDimov2014,ShaoSellier2015,th:Ellinghaus2016}, can reduce the particle number effectively for $d=1$ and still works fairly for $d=2$ (as shown in Section \ref{sec:num_res}). Unfortunately, it cannot work for higher dimensional systems because of the following problems, as also manifested in the statistical community\cite{bk:Liu2001,bk:RobertCasella2004,bk:HastieTibshiraniFriedman2009}.

\begin{description}

\item[(1)] In high dimensional situations, the dimension $J_1 \times J_2$ of the feature space (phase space cells)  is too much higher than the sample number, leading to a non-sparse structure and severe over-fitting.

\item[(2)] The uniform distributed hypercube in high-dimensional space is not very useful to characterize the edges of samples.

\item[(3)] The piecewise constant function is discontinuous in nature, so that sampling from a locally uniform distribution may cause additional bias.

\end{description}

To resolve these problems, one can utilize many advanced techniques in the statistical learning and density estimation. The key point is to choose an appropriate $J$ (or feature in statistical terminology) of the partition.  In principle, $J$ must be chosen to strike the balance between accuracy and efficiency. Too small $J$ is unable to capture the fine structure of the Wigner function $f$, whereas too large $J$ may increase the complexity and overfit $f$. For instance, a possible approach is to resort to tree-based methods to partition the phase space \cite{bk:HastieTibshiraniFriedman2009}. Considering that all the statistical techniques are devised for estimating a positive semidefinite density, instead of the quasi-distribution, we need to separate the positive and negative signed particles into two groups and make individual histograms, then merge them into a piecewise constant function. Such pattern is based on the decomposition of the signed measure. In a word, the story in the higher dimensional phase space is totally different because how to efficiently implement the so-called annihilation exploiting the cancelation of weights of opposite sign 
is still in progress.

The resampling in high dimensional phase space is a complicated issue and beyond the scope of this paper, so we would like to discuss it in our future work. In Section~\ref{sec:num_res}, we mainly focus on several typical tests for $d=1$ and $d=2$ and show the accuracy of WBRW as well as of the piecewise constant approximation by comparing with two accurate deterministic solvers,
i.e., SEM\cite{ShaoLuCai2011} and ASM\cite{XiongChenShao2015}.

In summary, the outline of WBRW is illustrated below from $t_l$ to $t_{l+1}$ with the time step $\Delta t$, $l=1,2,\cdots,n-1$. It suffices to take $\bm{r}=\bm{r}_\alpha, \bm{k}=\bm{k}_\alpha, t=t_l$ as the initial state,
and $\bm{r}^{\prime}=\bm{r}^{\prime}_\alpha, \bm{k}^{\prime}=\bm{k}^{\prime}_\alpha, t^{\prime}=t^{\prime}_{\alpha}$ for the offsprings in Eq.~\eqref{Adjoint_renewal_type_equation}.

\begin{description}

\item[Step 1: Sample from $f_I(\bm{r}, \bm{k}, t_l)$]
The first step is to sample $N_\alpha$  ancestor particles according to the instrumental distribution $f_I(\bm{r}, \bm{k}, t_l)$ (see Eq.~\eqref{instrumental_distribution}). Each particle
has a state $(\bm{r}_{\alpha}, \bm{k}_{\alpha})$ and carries an initial weight $\phi_{\alpha}$ and a sign $s_\alpha$. In general, we can simply take $\phi_{\alpha} = 1$.

\begin{figure}
\centering
\includegraphics[width=0.49\textwidth,height=0.38\textwidth]{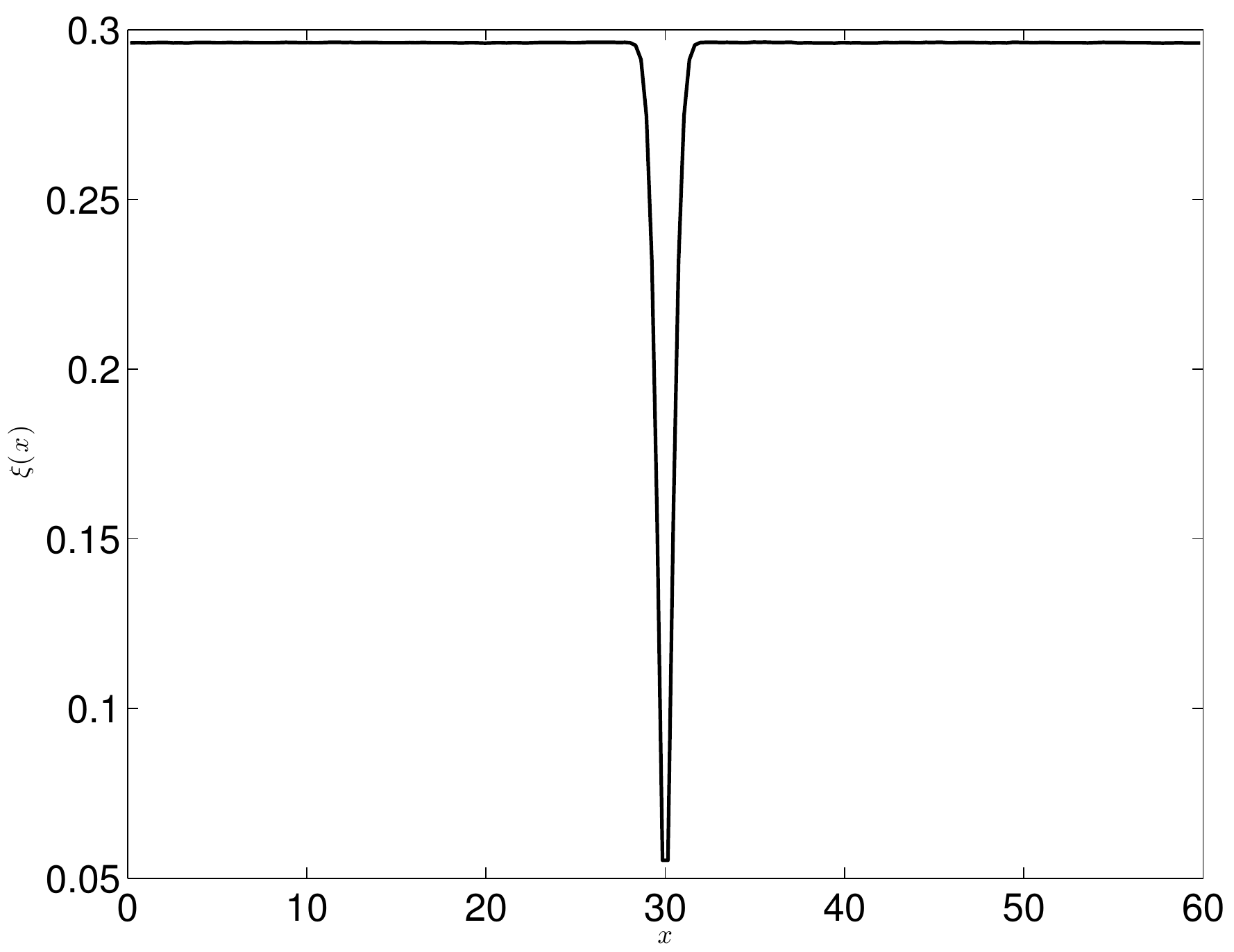}
\caption{\small The normalization factor $\xi(x)$ for the Gaussian barrier~\eqref{eq:gp} with $H_B=0.3$eV and $x_B=30$nm is utilized in the $k$-truncated Wigner simulations.}
\label{fig:xi}
\end{figure}

\item[Step 2: Evolve the particles]
The second step is to evolve super-particles according to the rules of branching particle systems. Suppose a particle is born at $t^{\prime}_\alpha \in [t_l, t_{l+1}]$ at state $(\bm{r}_\alpha^{\prime}, \bm{k}_\alpha^{\prime})$ with weight  $\phi_{\alpha}^\prime$, and it has a random life-length $\tau_{\alpha}^{\prime}$ satisfying
\begin{equation}
\tau_{\alpha}^{\prime} \propto \frac{\dif \mathcal{G}(t^\prime; \bm{r}_\alpha^\prime, t_\alpha^\prime)}{\dif t^\prime}\Big |_{t^\prime=t_\alpha^\prime+\tau_\alpha^\prime}=\gamma(\bm{r}_{\alpha}^{\prime}(\tau_{\alpha}^{\prime})) \me^{-\int_{t^{\prime}_{\alpha}}^{t^{\prime}_{\alpha}+\tau_{\alpha}^{\prime}} \gamma(\bm{r}^{\prime}_{\alpha}(s-t^{\prime}_\alpha))\D s}.
\end{equation}
For the ancestor particle, we have  $t_{\alpha}^{\prime}=t_l$, $(\bm{r}_\alpha^{\prime}, \bm{k}_\alpha^{\prime})= (\bm{r}_\alpha, \bm{k}_\alpha)$, $\phi_{\alpha}^{\prime}= \phi_{\alpha}$.

If $\tau_{\alpha}^{\prime} \ge t_{l+1}-t_{\alpha}^{\prime}$,
the particle is frozen at the state $(\bm{r}^{\prime}_{\alpha}(t_{l+1}-t^{\prime}_{\alpha}), \bm{k}^{\prime}_{\alpha})$ and the probability of this event is
\begin{equation}
\Pr(\tau_{\alpha}^{\prime} \ge t_{l+1}-t_{\alpha}^{\prime})=
1-\mathcal{G}(t_{l+1}; \bm{r}_\alpha^\prime, t_\alpha^\prime) =\me^{-\int_{t_{\alpha}^{\prime}}^{t_{l+1}} \gamma(\bm{r}_{\alpha}^{\prime}(s-t_{\alpha}^{\prime})) \D s}.
\end{equation}
Otherwise, the particle travels to a new position $\bm{r}^{\prime}_{\alpha}(\tau^{\prime}_{\alpha})$ and dies at time $t^{\prime}_{\alpha}+\tau_{\alpha}^{\prime}$ at state $(\bm{r}^{\prime}_{\alpha}(\tau^{\prime}_{\alpha}), \bm{k}^{\prime}_{\alpha})$, and meanwhile, several new particles are generated according to Rule 3, the probability of which is $\mathcal{G}(t_\alpha^\prime+\tau_\alpha^\prime; \bm{r}_\alpha^\prime, t_\alpha^\prime)$.

\begin{figure}
%\vspace{-2cm}
\subfigure[$t=5$.]{\includegraphics[width=0.49\textwidth,height=0.27\textwidth]{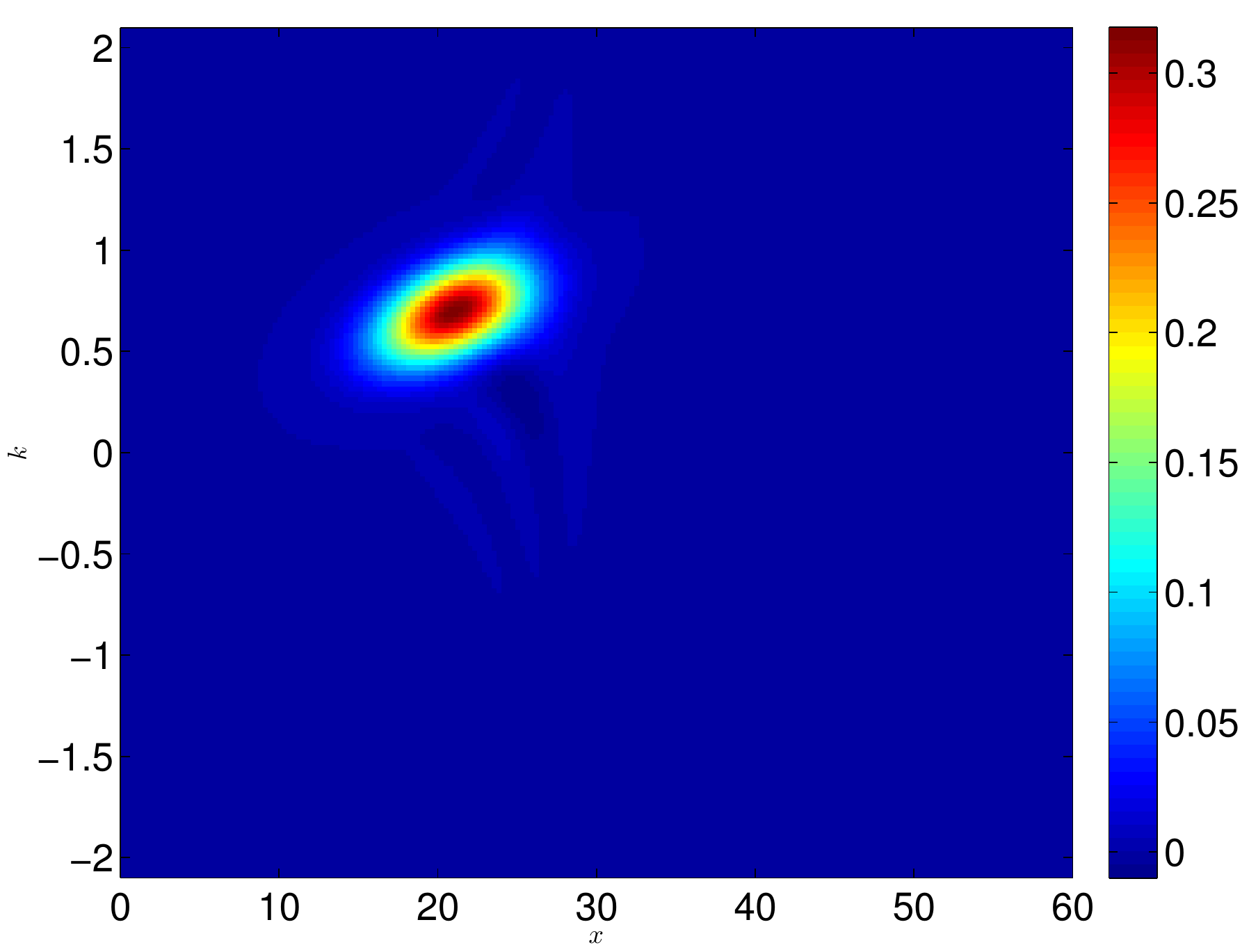}{\includegraphics[width=0.49\textwidth,height=0.27\textwidth]{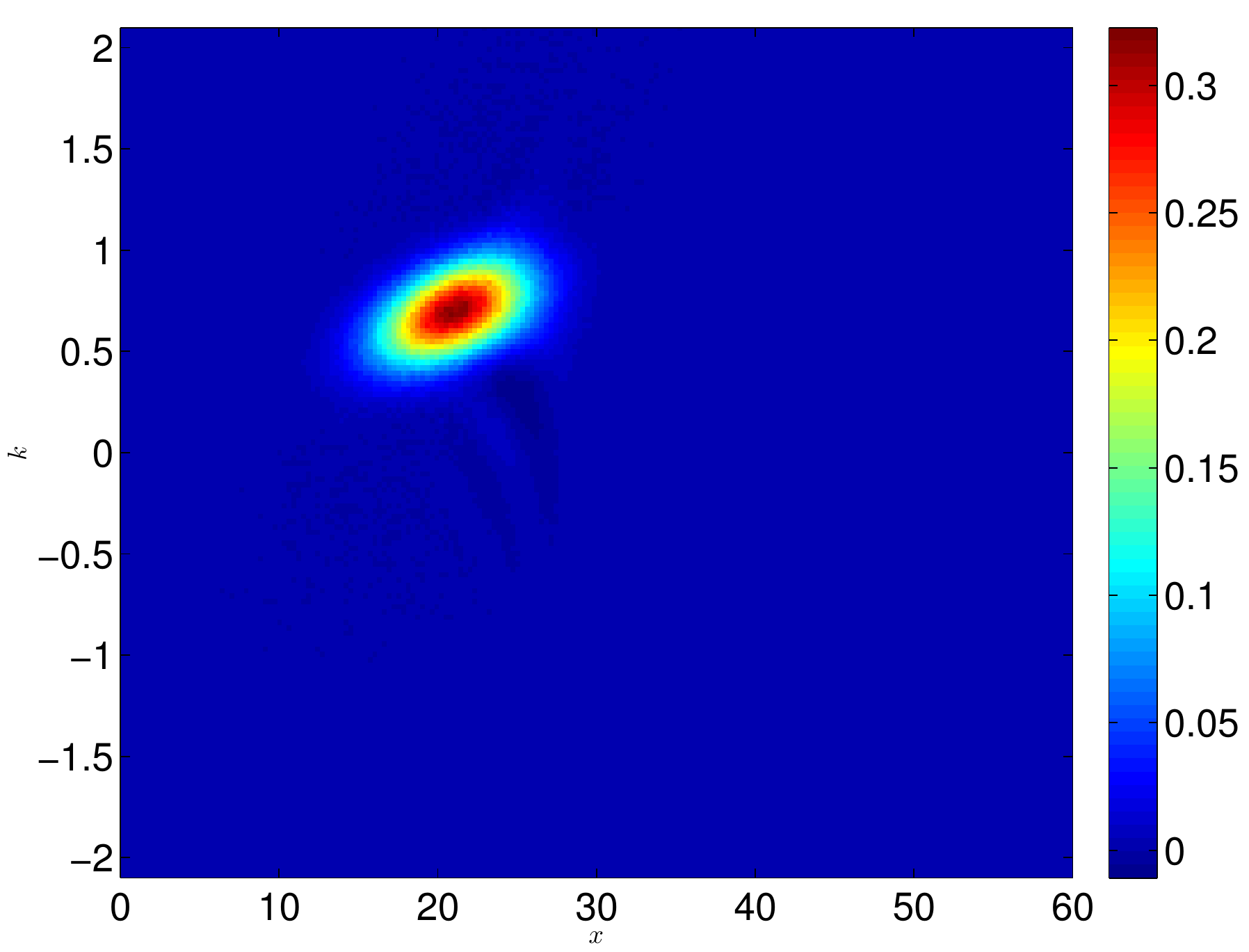}}}
\subfigure[$t=10$.]{\includegraphics[width=0.49\textwidth,height=0.27\textwidth]{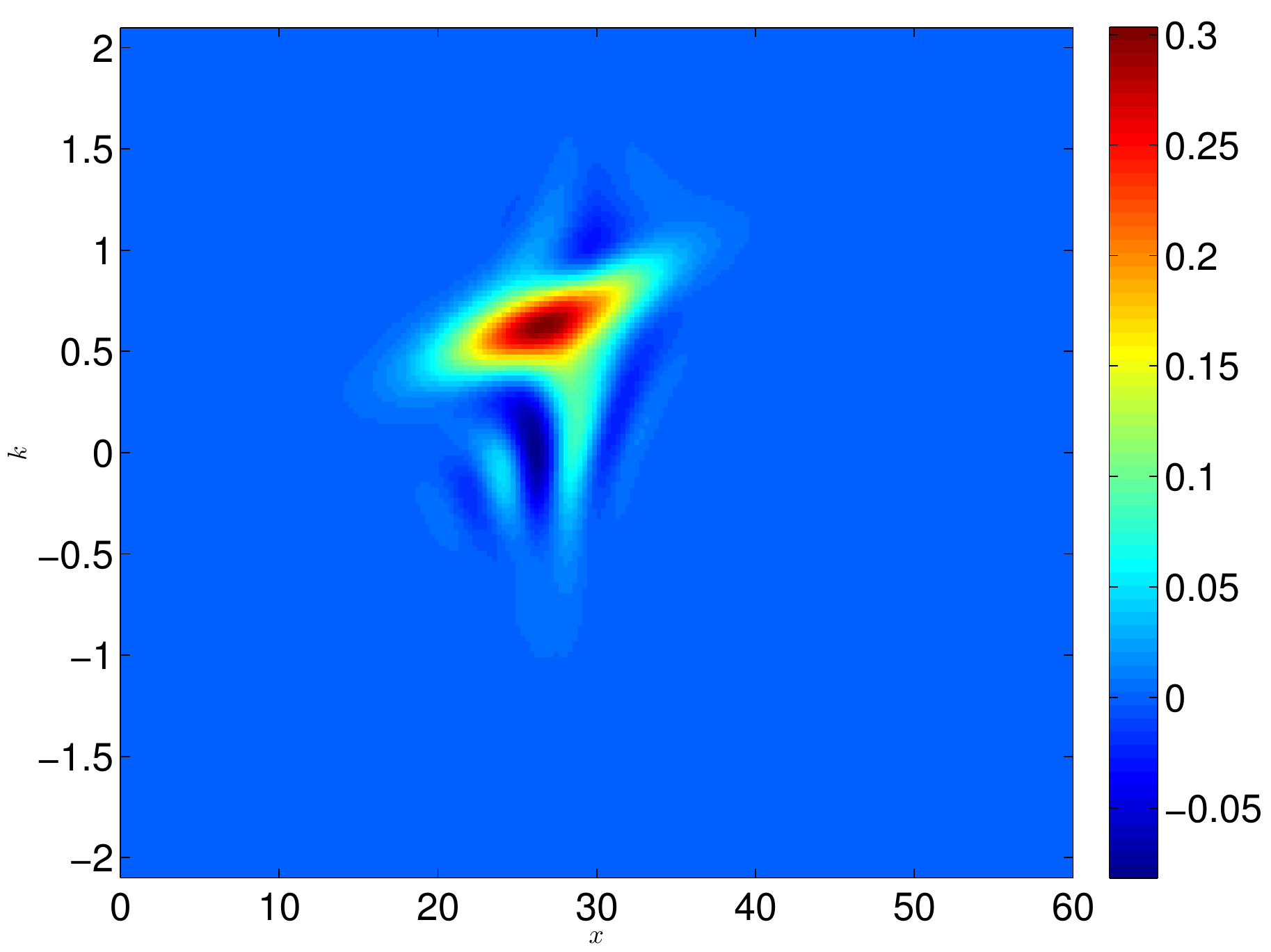}{\includegraphics[width=0.49\textwidth,height=0.27\textwidth]{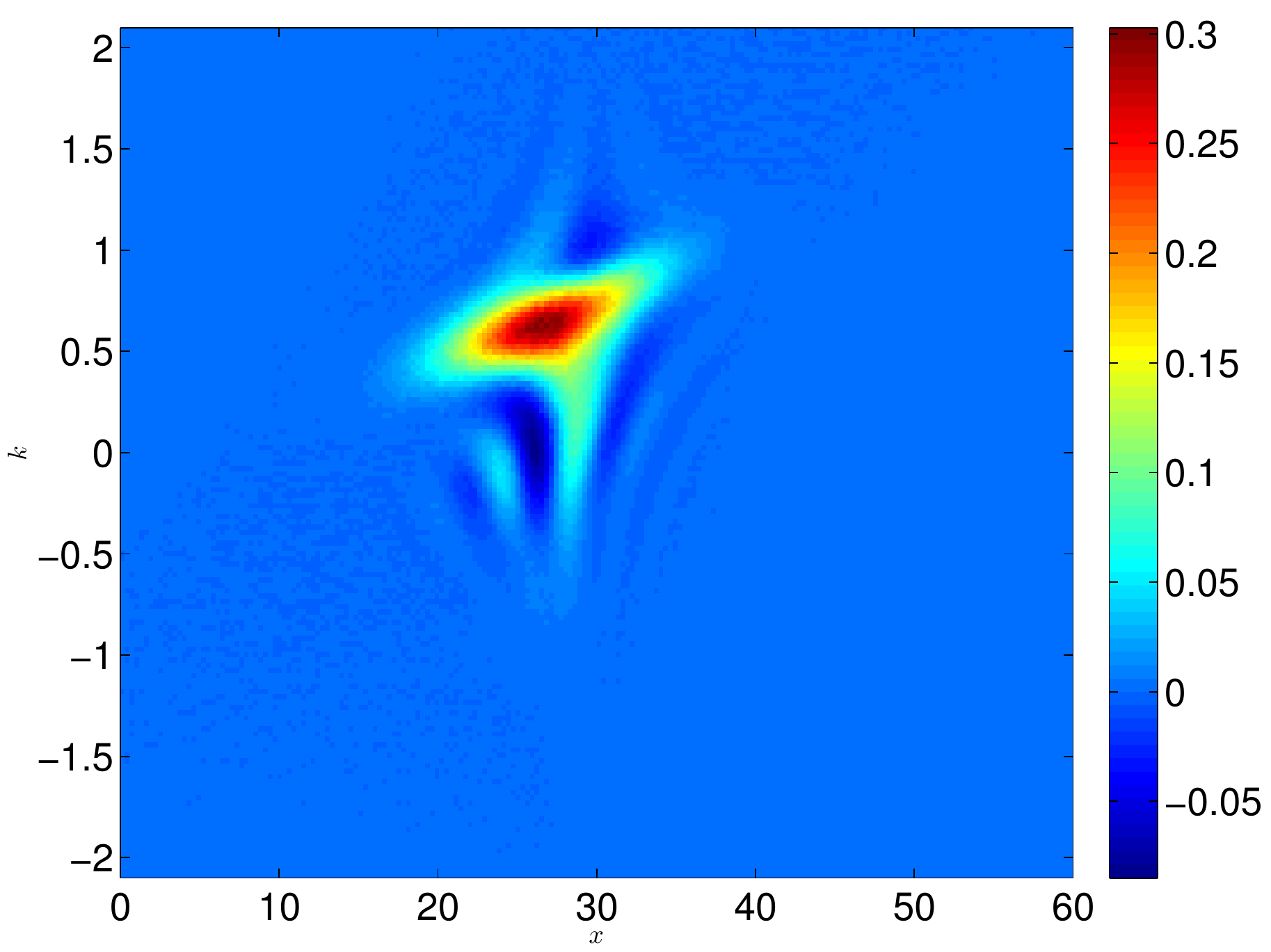}}}
\subfigure[$t=15$.]{\includegraphics[width=0.49\textwidth,height=0.27\textwidth]{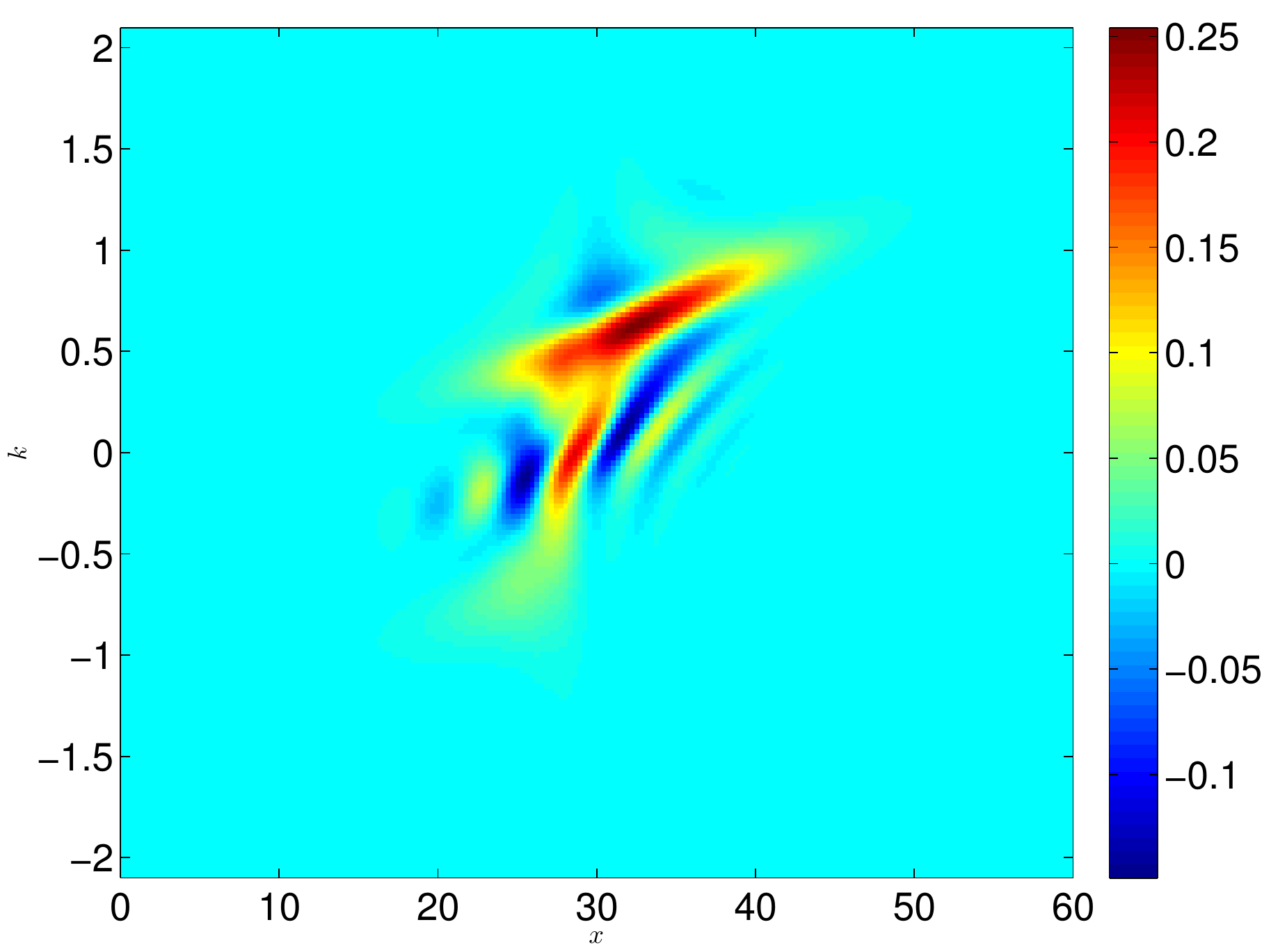}{\includegraphics[width=0.49\textwidth,height=0.27\textwidth]{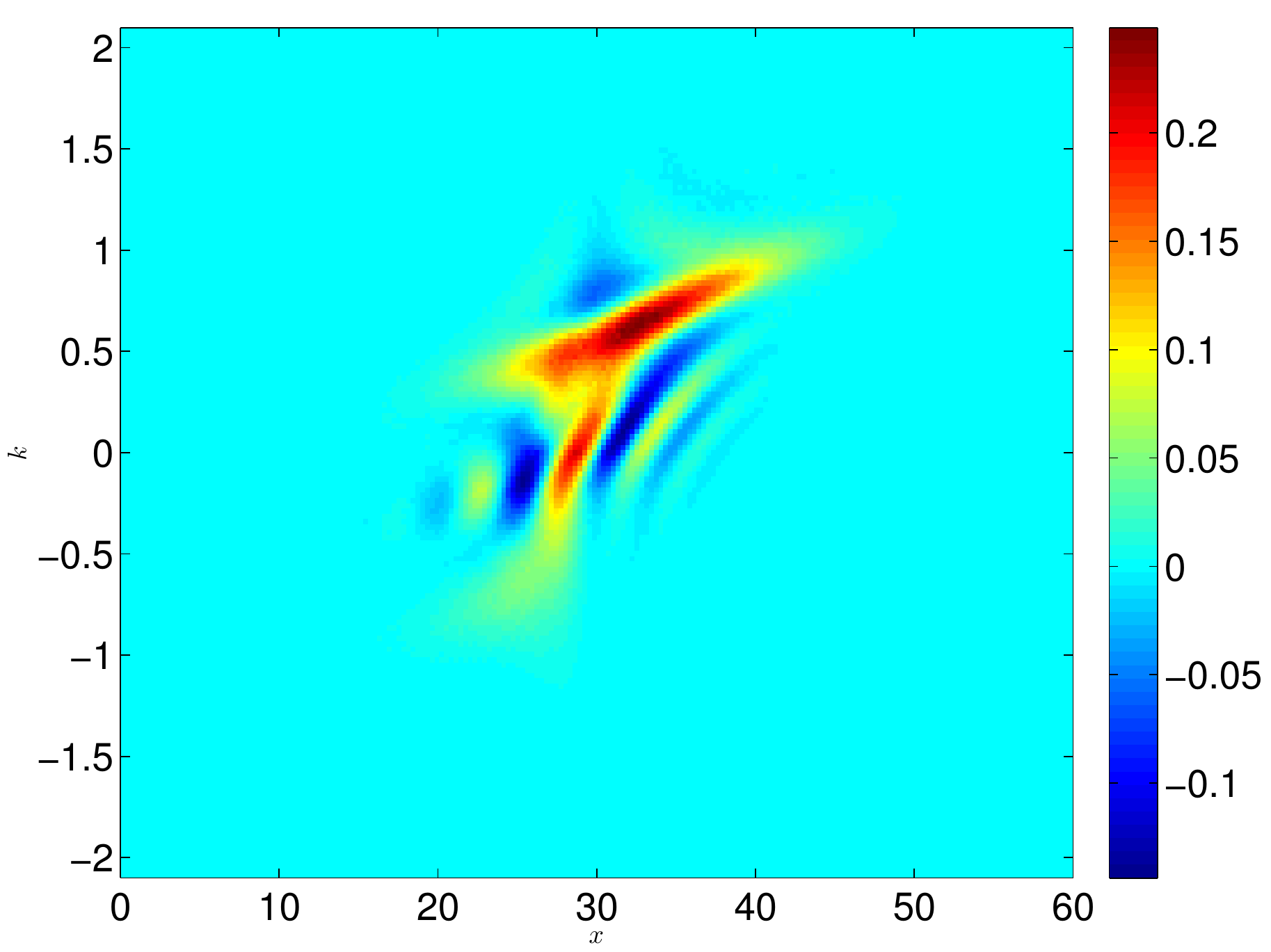}}}
\subfigure[$t=20$.]{\includegraphics[width=0.49\textwidth,height=0.27\textwidth]{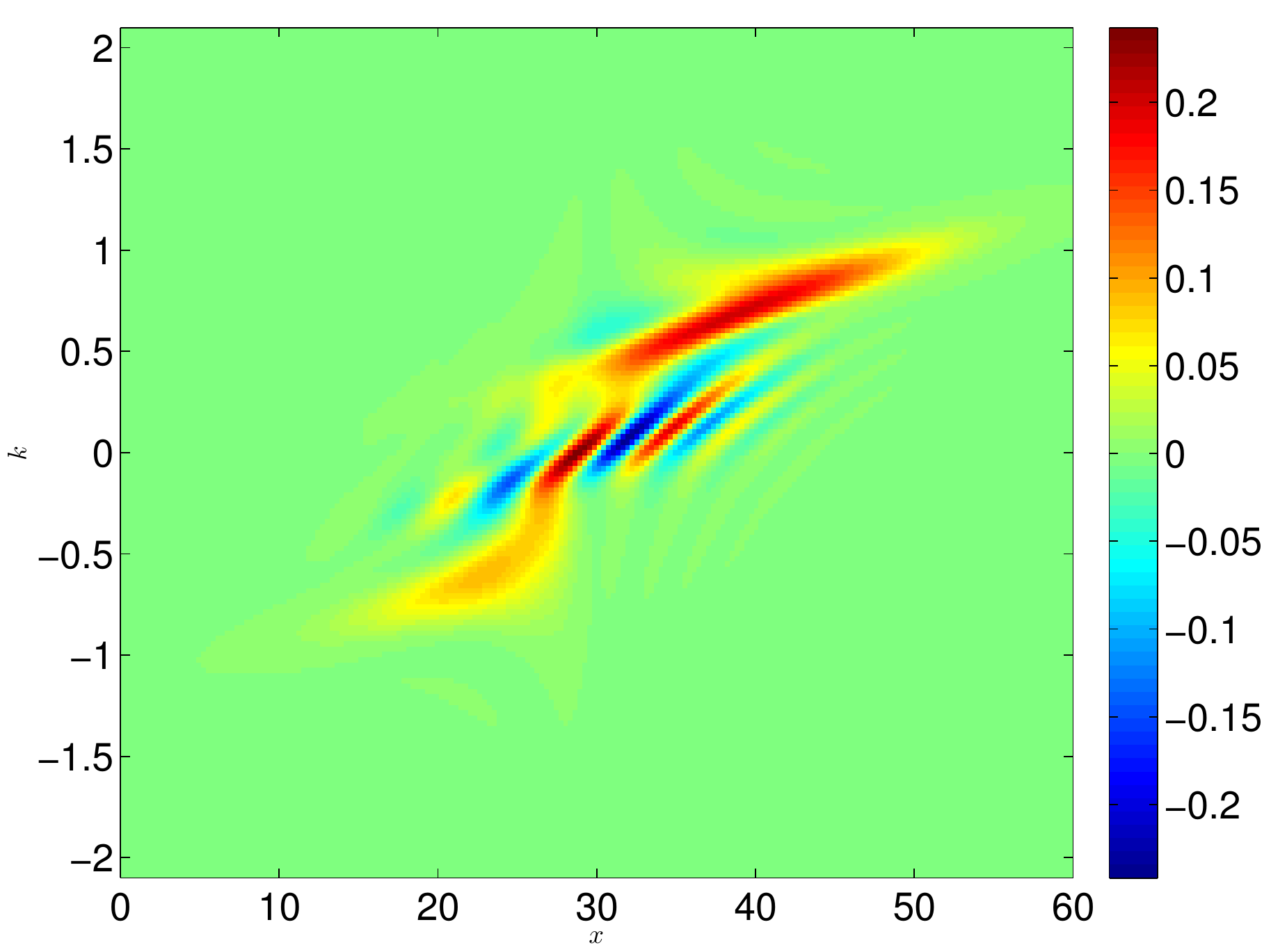}{\includegraphics[width=0.49\textwidth,height=0.27\textwidth]{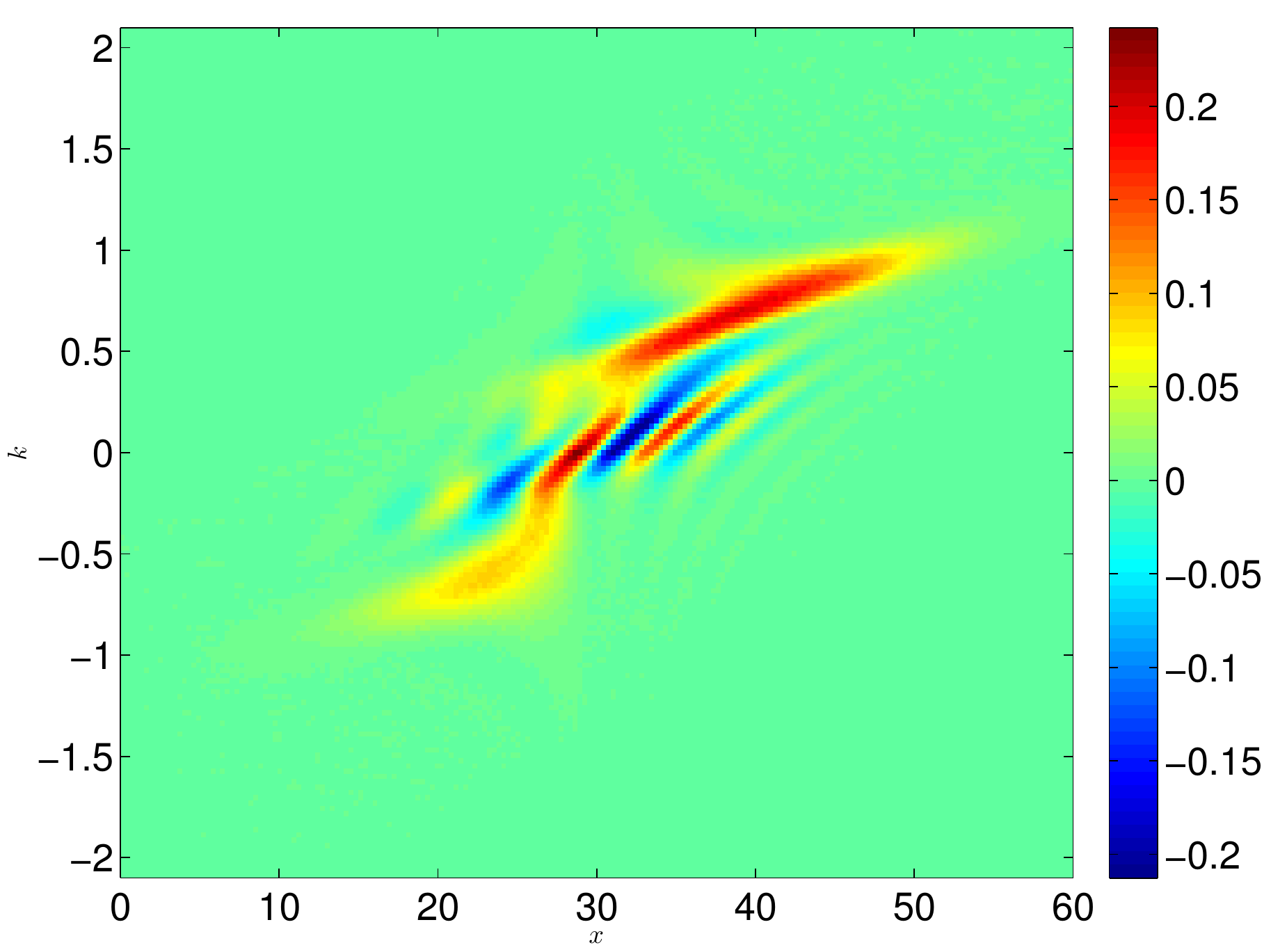}}}
\caption{\small Partial reflection by the Gaussian barrier: Numerical Wigner functions at different time instants $t=5,20,15,20$fs.
The reference solution by SEM is displayed in the left-hand-side column, while the right-hand-side column shows the numerical solution obtained by WBRW with the auxiliary function $\gamma(x)=3\check{\xi}$ as well as $\Delta t=1$fs and $T_A=1$fs.}
\label{fig:ex1-wf}
\end{figure}

\item[Step 3: Density estimation]
When all particles in the branching system are frozen, one can record their positions, wavevectors, and weights.
Let $\mathcal{E}_\alpha$ denote the index set of all frozen particles with
the same ancestor initially at state $(\bm{r}_{\alpha}, \bm{k}_{\alpha})$,
$\{ (\bm{r}_{i, \alpha}, \bm{k}_{i, \alpha}), i\in \mathcal{E}_{\alpha}\}$  the collection of corresponding frozen states,
and $\phi_{i, \alpha}$ the updated weight of the $i$-th particle. Accordingly, $\langle \hat{A} \rangle_{t_{l+1}}$ can be estimated as
\begin{equation}\label{eq:A_approx}
\langle \hat{A} \rangle_{t_{l+1}}\approx \sum_{\alpha} \sum_{i\in \mathcal{E}_\alpha} \phi_{i,\alpha}\cdot w_\alpha(t_l) \cdot  A(\bm{r}_{i, \alpha}, \bm{k}_{i, \alpha}).
\end{equation}
Particularly, plugging into $A(\bm{r}, \bm{k})= \mone_{D_j}(\bm{r},
\bm{k})$, we obtain $W_{D_j}(\bm{r}, \bm{k})$.

Based on a good partition of phase space: $\mathbb{R}^d \times \mathcal{K} = \bigcup_{j=1}^{J}D_j$, we are able to update the instrumental density function $f_I(\bm{r}, \bm{k}, t_{l+1})$ by the histogram \eqref{eq:histogram}.

\end{description}

\section{Numerical experiments}
\label{sec:num_res}

In order to investigate the performance of the WBRW algorithm as well as to verify the theoretical predictions
as we discussed earlier such as the effect of constant $\gamma_0$,
the increasing behavior of the particle number and the
effect of the time step and the annihilation frequency,
we simulate a one-body Gaussian barrier scattering in 2D phase space and a two-body Helium-like system in 4D phase space. The relative $L^{2}$ error is adopted to study the accuracy.
Let $f^{\text{ref}}(x,k,t)$ denote the
reference Wigner function (wf) which could be the exact solution or the numerical solution
on a relatively fine mesh, and $f^{\text{num}}(x,k,t)$ the numerical solution.
Then, the relative errors are written as
\begin{equation}
\text{err}_{wf}(t)  =(\frac{\int_{\mathcal{X}\times\mathcal{K}}(\Delta
f(x,k,t))^{2}\mathrm{d}x\mathrm{d}k}{\int_{\mathcal{X}\times\mathcal{K}}(f^{\text{ref}}(x,k,t))^{2}\mathrm{d}x\mathrm{d}k})^\frac12,
\end{equation}
where $\Delta f(x,k,t)=|f^{\text{num}}(x,k,t)-f^{\text{ref}}(x,k,t)|,$ and the
integrals above are evaluated using a simple rectangular rule over a uniform mesh. To obtain a more complete view of the accuracy, we also
measure corresponding relative errors for physical quantities, e.g.
the spatial marginal (sm) probability distribution and the momental
marginal (mm) probability distribution in a similar way, denoted by
$\text{err}_{sm}(t)$ and $\text{err}_{mm}(t)$, respectively.

\begin{figure}[t!]
\subfigure[The fifth period.]{\includegraphics[width=0.49\textwidth,height=0.38\textwidth]{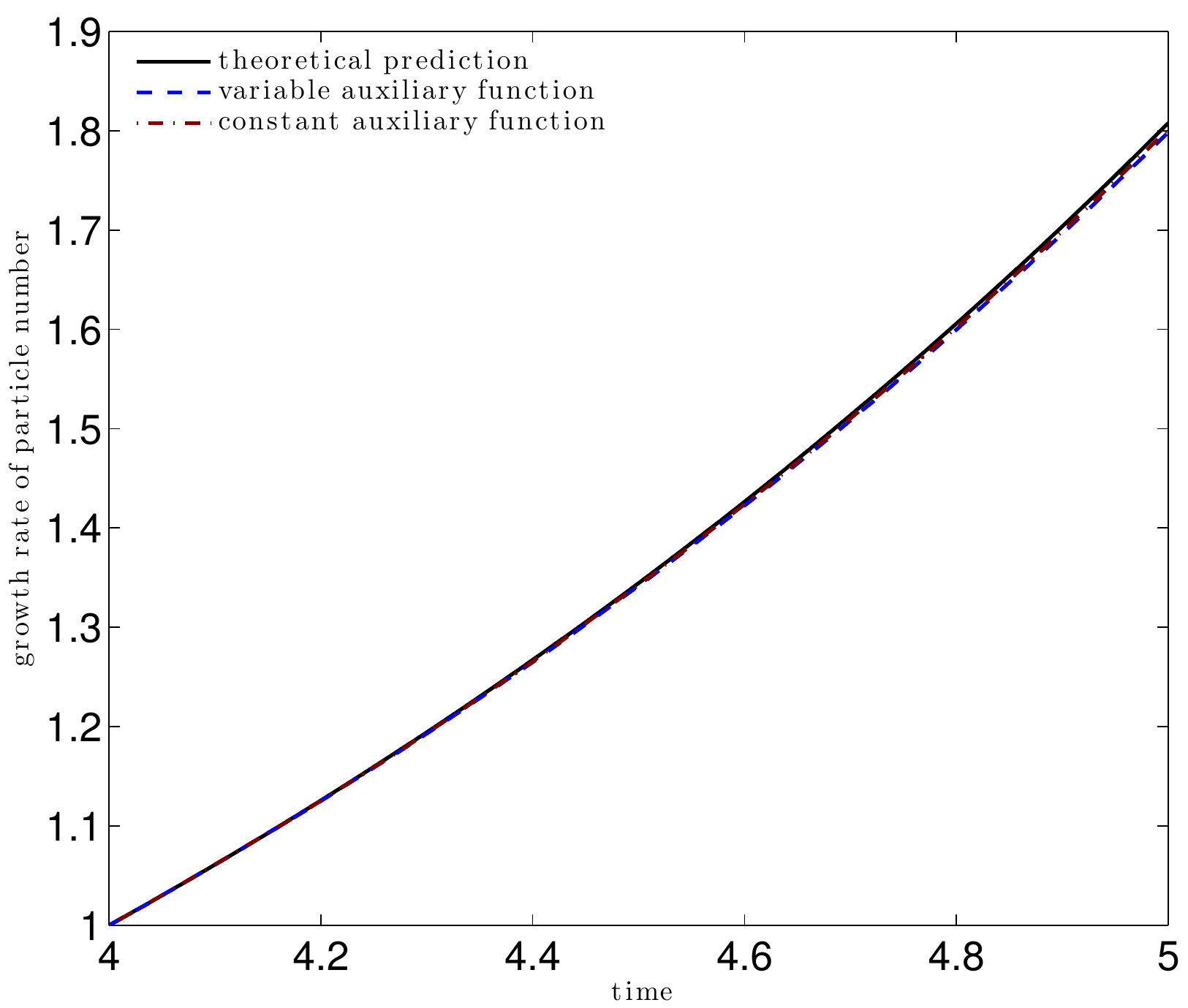}}
\subfigure[The tenth period.]{\includegraphics[width=0.49\textwidth,height=0.38\textwidth]{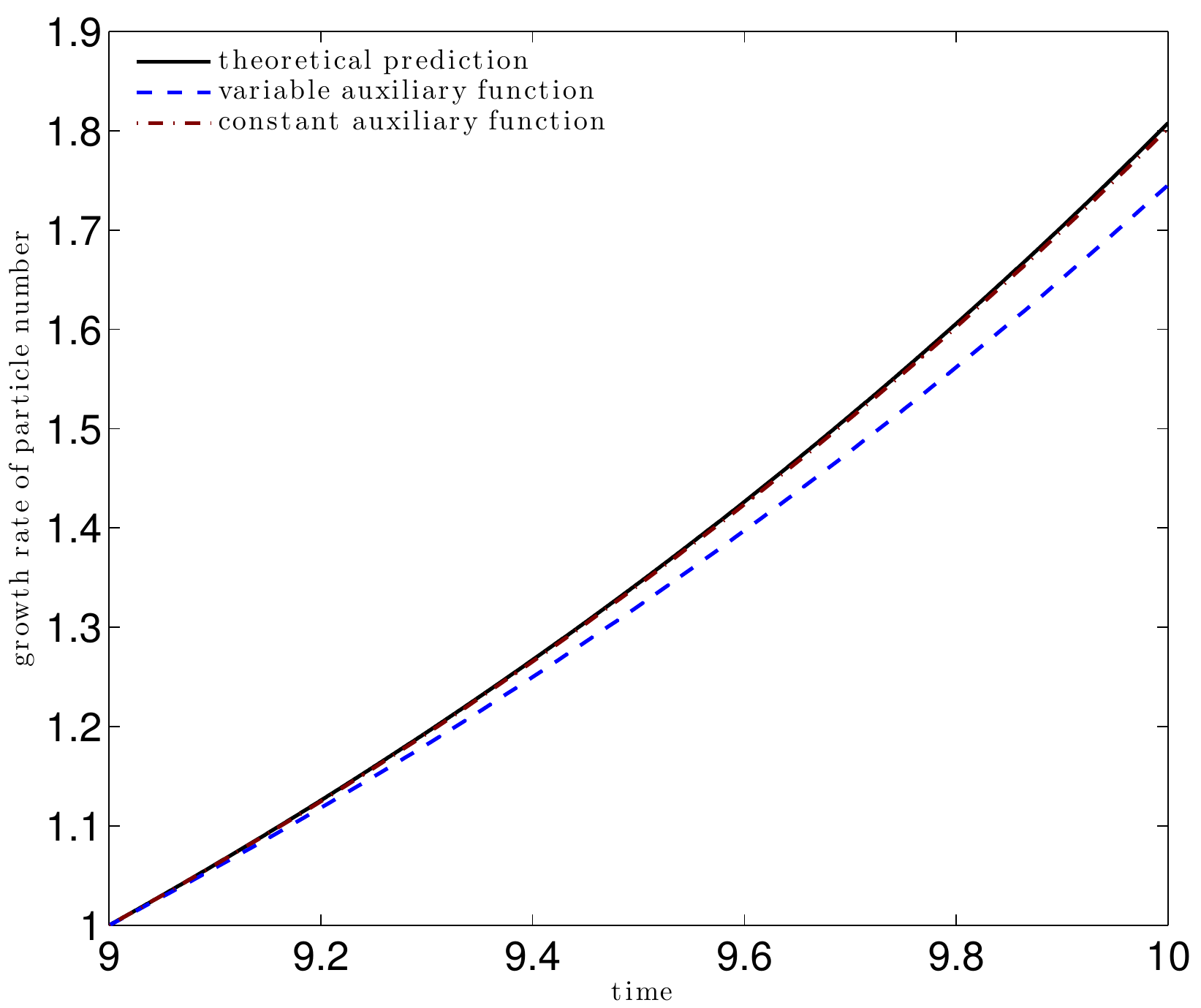}}\\
\subfigure[The fifteenth
period.]{\includegraphics[width=0.49\textwidth,height=0.38\textwidth]{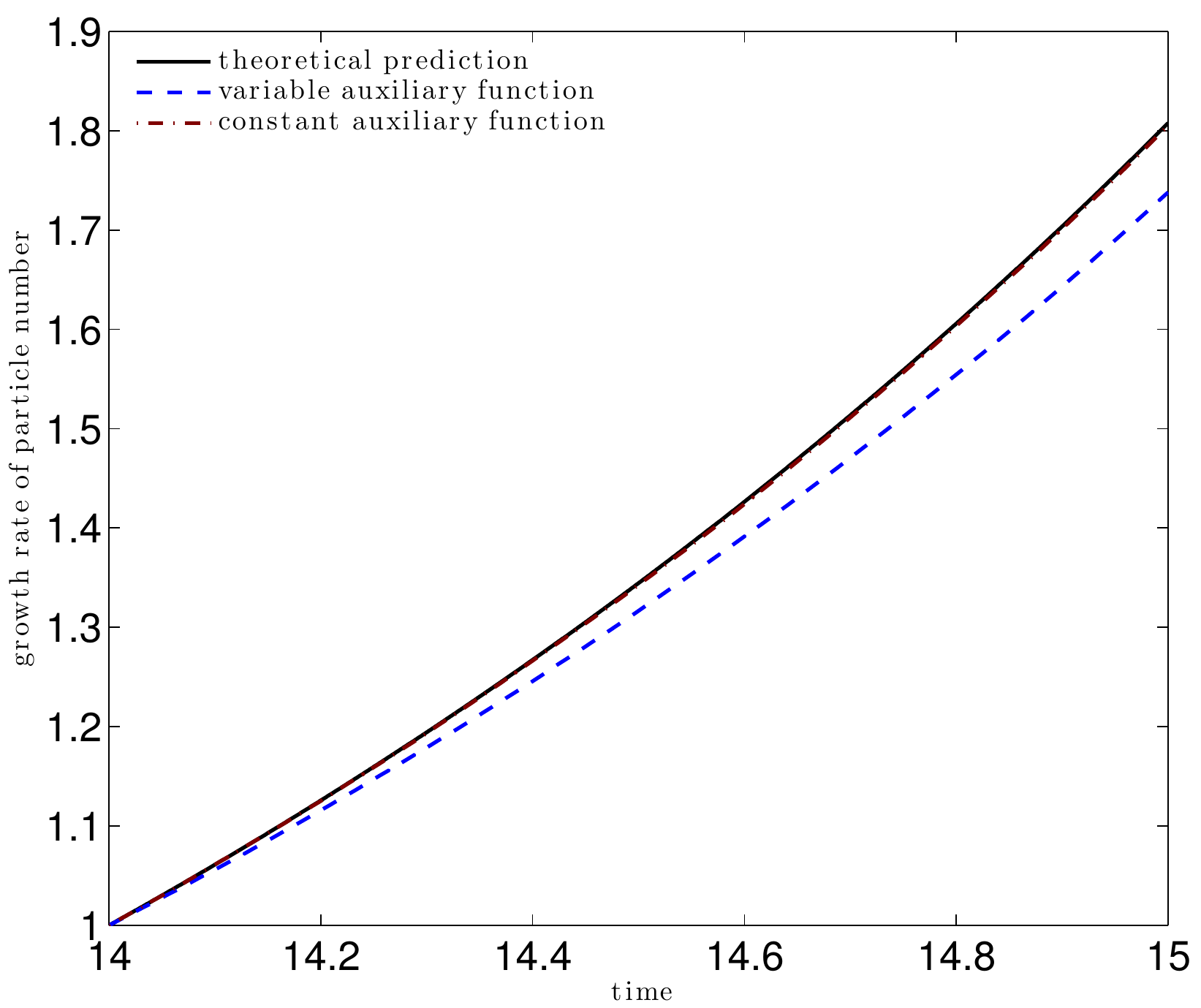}}
\subfigure[The twentieth
period.]{\includegraphics[width=0.49\textwidth,height=0.38\textwidth]{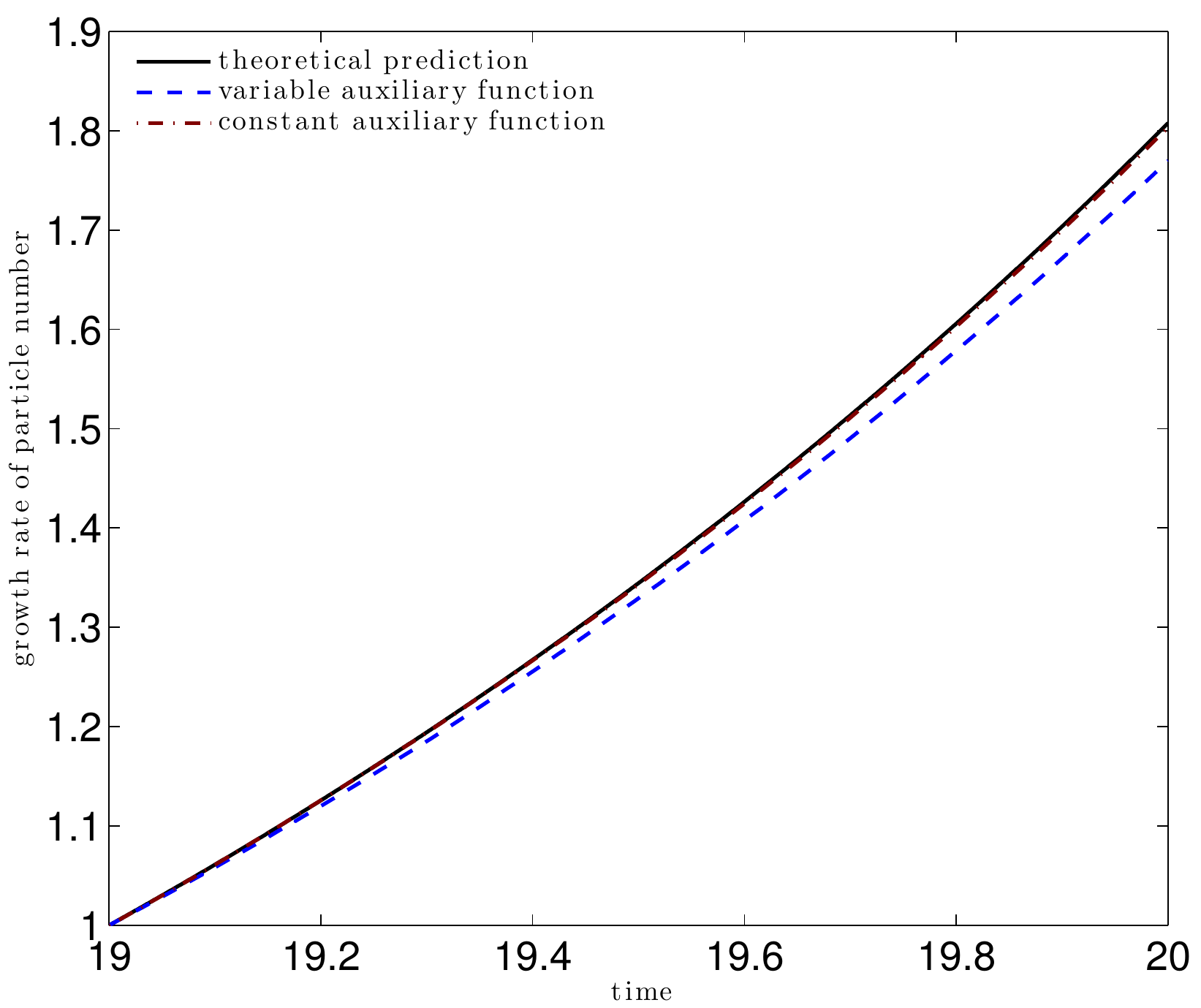}}
\caption{\small Partial reflection by the Gaussian barrier: Growth rates of particle number
within different annihilation periods for WBRW with $\Delta t=0.008$fs and $T_A=1$fs. The
curve of theoretical prediction can be described analytically by
$\me^{2\gamma_0t}$ when using a constant auxiliary function
$\gamma_0$. Here we set the constant auxiliary function
$\gamma_0=\check{\xi}$ and the variable one $\gamma(x)=\xi(x)$. }
\label{fig:ex1-np}
\end{figure}

\begin{figure}
\subfigure[Different $\gamma(x)$.]{\includegraphics[width=0.49\textwidth,height=0.38\textwidth]{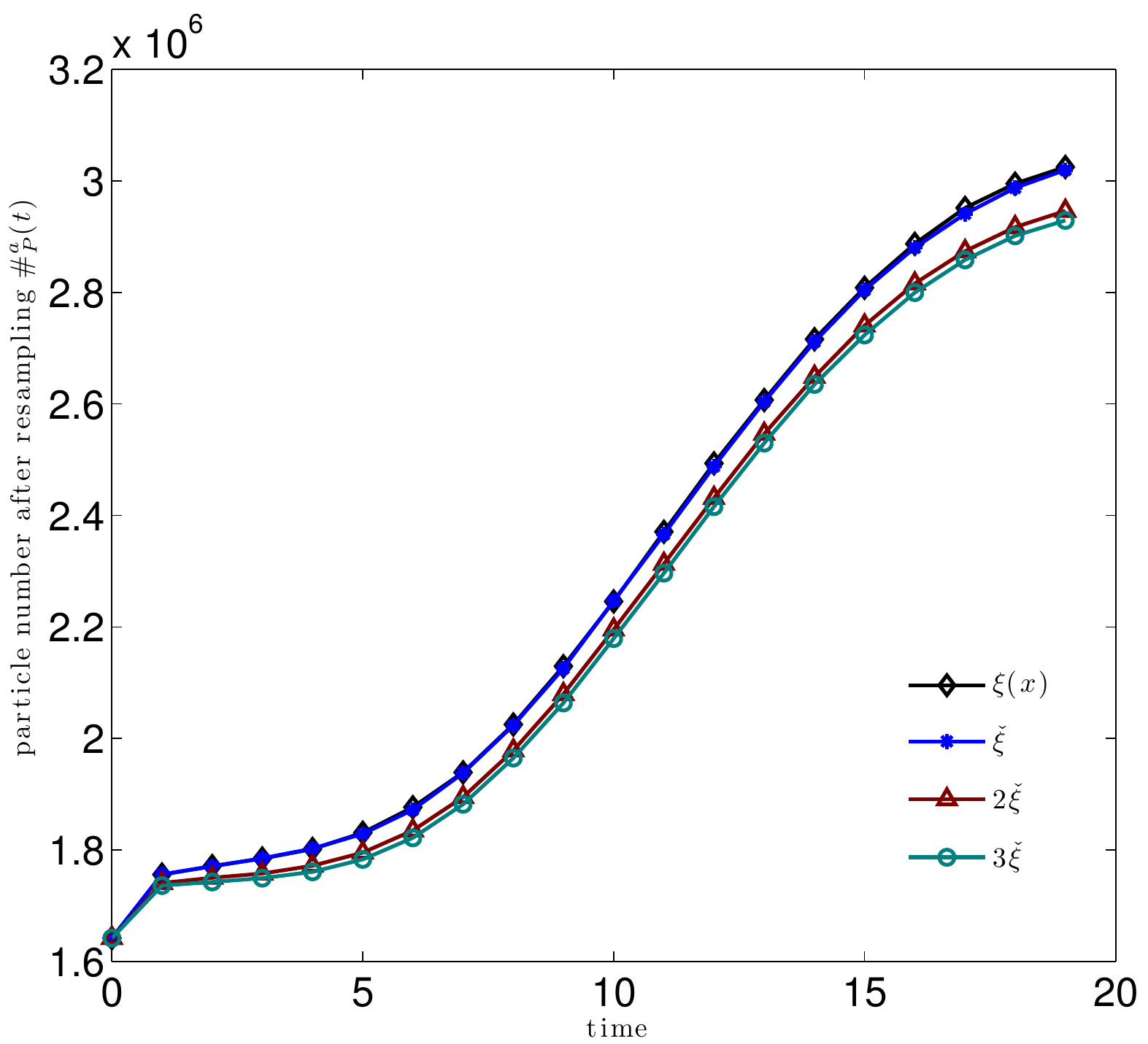}}
\subfigure[Different $T_A$.]{\label{fig:ex1-npaa-b}\includegraphics[width=0.49\textwidth,height=0.38\textwidth]{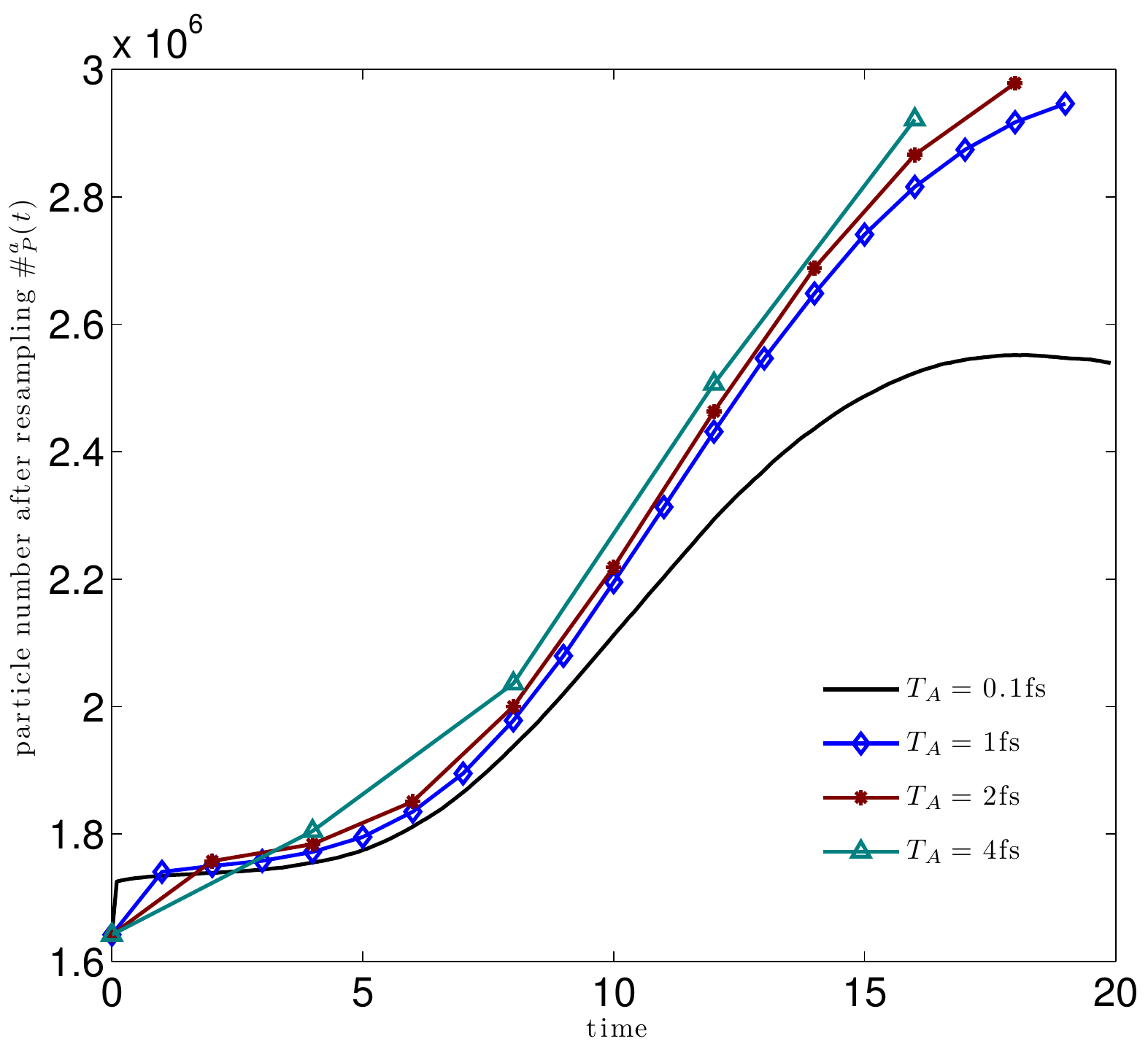}}
\caption{\small Partial reflection by the Gaussian barrier: Particle number after resampling (annihilation). The left plot shows the behavior for different auxiliary functions $\gamma(x)$ with the same annihilation period $T_A=1$fs. The right plot displays
the behavior for different annihilation periods with the same constant auxiliary function $\gamma_0=2\check{\xi}$.}
\label{fig:ex1-npaa}
\end{figure}

%\begin{equation}\label{eq:mesh}
%\begin{split}
%x_{i}&=x_{min}+(i-1/2) \Delta x,\quad \Delta {x}=\frac{x_{max}-x_{min}}{200}, \quad i=1,2,\cdots,200,\\
%k_{j}&=k_{min}+(j-1) \Delta k, \quad \Delta {k}=\frac{k_{max}-k_{min}
%}{400}, \quad j=1,2,\cdots,400.
%\end{split}
%\end{equation}

Once WBRW starts, the number of particles
increases exponentially with time, thus necessary annihilation
operations are required to make the simulation go on well, though
they are not inherited in the branching process from the theoretical
point of view. In this work, we do such annihilations at a constant
frequency, say $1/T_A$. That is, we divide equally the time interval
$[0,t_{fin}]$ into $n_A$ subintervals with the partition being
\[
0 = t^0 < t^1 < t^2 < \cdots < t^{n_A} = t_{fin}, \quad n_A={t_{fin}}/{T_A}.
\]
The annihilations occur exactly at the time instant $t^i$ for $1\leq
i\leq n_A-1$, at which the particle number decreases significantly
from $\#_P^b(t^i)$ to $\#_P^a(t^i)$, where $\#_P^b$ (resp. $\#_P^a$)
represents the particle number before (resp. after) the
annihilation. For convenience, we denote the particle number at
$t^0$ and $t^{n_A}$ by $\#_P^a(t^0)$ and $\#_P^b(t^{n_A})$,
respectively. In each time period $[t^{i-1},t^{i}]$, the particle
number increases from $\#_P^a(t^{i-1})$ to $\#_P^b(t^i)$ and
corresponding multiple is denoted by
\begin{equation}\label{eq:rate}
M_i = \#_P^b(t^i)/\#_P^a(t^{i-1}).
\end{equation}
For the $k$-truncated Wigner branching particle model with constant auxiliary function $\gamma(x)\equiv\gamma_0$,
it has been proved in Theorem~\ref{th:exp} that such increasing multiple only depends on the time increment $T_A$ and $\gamma_0$, which means
the same increasing multiple exists for each time period,
i.e., $M_i\equiv \me^{2\gamma_0T_A}$ for any $i\in\{1,2,\cdots,n_A\}$.

All the numerical results are obtained with our own Fortran implementations of WBRW, SEM and ASM on the computing platform: Dell Poweredge R820 with $4\times$ Intel Xeon processor E5-4620
(2.2 GHz, 16 MB Cache, 7.2 GT/s QPI Speed, 8 Cores, 16 Threads) and 256GB memory.
A fixed time step $\Delta t$ is applied and then the total number of time steps becomes
\[
n = t_{fin}/\Delta t.
\]
When the branching process evolves from $t_{j-1}=(j-1)\Delta t$ to $t_{j}=j \Delta t$ for $1\leq j\leq n$,
particle offspring will be generated.

\subsection{Gaussian barrier scattering}

Two Gaussian barrier scattering experiments are conducted in 2D
phase space. The first experiment is exactly the same as that
adopted in \cite{ShaoSellier2015}, while the only change for the
second one is the barrier height is increased to $1.3$eV. The
readers are referred to \cite{ShaoSellier2015} for the details on
the problem setting. As we pointed out earlier,  both the
$k$-truncated (see Eq.~\eqref{eq.k_truncated_Wigner}, the model
parameter is $\Delta y$) and $y$-truncated (see
Eq.~\eqref{eq.discrete_Wigner}, the model parameter is $\Delta k$)
branching particle models can be regarded as approximations of the
same Wigner equation in the unbounded domain, and thus comparable
results are expected on the same footing because both Gaussian
wavepacket and Gaussian barrier possess a very nice localized
structure. Hence we only report numerical results for the
$k$-truncated model and those for the $y$-truncated model can be
found in \cite{ShaoSellier2015} as well as in an early version of
this work \cite{ShaoXiong2016}. The initial particle number is fixed
to be $\#_P^a(t^0) = 1641810$, and the reference solutions are
obtained by SEM, the spectral accuracy of which was well
demonstrated in \cite{ShaoLuCai2011,ShaoSellier2015}.

\begin{table}
  \centering
  \caption{\small Partial reflection by the Gaussian barrier: Numerical data for WBRW.
The errors in the second, third and fourth columns are calculated at
the final time $t_{fin}=20$fs. The particle numbers in the fifth and
sixth columns and the running CPU time in the last column are
measured in million and minutes, respectively. While using constant
auxiliary function $\gamma(x)\equiv\gamma_0$ the increasing multiple
of particle number within an annihilation period is
$\me^{2\gamma_0T_A}$. Three kinds of constant auxiliary functions,
$\gamma_0=\check{\xi},2\check{\xi},3\check{\xi}$, are tested, where
$\check{\xi}=\max_{x\in\fx}\{\xi(x)\}\approx$ 2.96E-01.
}\label{tab:eg1-1}
 \newsavebox{\tablebox}
 \begin{lrbox}{\tablebox}
  \begin{tabular}{ccccccccc}
\hline\hline
$\gamma(x)$ & $\text{err}_{wf}$ & $\text{err}_{sm}$ & $\text{err}_{mm}$ & $\check{\#}_P^b$ & $\check{\#}_P^a$ & $\overline{M}$ & $\me^{2\gamma_0T_A}$ & Time\\
\hline
\multicolumn{9}{c}{$\Delta t=0.008$fs, $T_A=1$fs}\\
\hline
$\xi(x)$ &
9.08E-02 &
3.08E-02 &
3.20E-02 &
5.34 &
3.02 &
1.77 &
-- &
286.28\\
$\check{\xi}$ &
9.01E-02 &
3.16E-02 &
3.05E-02 &
5.44 &
3.02 &
1.80 &
1.81 &
291.45\\
$2\check{\xi}$ &
7.95E-02 &
2.74E-02 &
2.46E-02 &
9.62 &
2.95 &
3.26 &
3.27 &
297.72\\
$3\check{\xi}$ &
7.51E-02 &
2.52E-02&
1.97E-02&
17.25 &
2.93 &
5.89 &
5.91 &
314.83\\
\hline
\multicolumn{9}{c}{$\Delta t=1$fs, $T_A=1$fs}\\
\hline
$\xi(x)$ &
9.25E-02 &
3.21E-02 &
3.31E-02 &
5.35 &
3.02 &
1.77 &
-- &
3.38 \\
$\check{\xi}$ &
8.98E-02 &
3.32E-02 &
2.73E-02 &
5.44 &
3.02 &
1.80 &
1.81 &
3.37\\
$2\check{\xi}$ &
7.94E-02 &
2.68E-02 &
2.21E-02 &
9.61 &
2.95 &
3.26 &
3.27 &
4.57 \\
$3\check{\xi}$ &
7.55E-02 &
2.51E-02 &
1.99E-02 &
17.26 &
2.93 &
5.89 &
5.91 &
6.55\\
\hline
\multicolumn{9}{c}{$\Delta t=2$fs, $T_A=2$fs}\\
\hline
$\xi(x)$ &
9.43E-02 &
3.18E-02 &
3.75E-02 &
9.72 &
3.12 &
3.13 &
-- &
2.22\\
$\check{\xi}$ &
9.00E-02 &
3.45E-02 &
3.43E-02 &
10.07 &
3.10 &
3.25 &
3.27 &
2.95\\
$2\check{\xi}$ &
6.32E-02 &
2.70E-02 &
2.36E-02 &
31.62 &
2.98 &
10.62 &
10.68 &
5.50\\
$3\check{\xi}$ &
5.48E-02 &
2.45E-02 &
2.30E-02 &
102.25 &
2.94 &
34.71 &
34.88 &
14.63 \\
\hline
\multicolumn{9}{c}{$\Delta t=4$fs, $T_A=4$fs}\\
\hline
$\xi(x)$ &
1.39E-01 &
4.63E-02 &
4.66E-02 &
30.69 &
3.26 &
9.72 &
-- &
2.25 \\
$\check{\xi}$ &
1.27E-01 &
4.93E-02 &
4.74E-02 &
33.19 &
3.21 &
10.46 &
10.68 &
2.63 \\
$2\check{\xi}$ &
6.69E-02 &
2.83E-02 &
2.79E-02 &
326.99 &
2.92 &
112.65 &
113.98 &
25.10\\
\hline
\multicolumn{9}{c}{$\Delta t=0.1$fs, $T_A=0.1$fs}\\
\hline
$\xi(x)$ &
2.86E-01 &
8.12E-02 &
6.82E-02 &
2.73 &
2.58 &
1.06 &
-- &
23.67\\
$\check{\xi}$ &
2.87E-01 &
8.02E-02 &
6.53E-02 &
2.73 &
2.57 &
1.06 &
1.06 &
23.82\\
$2\check{\xi}$ &
2.84E-01 &
8.01E-02 &
6.61E-02 &
2.87 &
2.55 &
1.13 &
1.13 &
24.07 \\
$3\check{\xi}$ &
2.84E-01 &
8.10E-02 &
6.60E-02 &
3.04 &
2.55 &
1.19 &
1.19 &
24.83\\
\hline
\hline
  \end{tabular}
\end{lrbox}
\scalebox{0.88}{\usebox{\tablebox}}
\end{table}

In general, the calculation of the normalizing factor $\xi(x)$ in Eq.~\eqref{eq:normalizing} and  sampling from $V_{w}^{+}(x, k) / \xi(x)$ can be realized simultaneously, say, we can calculate the normalization factor through sampling. The Gaussian barrier potential reads
\begin{equation}\label{eq:gp}
V(x)=H_{B} \exp{\left[-\frac{(x-x_{B})^{2}}{2}\right]},
\end{equation}
where $H_B$ and $x_B$ denote the barrier height and the barrier center,
respectively, and the explicit expression of
corresponding Wigner kernel is
\begin{equation}\label{eq:gp_vw}
V_{w}(x, k)=\frac{2H_{B}}{\hbar} \sqrt{\frac{2}{\pi}} \me^{-2k^2} \sin(2k(x-x_{B})).
\end{equation}
It can be easily seen here that $\sqrt{\frac{2}{\pi}} \me^{-2k^2}$ in Eq.~\eqref{eq:gp_vw} is the probability density of the normal distribution $\mathcal{N}(0, 1/2)$, with which we can calculate the integral by the rejection sampling. In actual simulations, the number of samples are chosen as $2\times 10^8$ for each $x$, and the numerical $\xi(x)$ is shown in Fig.~\ref{fig:xi} for $H_B=0.3$eV and $x_B=30$nm. We find there that the normalization factor has a sharp decrease around $x=30$, whereas it is very flat outside the neighborhood of $x=30$.

\begin{figure}
\subfigure[$t=5$.]{\includegraphics[width=0.49\textwidth,height=0.27\textwidth]{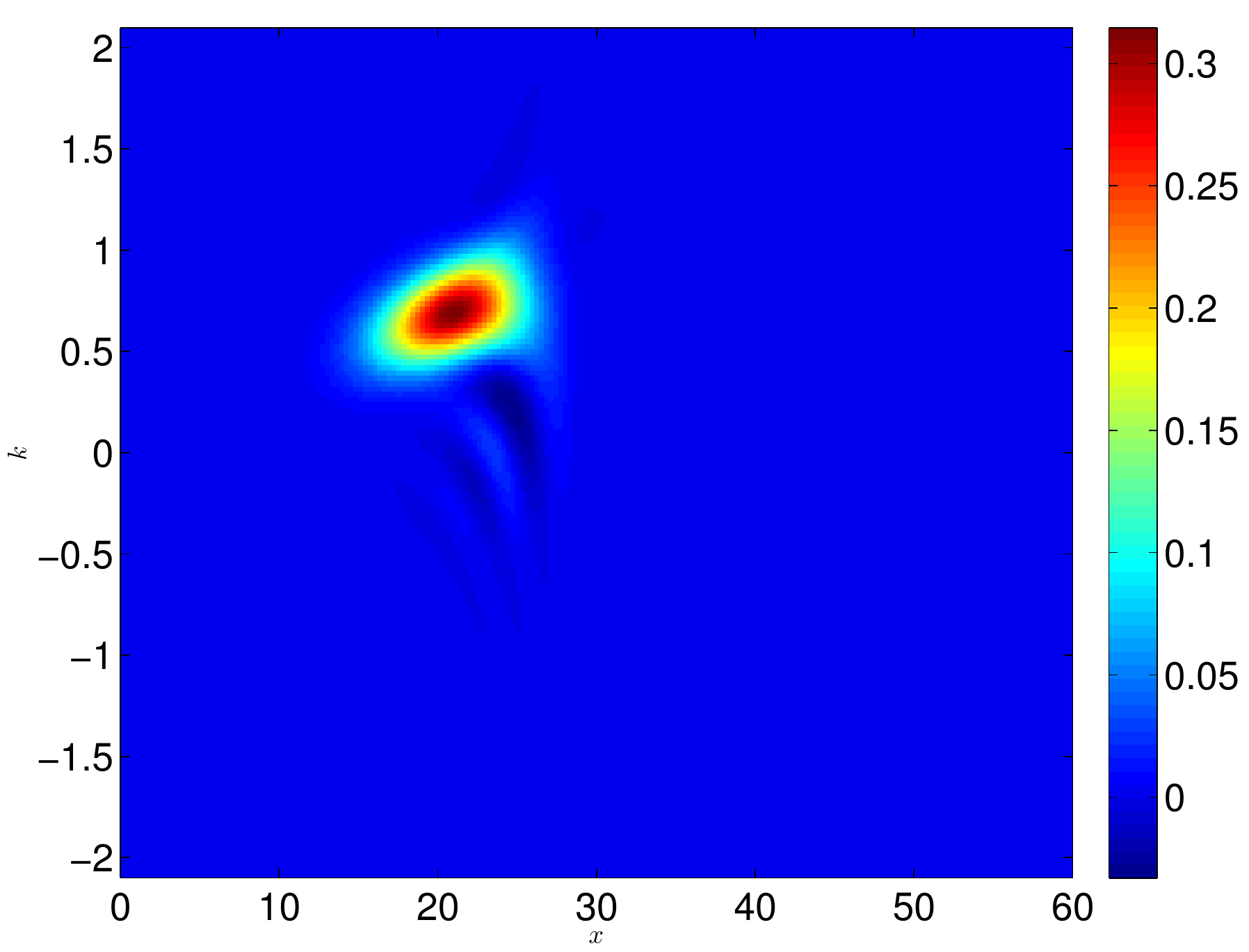}{\includegraphics[width=0.49\textwidth,height=0.27\textwidth]{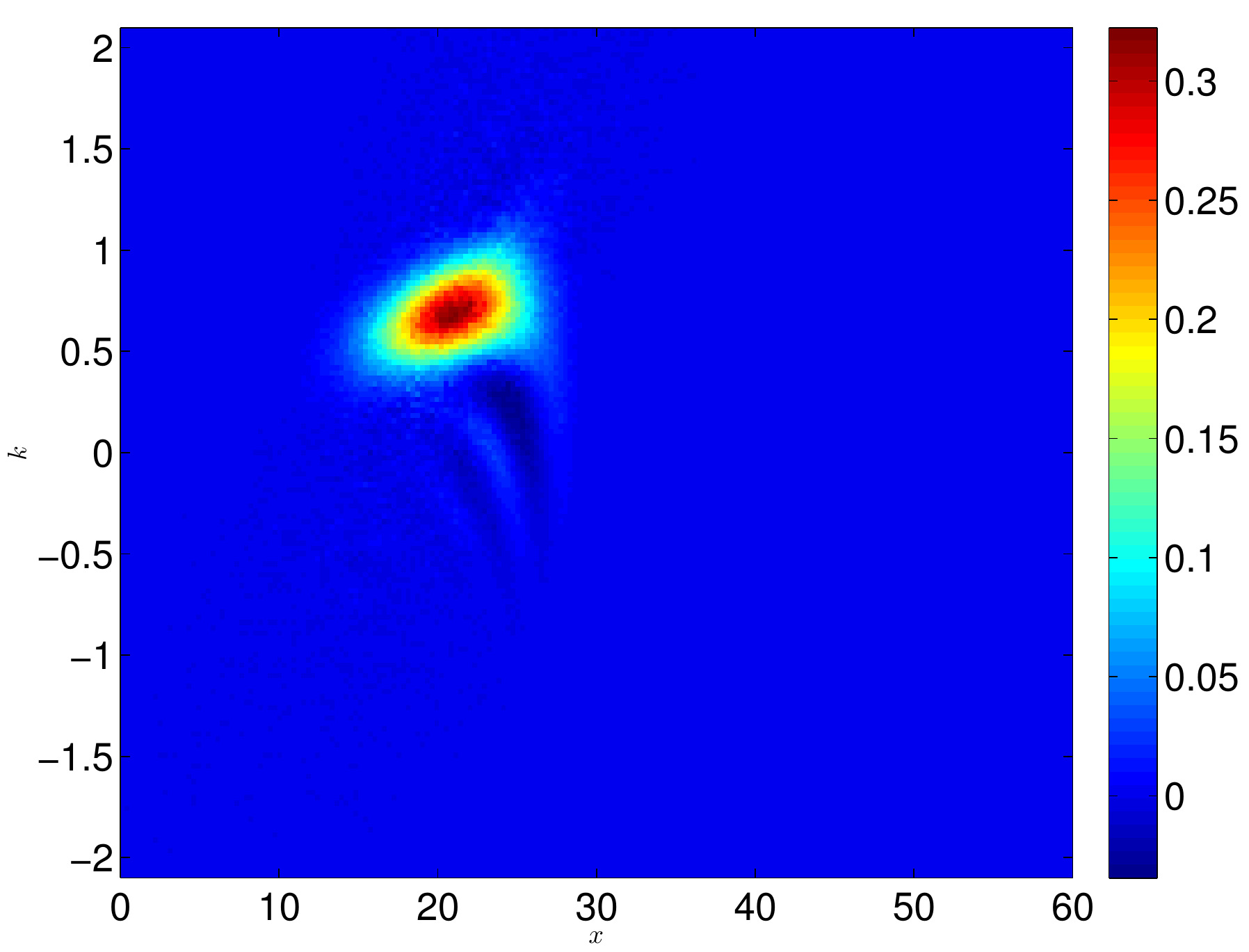}}}
\subfigure[$t=10$.]{\includegraphics[width=0.49\textwidth,height=0.27\textwidth]{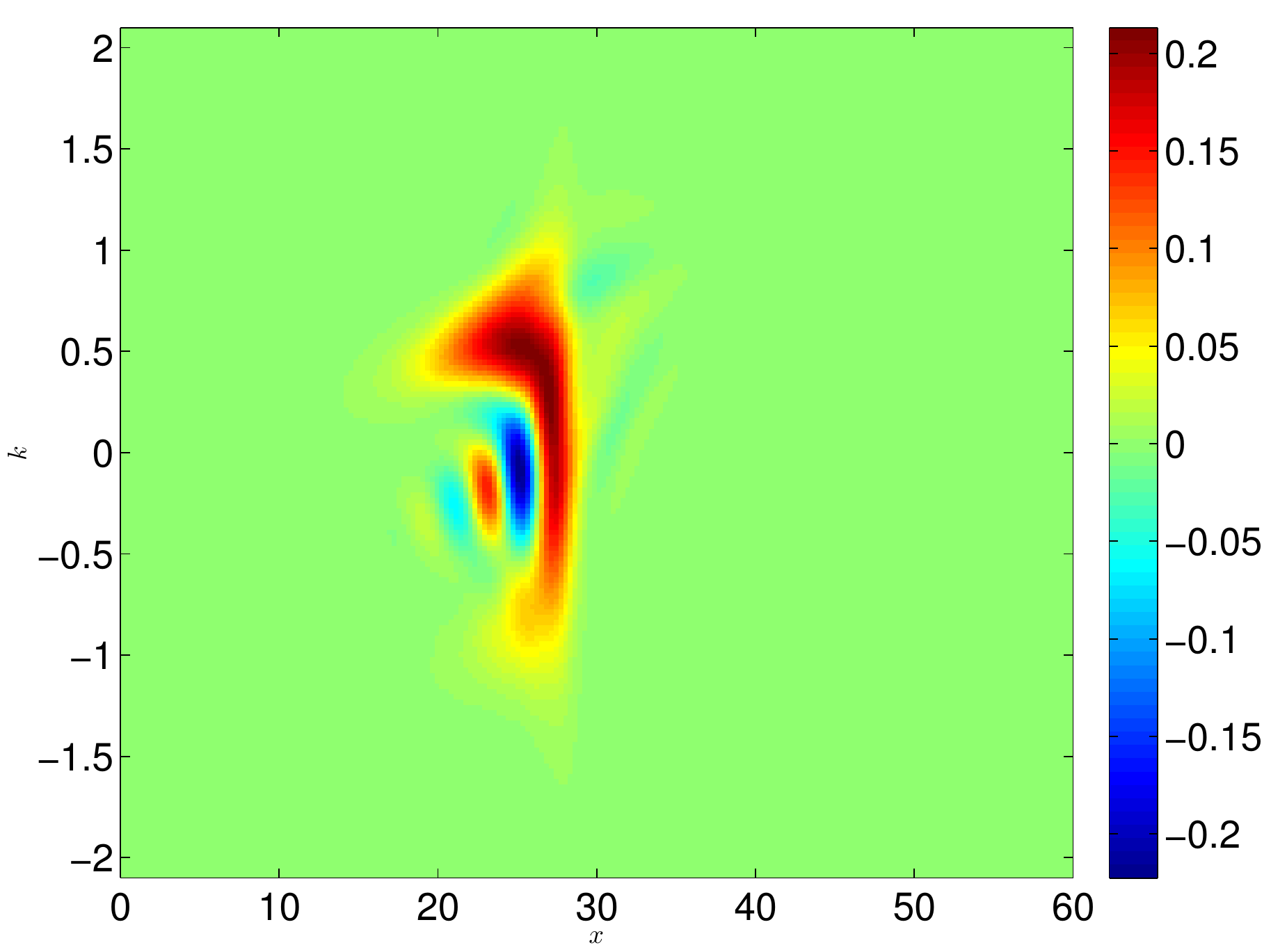}{\includegraphics[width=0.49\textwidth,height=0.27\textwidth]{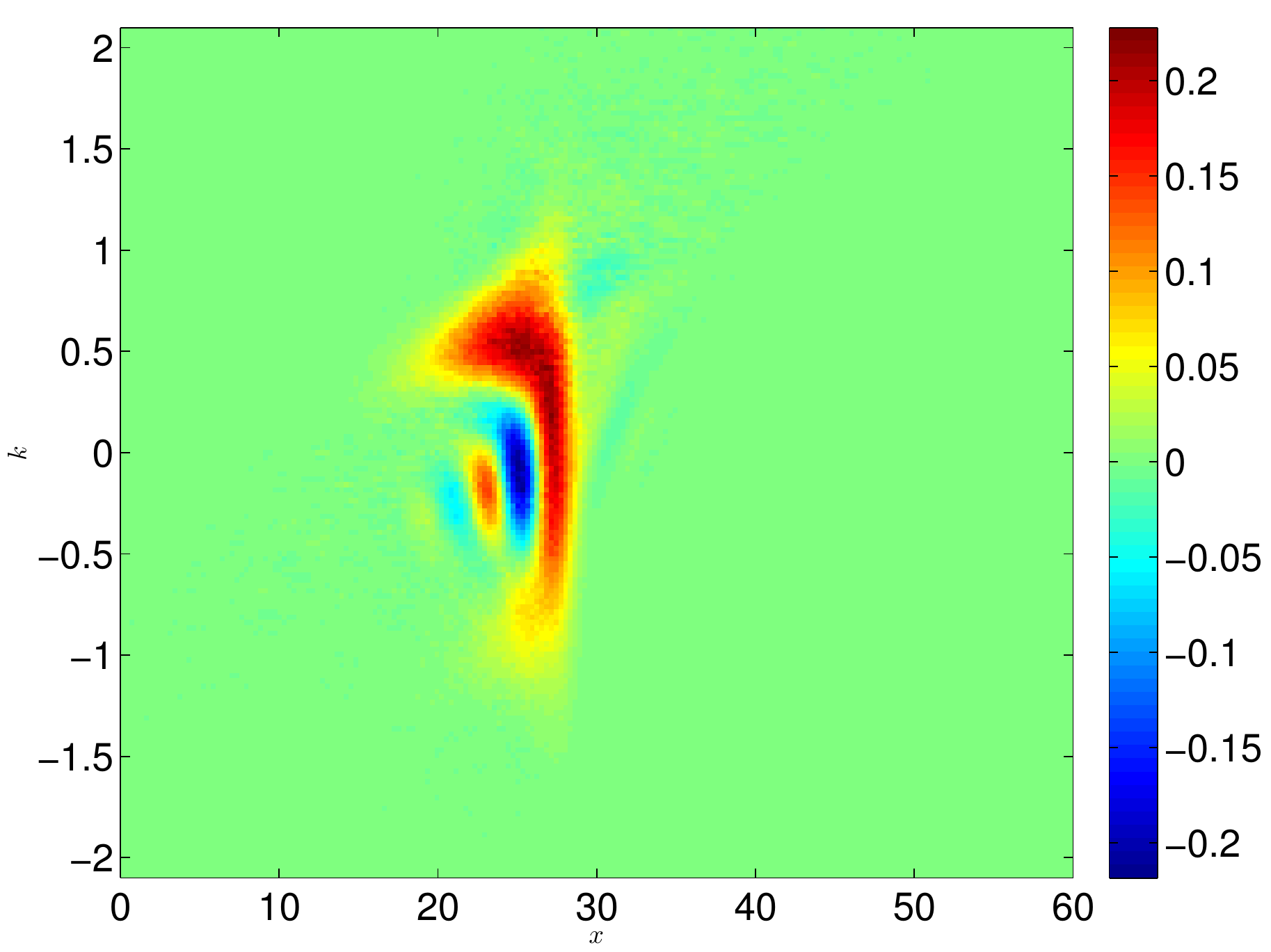}}}
\subfigure[$t=15$.]{\includegraphics[width=0.49\textwidth,height=0.27\textwidth]{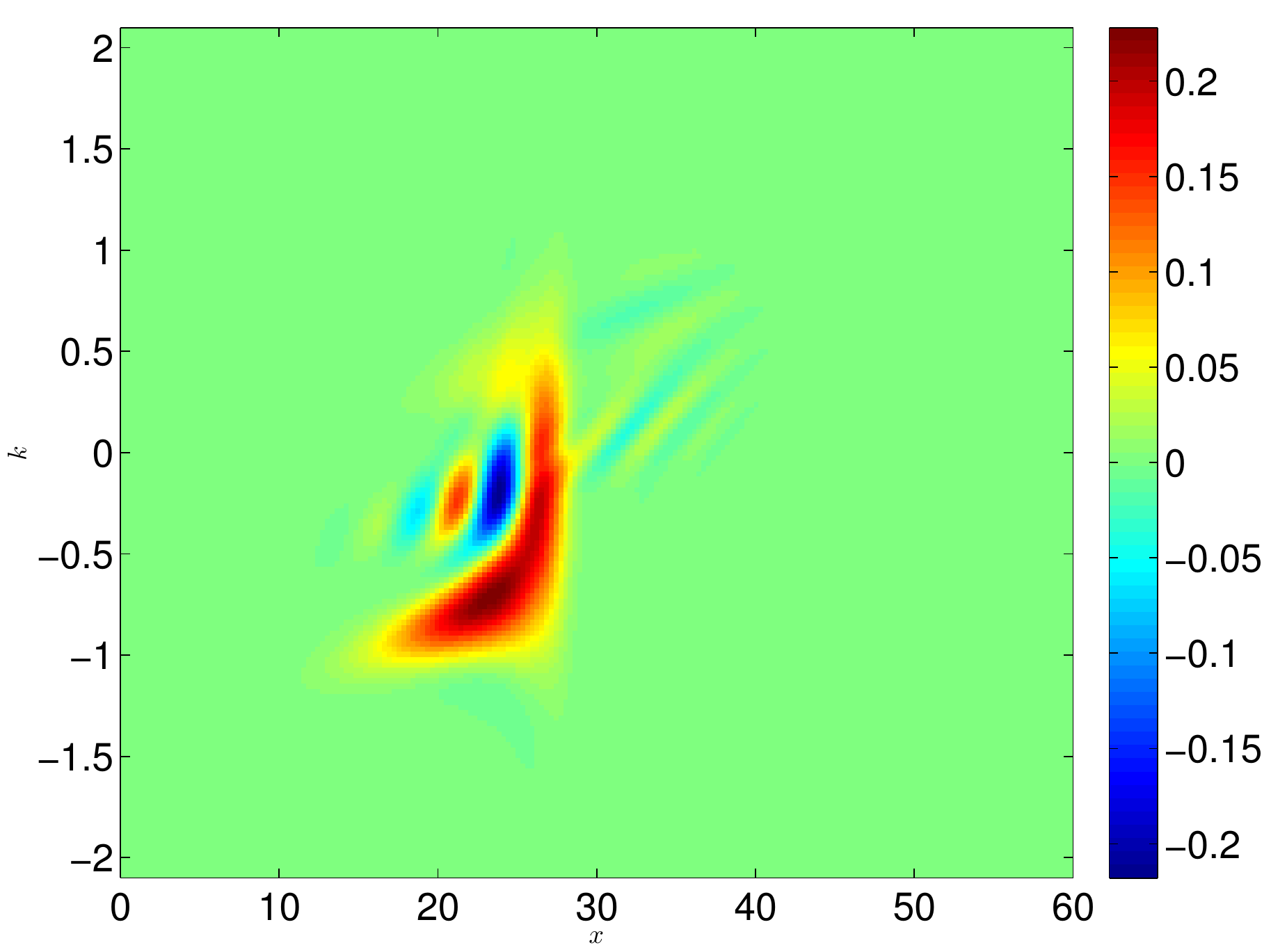}{\includegraphics[width=0.49\textwidth,height=0.27\textwidth]{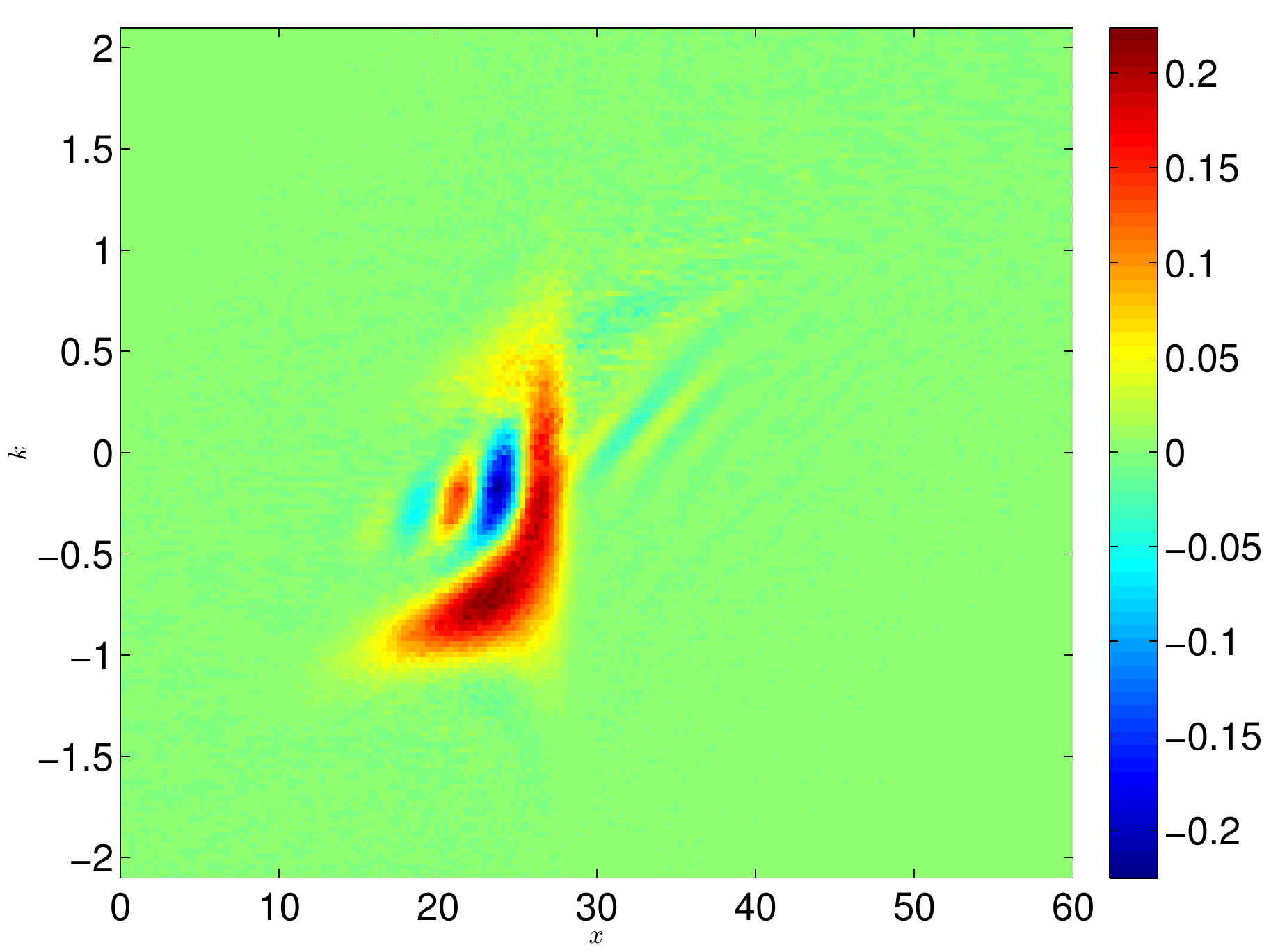}}}
\subfigure[$t=20$.]{\includegraphics[width=0.49\textwidth,height=0.27\textwidth]{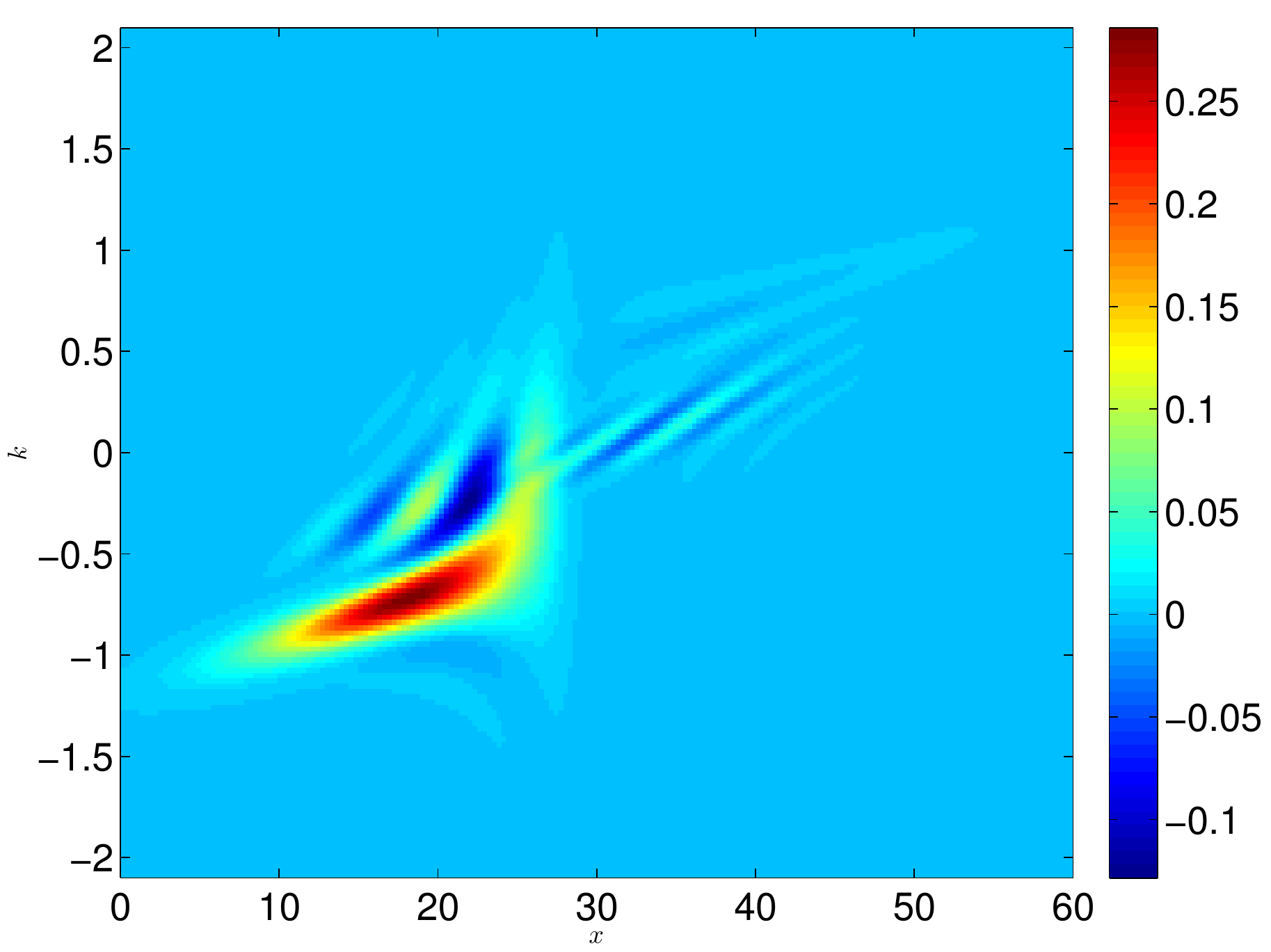}{\includegraphics[width=0.49\textwidth,height=0.27\textwidth]{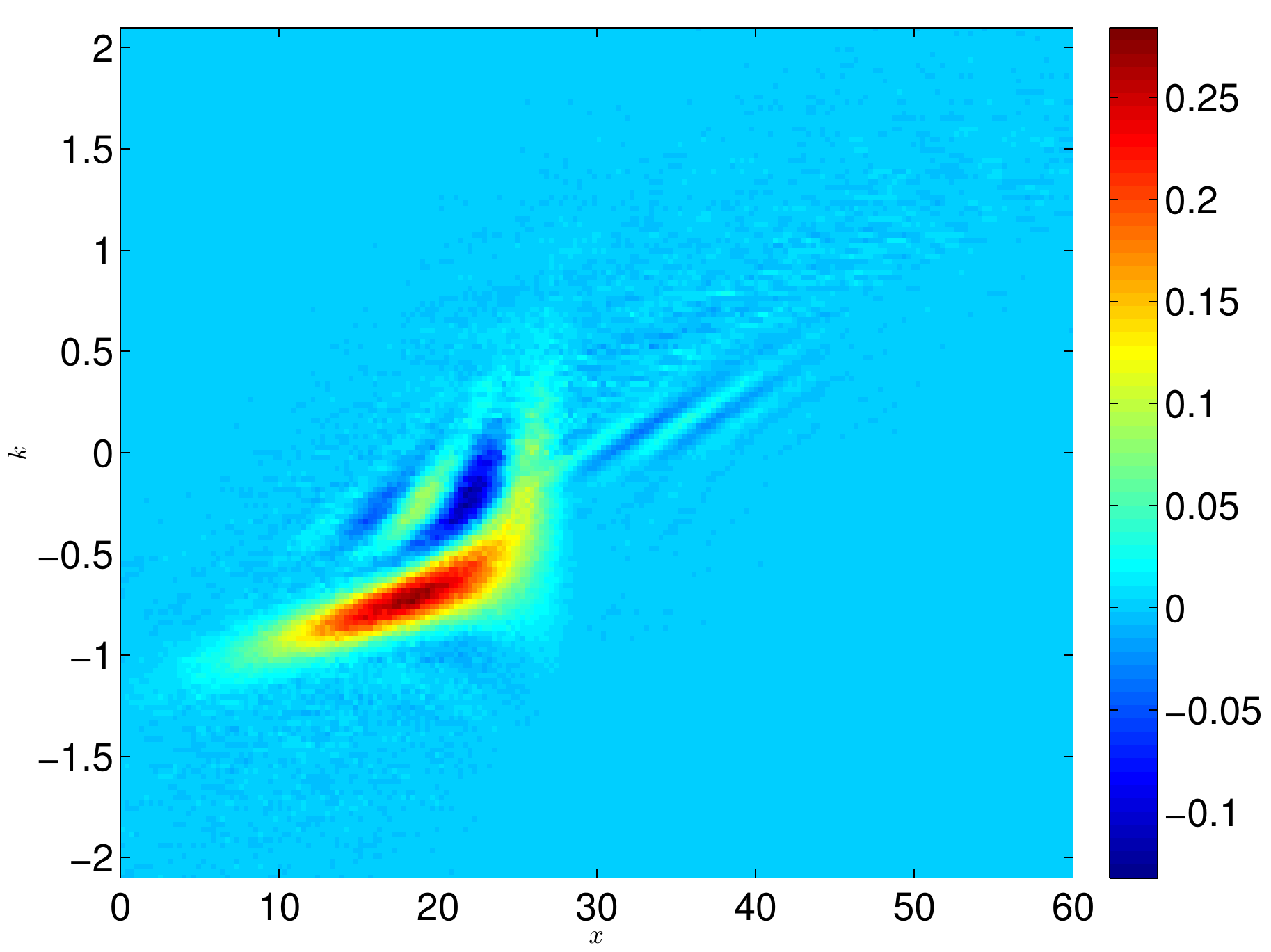}}}
\caption{\small Total reflection by the Gaussian barrier: Numerical Wigner functions at different time instants $t=5,20,15,20$fs.
The reference solution by SEM is displayed in the left-hand-side column, while the right-hand-side column shows the numerical solution obtained by WBRW with the auxiliary function $\gamma(x)=2\check{\xi}$ as well as $\Delta t=1$fs and $T_A=1$fs.}
\label{fig:ex2-wf}
\end{figure}

$\bullet$ \textbf{Experiment 1} To be convenient for comparison, we
first take the same experiment as  utilized before in
\cite{ShaoSellier2015}, in which the barrier height is set to be
$H=0.3$eV so that the Gaussian wavepacket will be partially
reflected. Such partial reflection is clearly shown in
Fig.~\ref{fig:ex1-wf}. Five groups of tests with different time
steps and annihilation frequencies are performed and related data
are displayed in Table \ref{tab:eg1-1}, from which we are able to
find several observations below concerning the accuracy.

\begin{description}
%\item[(1)] When $\Delta t=0.008$fs and $T_A=1$fs, our implementation of the $k$-truncated branching particle model with $\gamma(x)=z(x)$ gives more accurate solution than that shown in Fig.~7 of \cite{ShaoSellier2015}.
\item[(1)] The idea of choosing constant auxiliary function $\gamma(x)\equiv\gamma_0$ works very well. As we expected, with the same $\Delta t$ and $T_A$,
the larger value $\gamma_0$ takes, the more accurate solution we
obtain, though more running time it spends. Interestingly, both
accuracy and efficiency of the model with $\gamma(x)=\xi(x)$ are
very close to those of the model with $\gamma(x)\equiv\check{\xi}$, where
$\check{\xi}=\max_{x\in\fx} \xi(x)$. This also justifies the
proposed mathematical framework, the algorithm of branching process
as well as the implementation in some sense.
\item[(2)] Highly frequent annihilation operations,
e.g., $T_A=0.1$fs, destroy the accuracy,
even larger constant auxiliary function cannot save it.
But this does not mean a low annihilation frequency should be appreciated.
Actually, when using $T_A=4$fs, the accuracy becomes worse than that using
$T_A=1$fs or 2fs, which may be due to the accumulated numerical errors, such as the bias caused by the resampling. That is, as we mentioned before,
the annihilation adopted here is nothing but a kind of resampling according to the histogram, and thus possibly  cause some random noises due to its discontinuous nature.
\item[(3)] While using the same annihilation frequency, say $T_A=1$fs,
smaller time step, e.g., $\Delta t=0.008$fs, cannot improve the accuracy, as predicted by Theorem~\ref{th:G(t)}. In fact,
the accuracy with $\Delta t=0.008$fs is almost identical to that with $\Delta t=1$fs.
But the former takes much more running time than the latter. Moreover,
too small time steps will significantly reduce the probability of branching,
which is crucial to capture the quantum information in stochastic Wigner simulations,
so that nearly all particles do field-less travel in phase space.
This point has been also mentioned in \cite{ShaoSellier2015} when choosing the time step.
\end{description}

\begin{figure}
\vspace{-2cm}
\subfigure[$t=5$.]{\includegraphics[width=0.49\textwidth,height=0.35\textwidth]{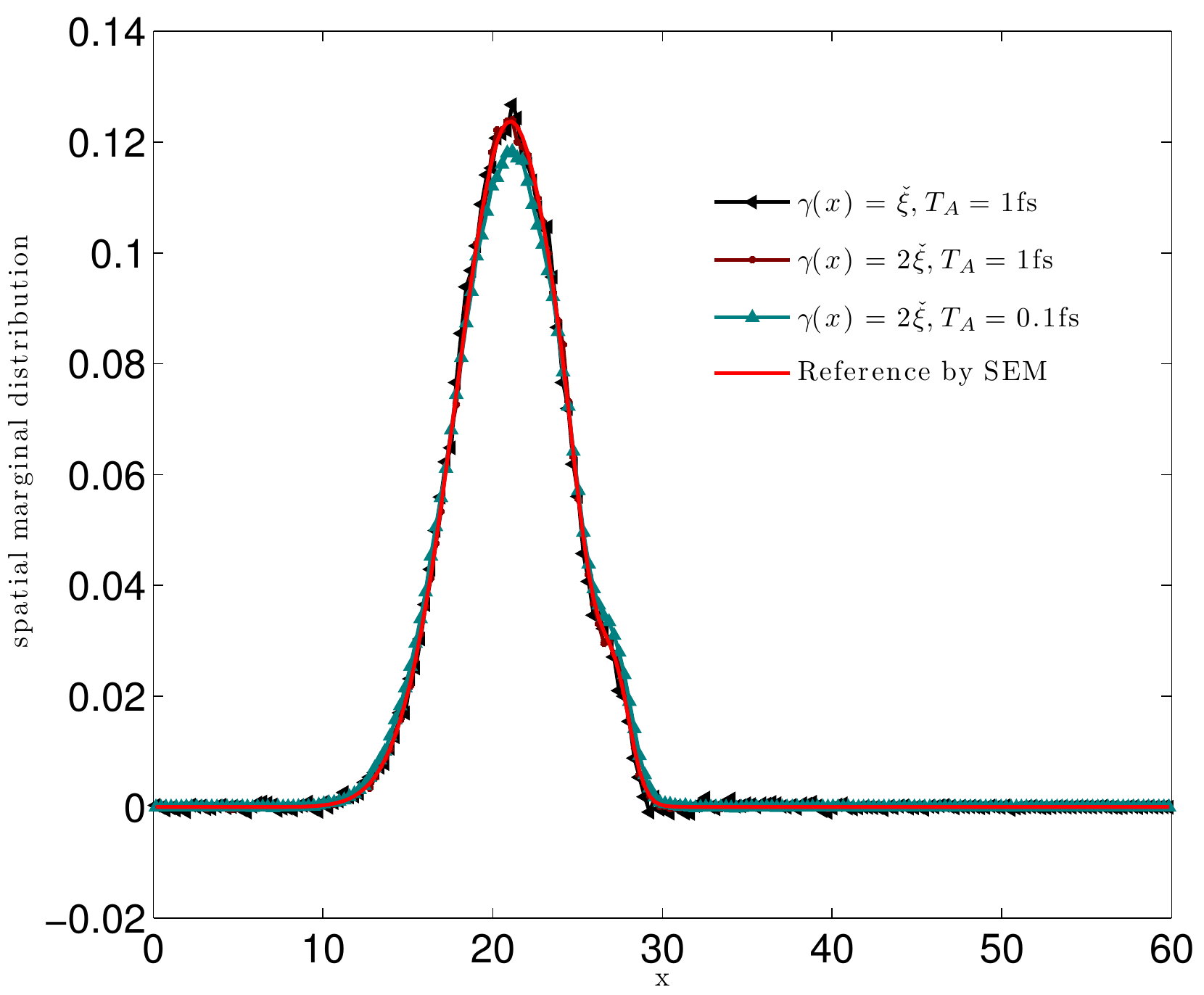}
\includegraphics[width=0.49\textwidth,height=0.35\textwidth]{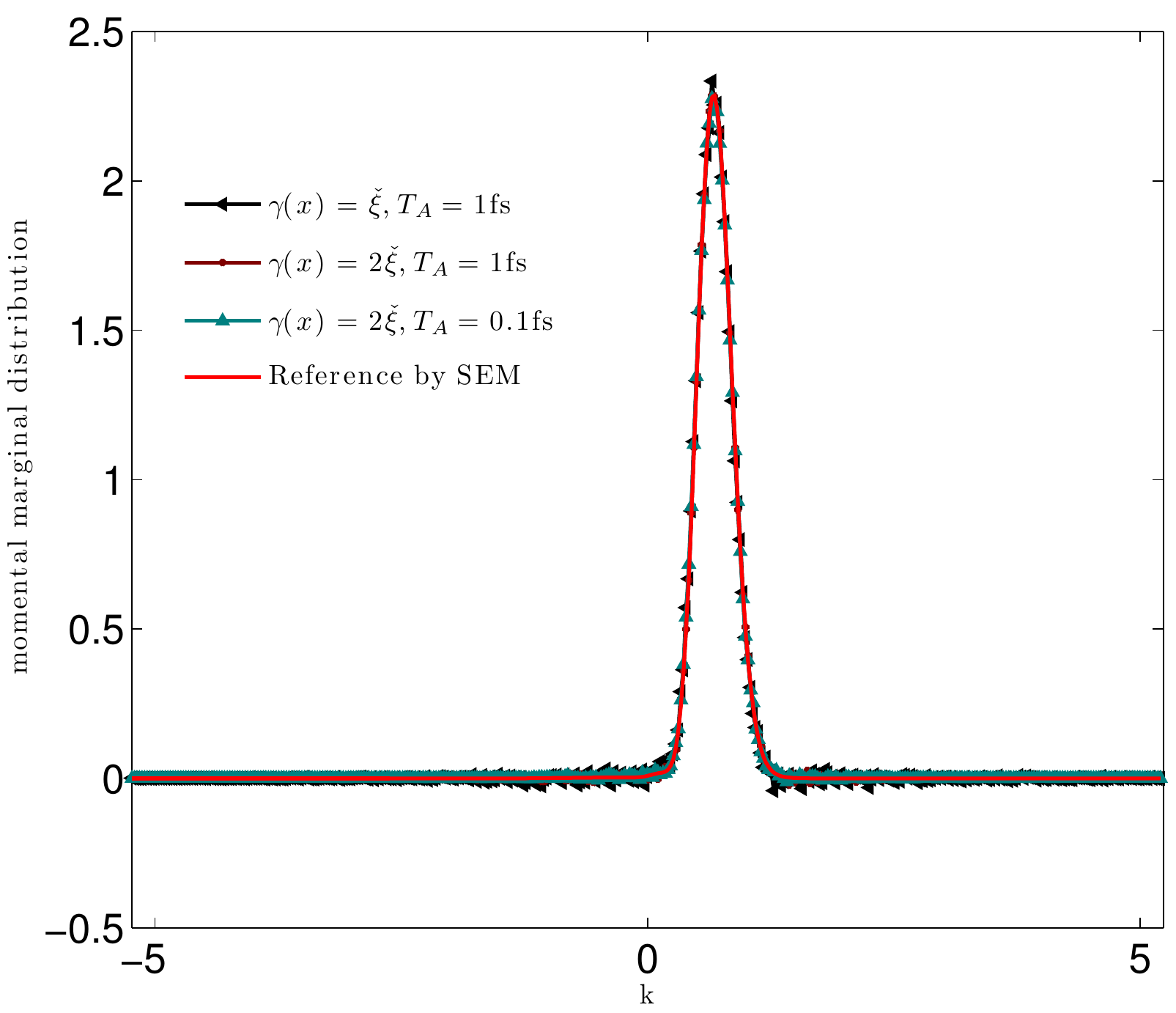}}\\
\subfigure[$t=10$.]{\includegraphics[width=0.49\textwidth,height=0.35\textwidth]{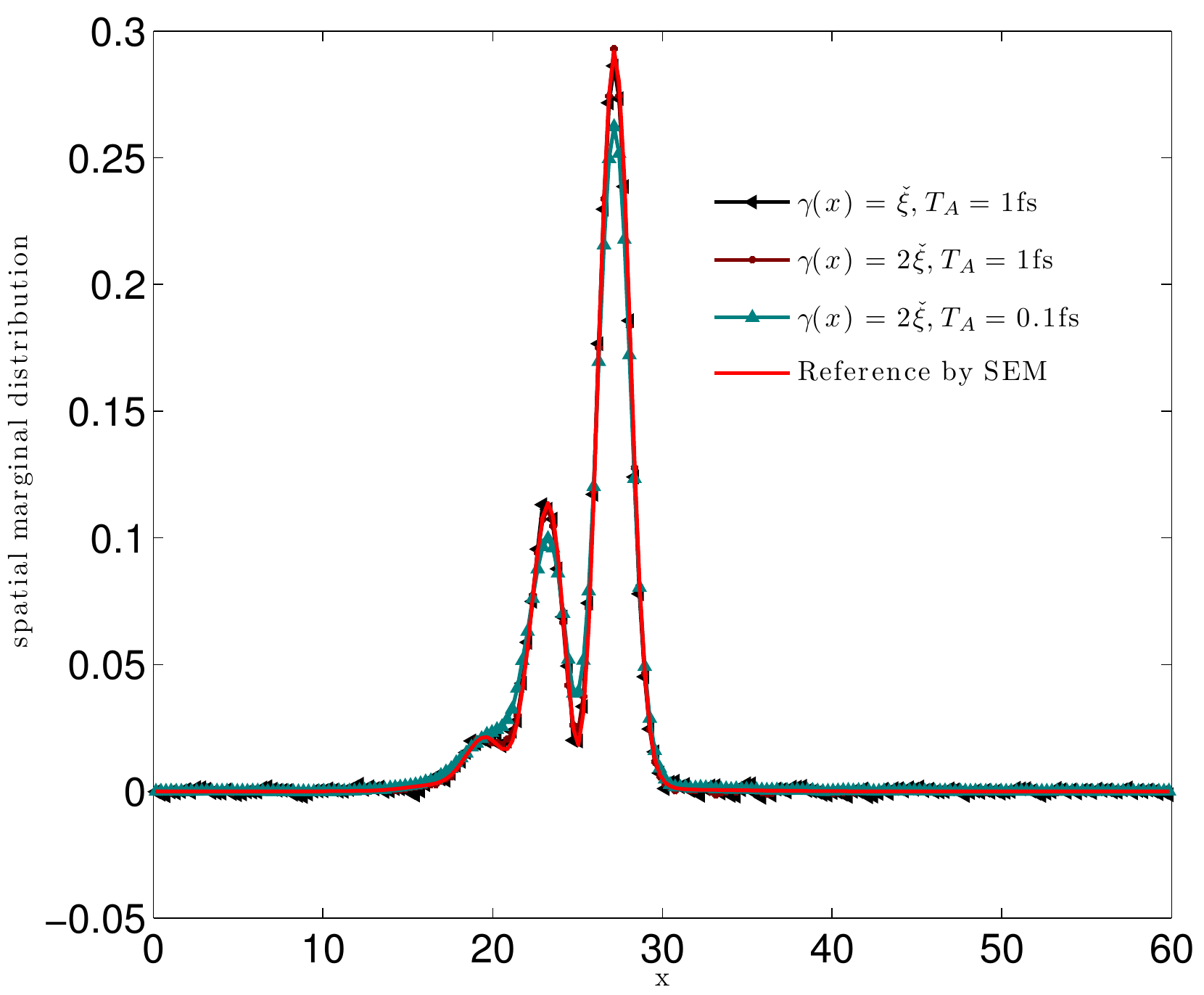}
\includegraphics[width=0.49\textwidth,height=0.35\textwidth]{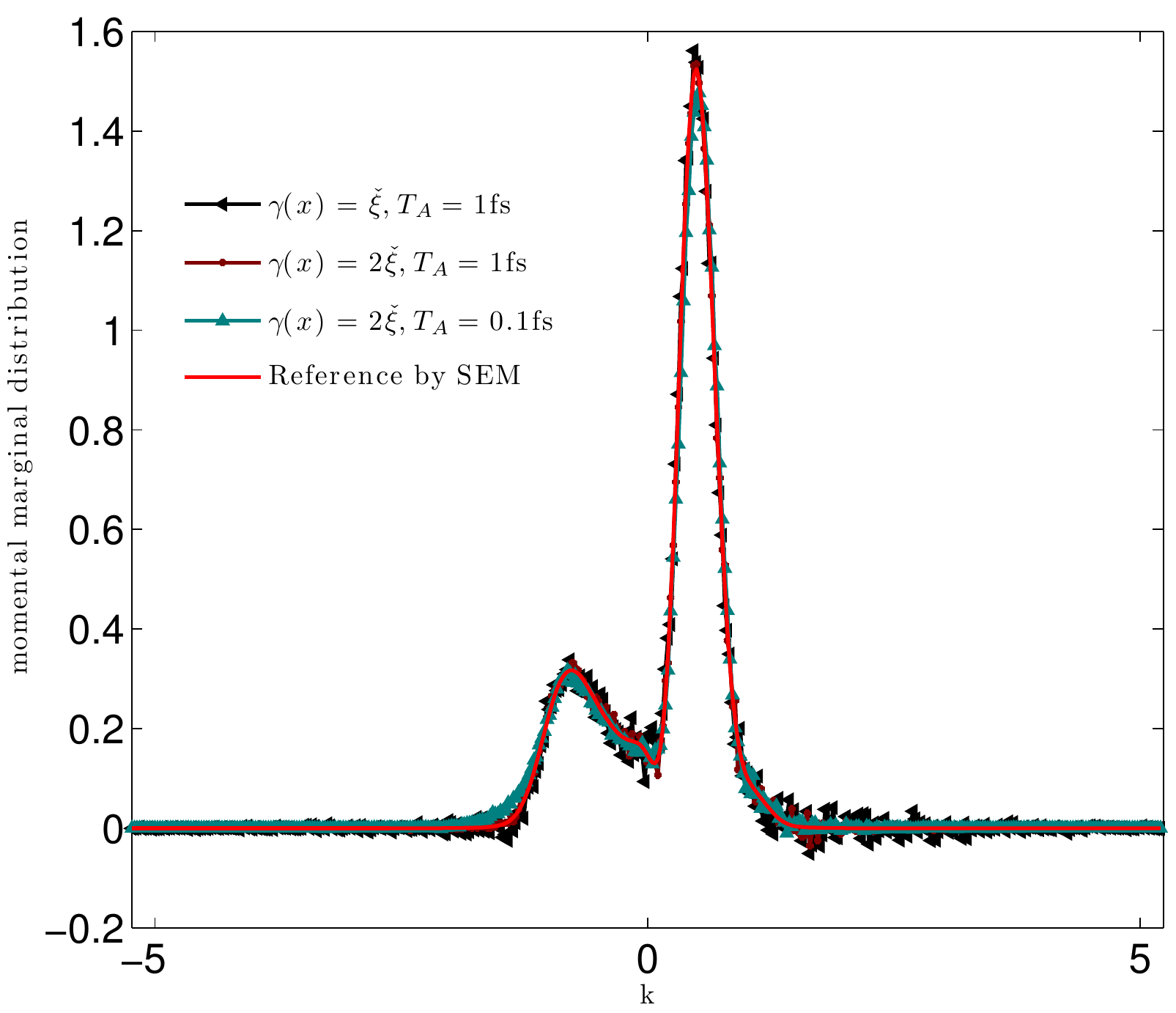}}
\subfigure[$t=15$.]{\includegraphics[width=0.49\textwidth,height=0.35\textwidth]{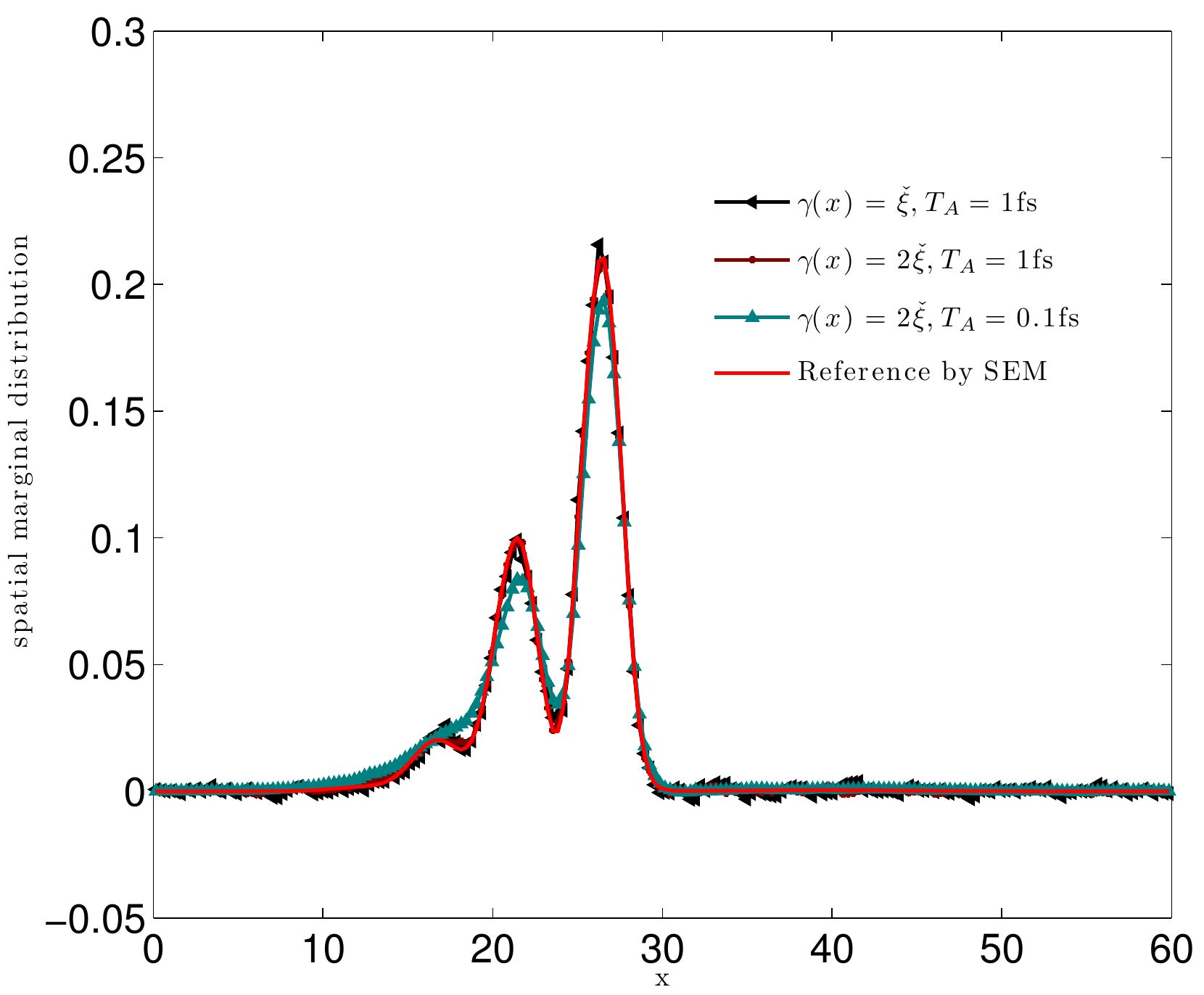}
\includegraphics[width=0.49\textwidth,height=0.35\textwidth]{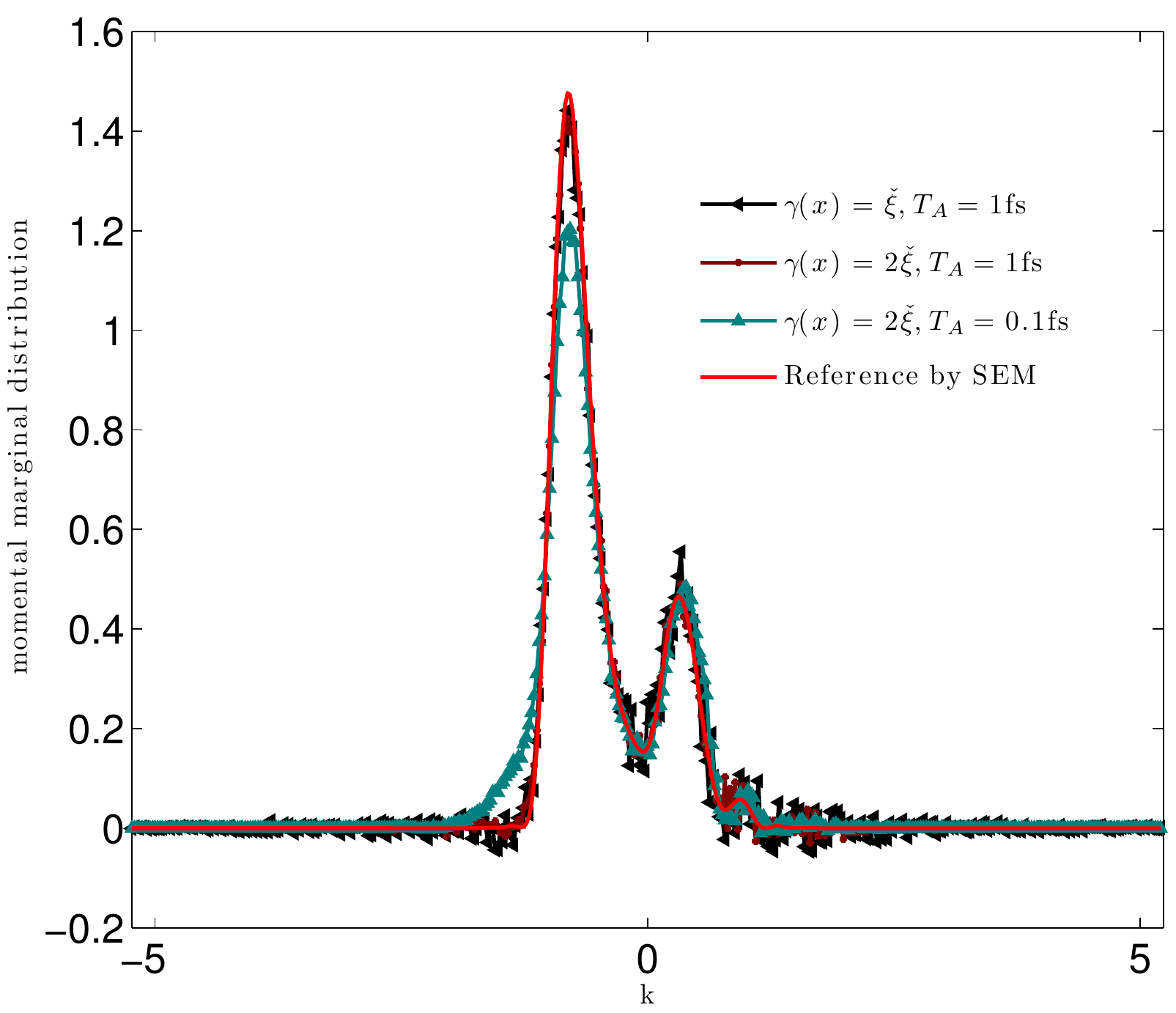}}\\
\subfigure[$t=20$.]{\includegraphics[width=0.49\textwidth,height=0.35\textwidth]{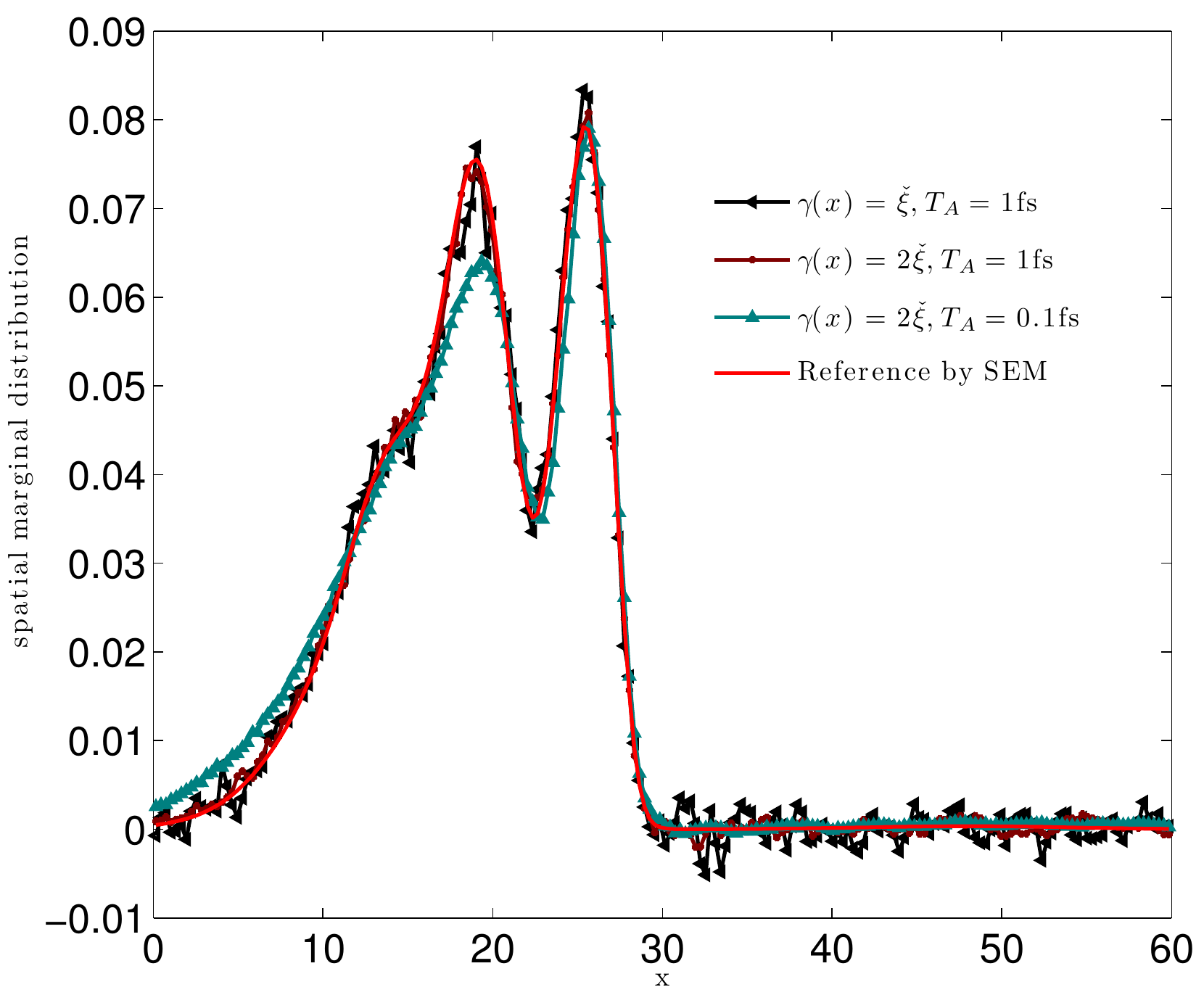}
\includegraphics[width=0.49\textwidth,height=0.35\textwidth]{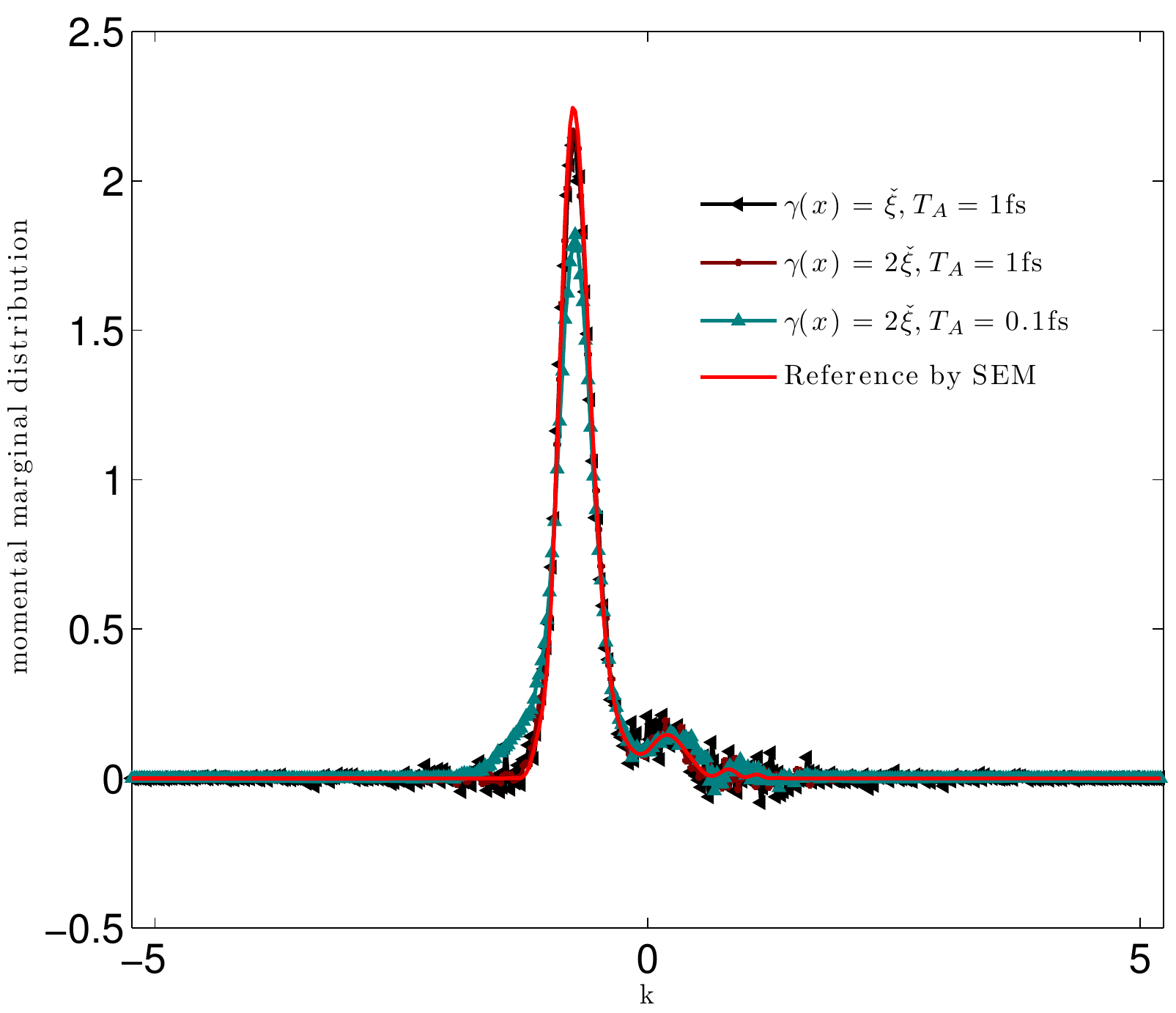}}
\vspace{-0.2cm}
\caption{\small Total reflection by the Gaussian barrier: Spatial (left column) and momental (right) marginal probability distributions at $t=5,10,15,20$fs.}
\label{fig:ex2-sm}
\end{figure}

Next we focus on the efficiency. One of the main variables shaping
the efficiency is the particle number. Since the same initial
particle distribution is employed for all runs, we only need to
consider the growth rate of particle number. Once choosing the
constant auxiliary function $\gamma_0$ and the annihilation
frequency $1/T_A$, the increasing multiple of particle number can be
exactly determined by $\me^{2\gamma_0T_A}$ from
Theorem~\ref{th:exp}, and the eighth column of Table \ref{tab:eg1-1}
shows corresponding theoretical predictions. In our numerical
simulations, within every annihilation period $[t^{i-1},t^i]$ with
$1\leq i\leq n_A$, we record the starting particle number
$\#_P^a(t^{i-1})$, the ending particle number $\#_P^b(t^i)$, and the
related growth rate $M_i$ in Eq.~\eqref{eq:rate}. Let
\[
\check{\#_P^{a}} = \max\limits_{0\leq i\leq {n_A}-1}\{\#_P^{a}(t^i)\},
\quad
\check{\#_P^{b}} = \max\limits_{1\leq i\leq n_A}\{\#_P^b(t^i)\},
\quad
\overline{M} = \frac{1}{n_A}\sum_{i=1}^{n_A} M_i.
\]
Table \ref{tab:eg1-1} gives numerical values of above three
quantities, see the fifth, sixth and seventh columns. According to
Table \ref{tab:eg1-1}, we can figure out the following facts on the
efficiency.

\begin{table}
  \centering
  \caption{\small Total reflection by the Gaussian barrier: Numerical data for WBRW. Detailed explanations are referred to Table~\ref{tab:eg1-1}, except for $\check{\xi}\approx$ 1.28 here.}
\label{tab:eg2-1}
 %\newsavebox{\tablebox}
 \begin{lrbox}{\tablebox}
  \begin{tabular}{ccccccccc}
\hline\hline
$\gamma(x)$ & $\text{err}_{wf}$ & $\text{err}_{sm}$ & $\text{err}_{mm}$ & $\check{\#}_P^b$ & $\check{\#}_P^a$ & $\overline{M}$ & $\me^{2\gamma_0T_A}$ & Time\\
\hline
\multicolumn{9}{c}{$\Delta t=1$fs, $T_A=1$fs}\\
\hline
$\xi(x)$ &
2.6210E-01&
7.1421E-02&
7.6208E-02&
48.63 &
3.95 &
12.45 &
--&
11.33   \\
$\check{\xi}$ &
2.5998E-01 &
6.7069E-02 &
8.4044E-02 &
49.58 &
3.94 &
12.71 &
13.04 &
12.48  \\
$2\check{\xi}$ &
1.3034E-01 &
3.1041E-02 &
5.5122E-02&
479.96 &
2.87 &
167.36 &
170.12 &
113.62  \\
\hline
\multicolumn{9}{c}{$\Delta t=0.1$fs, $T_A=0.1$fs}\\
\hline
$\xi(x)$ &
3.3899E-01&
1.0171E-01&
1.9413E-01&
3.22 &
2.50 &
1.29 &
--&
25.33  \\
$\check{\xi}$ &
3.4201E-01&
1.0674E-01&
1.9785E-01&
3.23 &
2.50 &
1.29 &
1.29 &
25.87  \\
$2\check{\xi}$ &
3.3697E-01&
1.0959E-01&
1.9413E-01&
4.05 &
2.43 &
1.67 &
1.67 &
28.98   \\
$3\check{\xi}$ &
3.4011E-01&
1.1157E-01&
1.9734E-01&
5.18 &
2.40 &
2.16 &
2.16 &
29.80  \\
\hline
\hline
  \end{tabular}
\end{lrbox}
\scalebox{0.88}{\usebox{\tablebox}}
\end{table}

\begin{description}
\item[(1)] Agreement between the mean value $\overline{M}$ and
the theoretical prediction $\me^{2\gamma_0T_A}$ is readily seen in
all situations. In this case, the average acceptance ratio
$\alpha_0$ almost equals to one due to the localized structure of
the Wigner kernel \eqref{eq:gp_vw}. Actually, the growth rates in
the first five annihilation periods, e.g. for $T_A=1$fs and
$\gamma_0=3\check{\xi}$, are 5.92, 5.87, 5.88, 5.89, and 5.88, all
of which are almost identical to the mean value of 5.89. When
$T_A=1,2,4$fs, the former is a little less than the latter, because
the particles moving outside the computational domain $\fx\times\fk$
are not taken into account. Within each annihilation period, the
maximum travel distance of particles can be calculated by
\[
\frac{\hbar}{m}\cdot \max\limits_{k\in\fk}\{|k|\} \cdot T_A,
\]
implying that, the larger value $T_A$ is, the more particles move
outside the domain. This explains the slight deviation between
$\overline{M}$ and $\me^{2\gamma_0T_A}$ increases from almost zero
to at most 1.33 as $T_A$ increases from $0.1$fs to $4$fs. Moreover,
when $T_A=1$fs, the increasing multiples for $\Delta t=0.008$fs are
identical to those for $\Delta t=1$fs, i.e., $\overline{M}$ is
independent of $\Delta t$, which has been also already predicted by
the theoretical analysis. More details about the agreement of the
growth rates of particle number with the theoretical prediction for
$T_A=1$fs and $\Delta t=0.008$fs can be found in
Fig.~\ref{fig:ex1-np}.
\item[(2)] Not like using the constant auxiliary function, we do not have a simple calculation formula so far for the growth rate of particle number when using the variable auxiliary function (i.e., depending both on time and trajectories). However, we can still utilize the growth rate for the case of $\gamma_0=\check{\xi}$ to provide a close upper bound for
the case of $\gamma(x)=\xi(x)$. As shown in the seventh column of
Table~\ref{tab:eg1-1}, the variation of the mean growth rate between
them is about 0, 0.03, 0.12 and 0.74 for $T_A=$ 0.1fs, 1fs, 2fs and
4fs, respectively. Fig.~\ref{fig:ex1-np} further compares the curves
of growth rate for $T_A=1$fs and $\Delta t=0.008$fs within four
typical annihilation periods. By comparing with the Wigner functions
shown in Fig.~\ref{fig:ex1-wf}, we find that the closer to the
center the Gaussian wavepacket lives, the larger the deviation
between the curves for the constant and variable auxiliary functions
becomes. Such deviation in accordance with the analysis of $\xi(x)$
shown in Fig.~\ref{fig:xi}
validates the proposed mathematical theory again.
\item[(3)] During the resampling (annihilation) procedure,
the main objective is to reconstruct the Wigner distribution using
less particles, which explores the cancelation of the weights with
opposite signs, see Eq.~\eqref{eq:s}. The sixth column of
Table~\ref{tab:eg1-1} tells us that the maximum particle numbers
after resampling for $T_A=1,2,4$fs are all around 3.00 million,
implying that there should be a minimal requirement of  particle
number to achieve a comparable accuracy. Otherwise, the accuracy
will decrease, for example, the values of $\check{\#}_P^a$ for
$T_A=0.1$fs are around 2.55 million. Fig.~\ref{fig:ex1-npaa} shows
more clearly the typical history of ${\#}_P^a(t)$. We can find there
that, no matter how huge the particle number before the annihilation
${\#}_P^b(t)$ (which depends on both $\gamma(x)$ and $T_A$) is, the
particle number after the annihilation ${\#}_P^a(t)$ for the
simulations with comparable accuracy exhibits almost the same
behavior, which recovers and extends the so-called ``bottom line"
structure described in \cite{ShaoSellier2015}. Such behavior may
depend only on the oscillating structure of the Wigner function. On
the other hand, highly frequent annihilations like $T_A=0.1$fs
destroy this bottom line structure and thus the accuracy, see
Fig.~\ref{fig:ex1-npaa-b}, implying that there are no enough
particles to capture the oscillating nature.
\end{description}

\begin{figure}
\subfigure[Relative errors.]{\includegraphics[width=0.49\textwidth,height=0.38\textwidth]{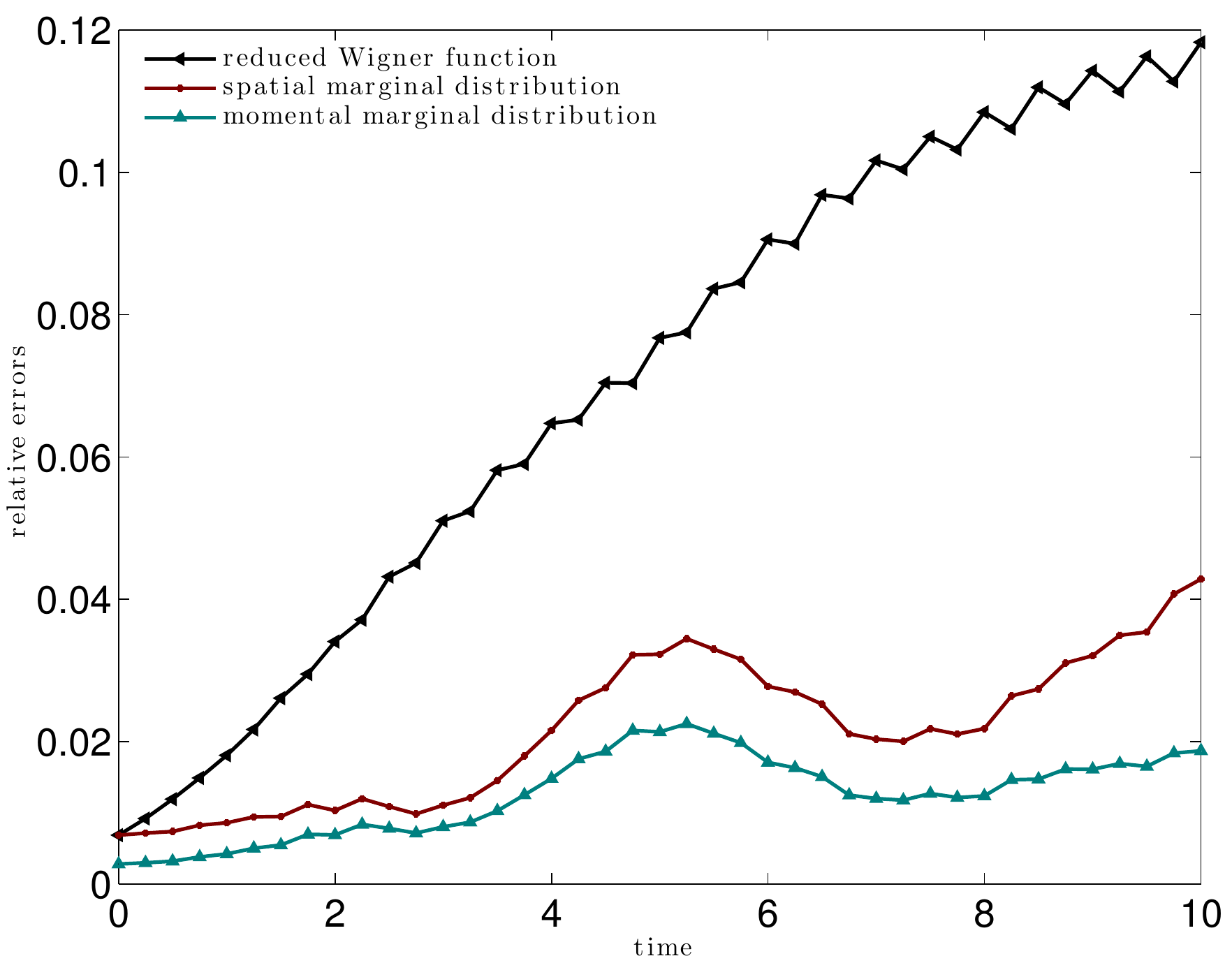} \label{fig:ex3-error}}
\subfigure[Particle number after resampling.]{\includegraphics[width=0.49\textwidth,height=0.38\textwidth]{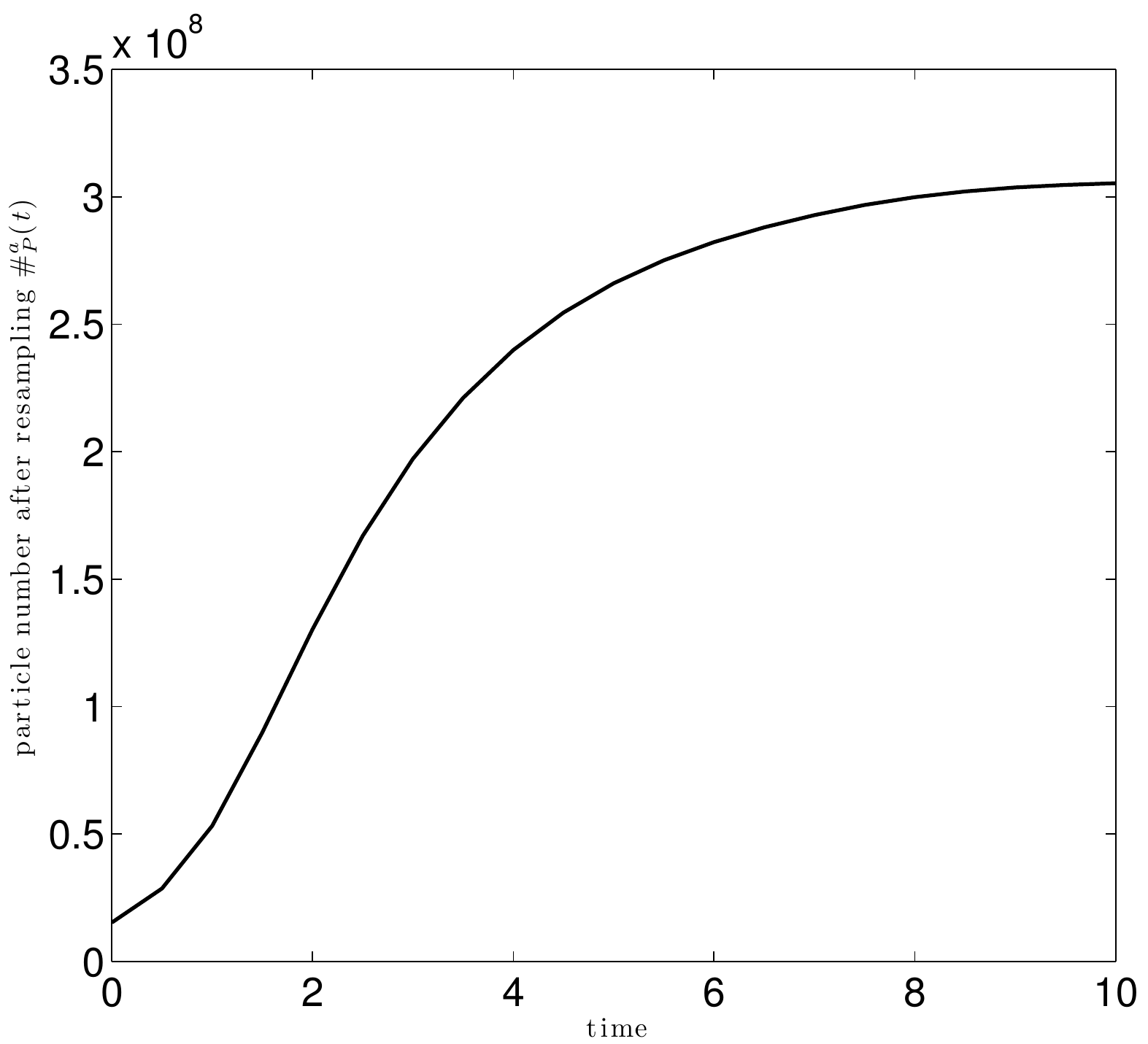} \label{fig:ex3-pn}}
\caption{\small The Helium-like system: The history of relative errors and particle number after resampling.}
\end{figure}

$\bullet$ \textbf{Experiment 2} In this example, we increase the
barrier height to $H=1.3$eV so that the Gaussian wavepacket will be
totally reflected, see Fig.~\ref{fig:ex2-wf}. Such augment of the
barrier height implies that the growth rate of particle number now
is about $1.3/0.3\approx4.33$ times larger than that for  Experiment
1, and thus it is more difficult to simulate accurately. Based on
the observations in Experiment 1, we only test two groups of
annihilation periods, $T_A=0.1, 1$fs. Table~\ref{tab:eg2-1}
summarizes the running data and confirms again that, the larger
constant auxiliary function improves the accuracy, whereas the
higher annihilation frequency destroys the accuracy. In order to get
a more clear picture on this accuracy issue, we plot both spatial
and momental probability distributions at different time instants
$t=5,10,15,20$fs in Fig.~\ref{fig:ex2-sm} against the reference
solutions by SEM. We can easily see there that, the loss of accuracy
when using $T_A=0.1$fs is mainly due to that there are no enough
generated particles to capture the peaks reflecting off the barrier;
while the increase of accuracy when using a larger constant
auxiliary function, e.g., $\gamma_0=2\check{\xi}$, comes from the
smaller variation. Actually, similar phenomena also occur in
Experiment 1.

\begin{figure}
\subfigure[$t=2.5$.]{\includegraphics[width=0.49\textwidth,height=0.27\textwidth]{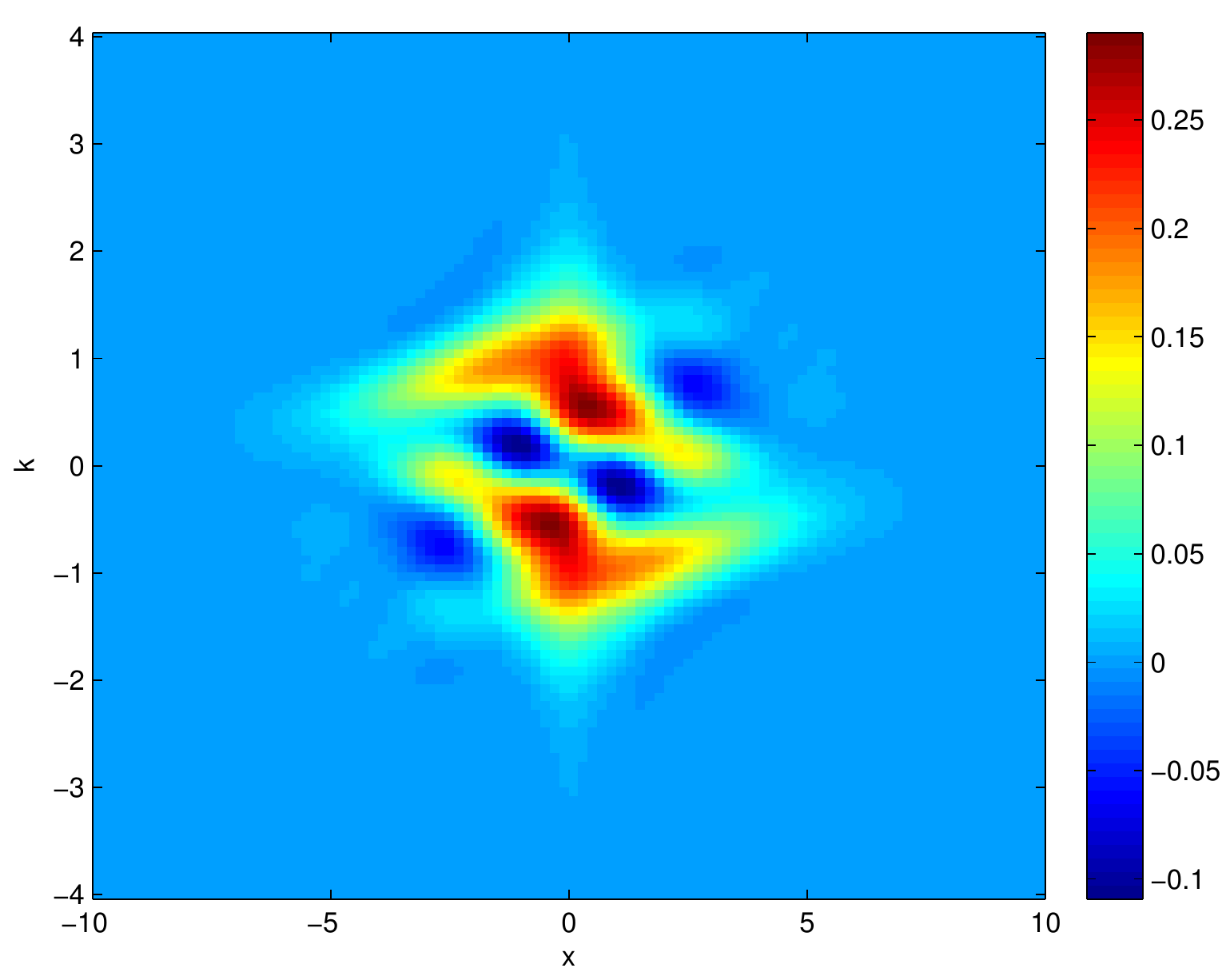}{\includegraphics[width=0.49\textwidth,height=0.27\textwidth]{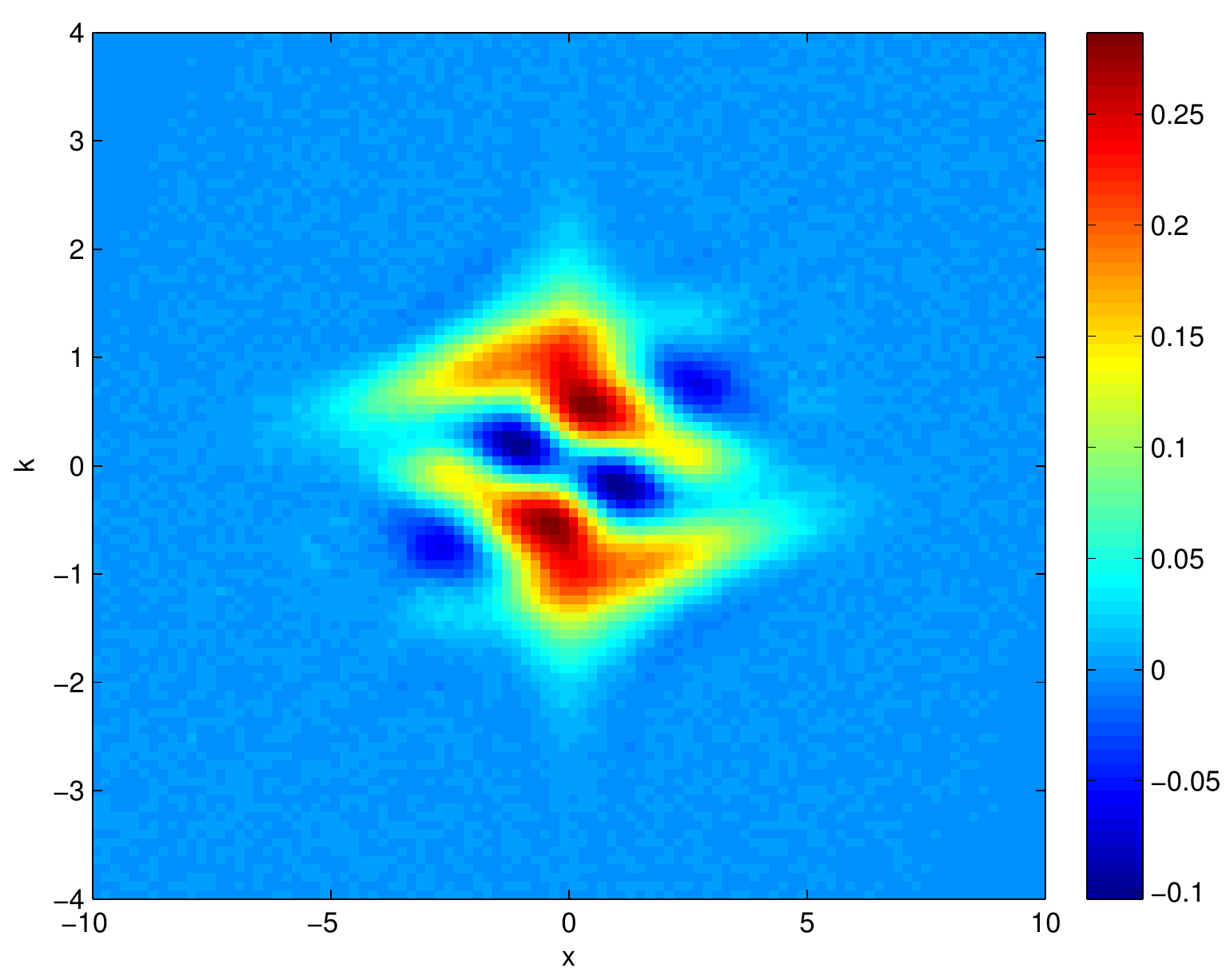}}}
\subfigure[$t=5$.]{\includegraphics[width=0.49\textwidth,height=0.27\textwidth]{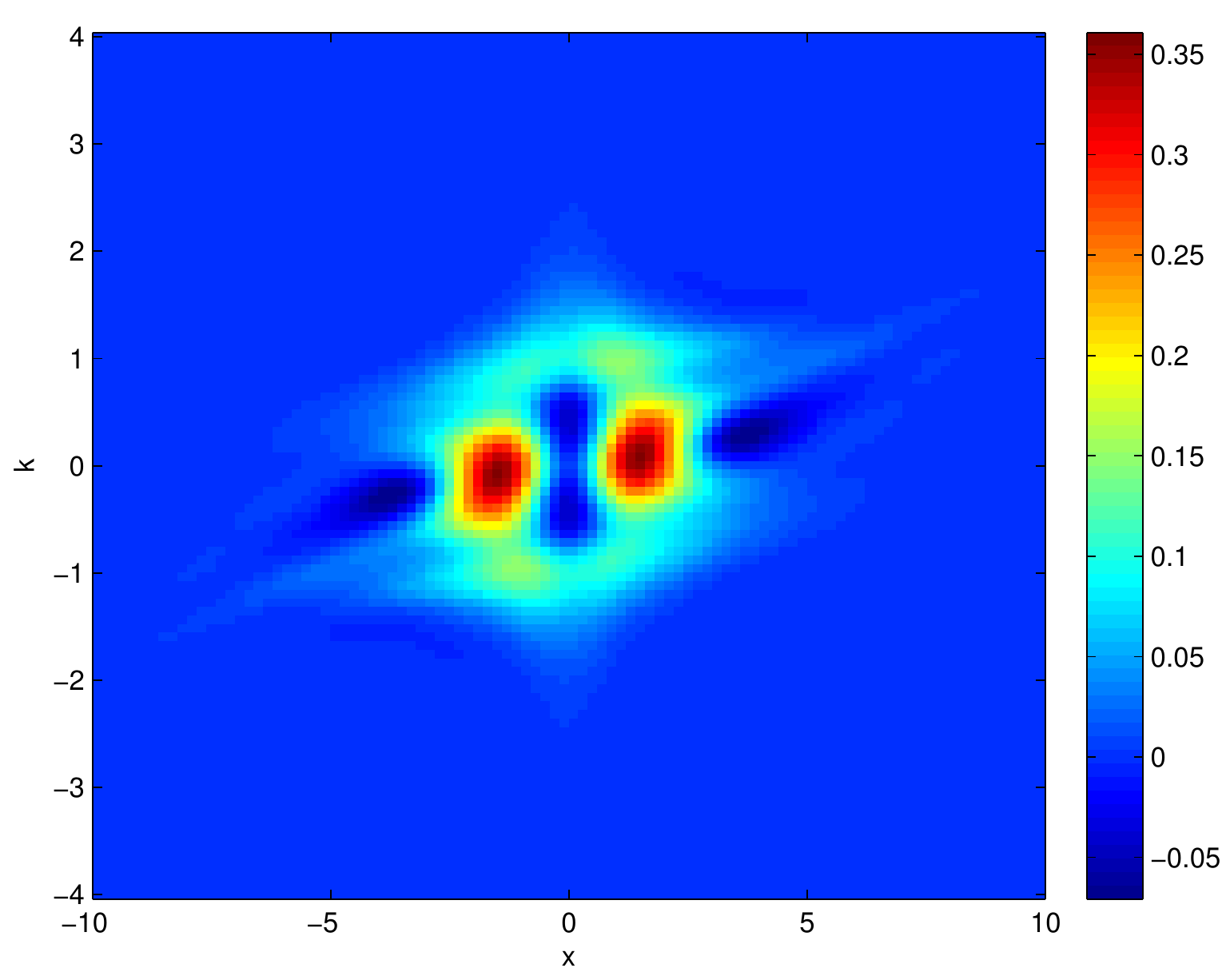}{\includegraphics[width=0.49\textwidth,height=0.27\textwidth]{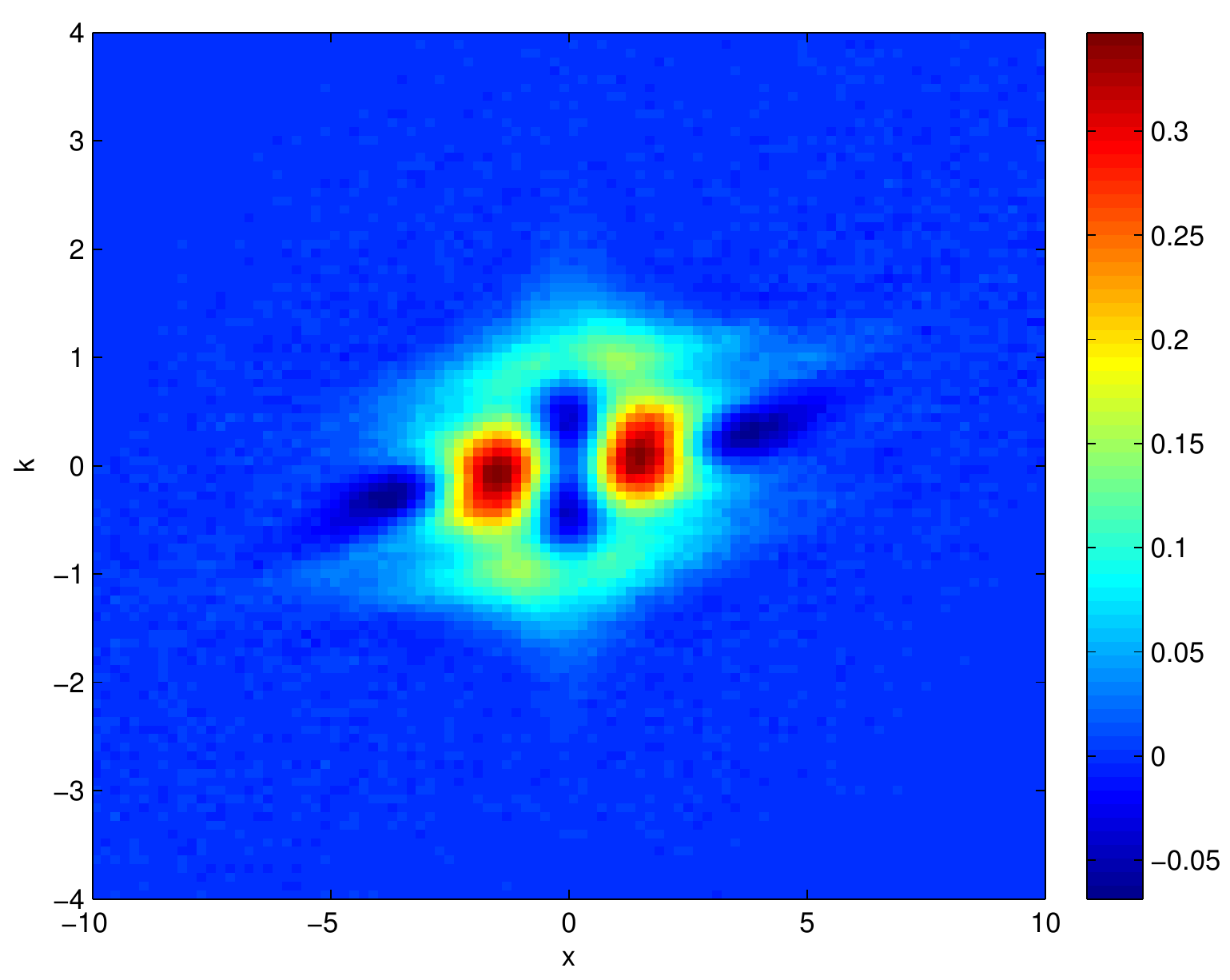}}}
\subfigure[$t=7.5$.]{\includegraphics[width=0.49\textwidth,height=0.27\textwidth]{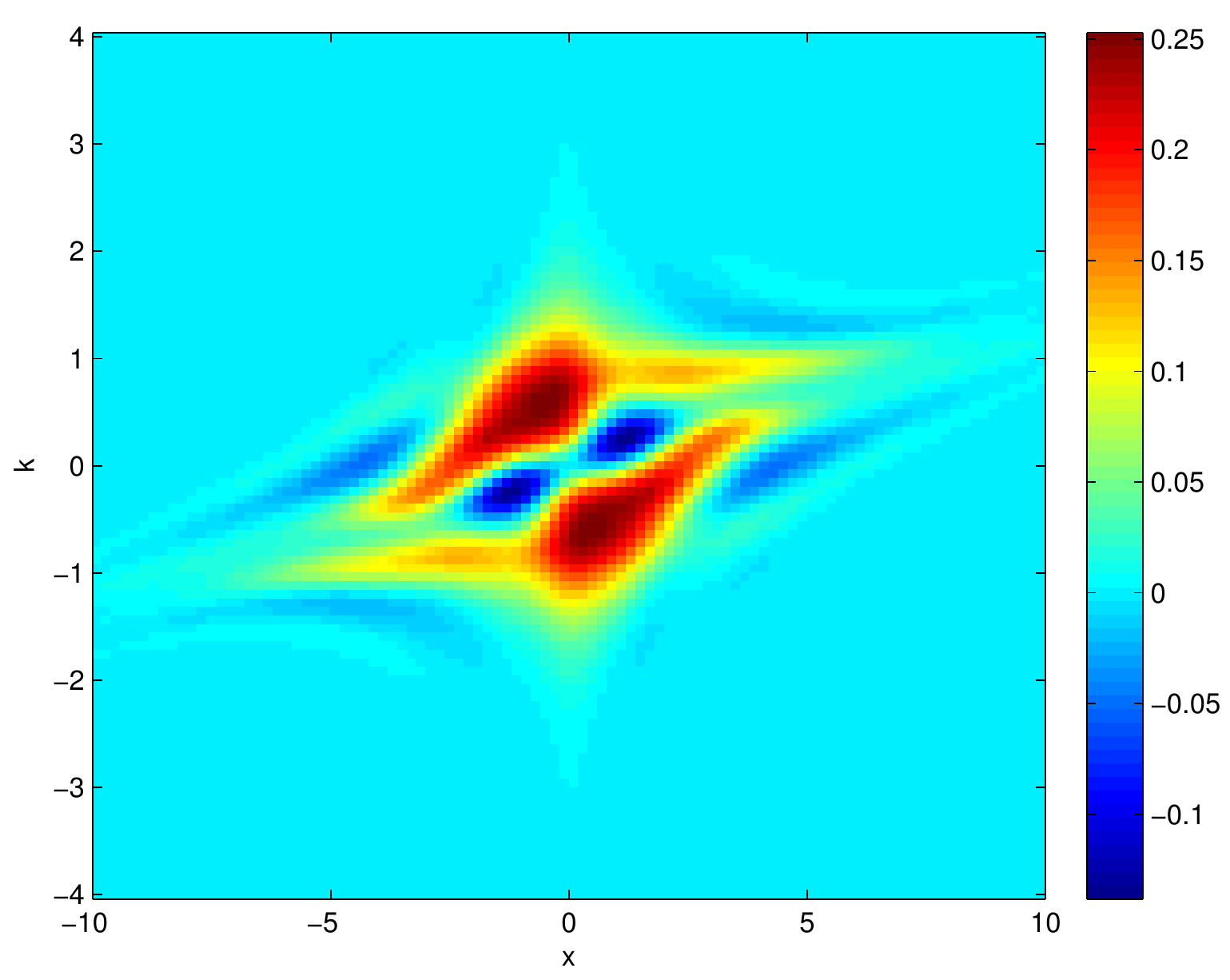}{\includegraphics[width=0.49\textwidth,height=0.27\textwidth]{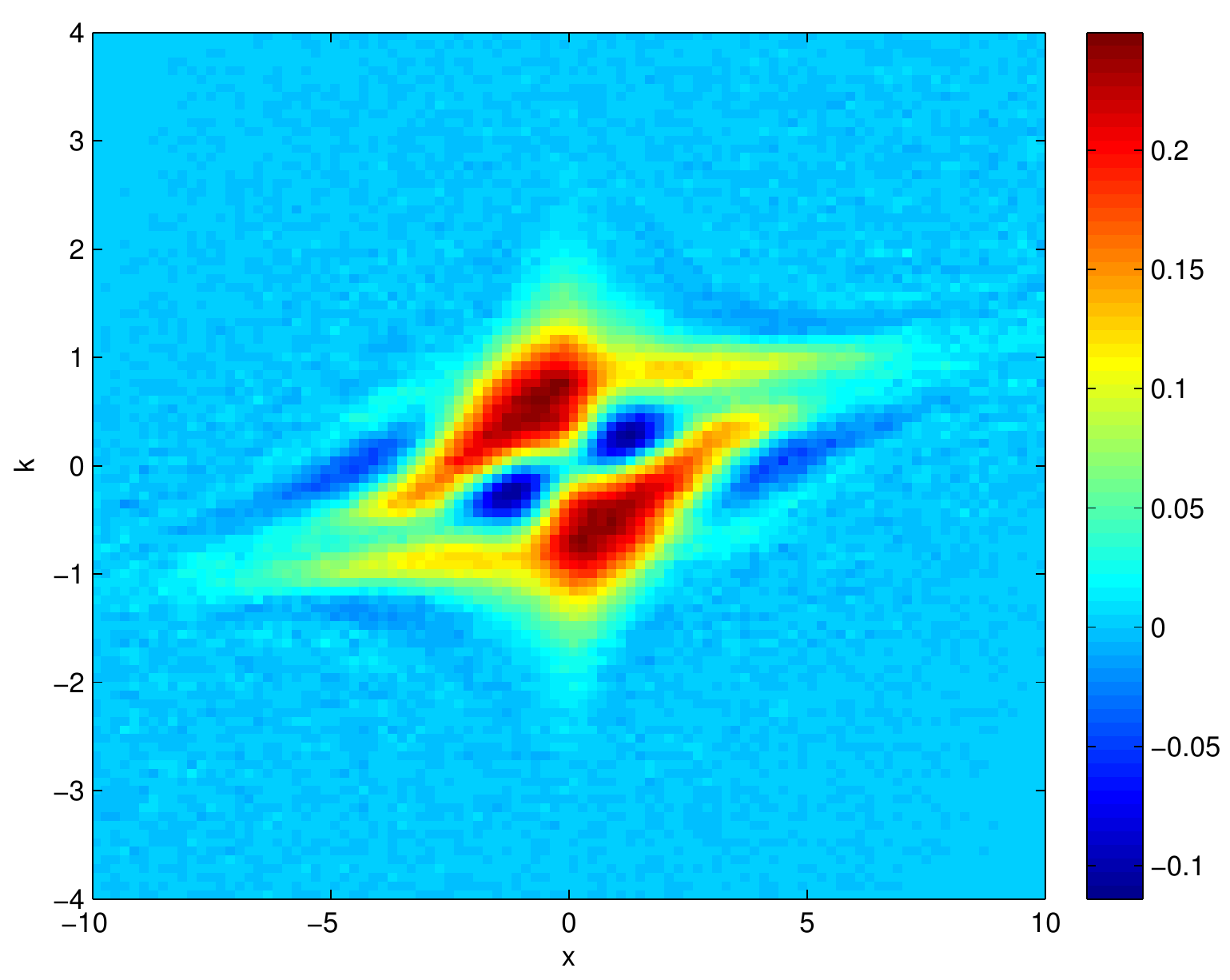}}}
\subfigure[$t=10$.]{\includegraphics[width=0.49\textwidth,height=0.27\textwidth]{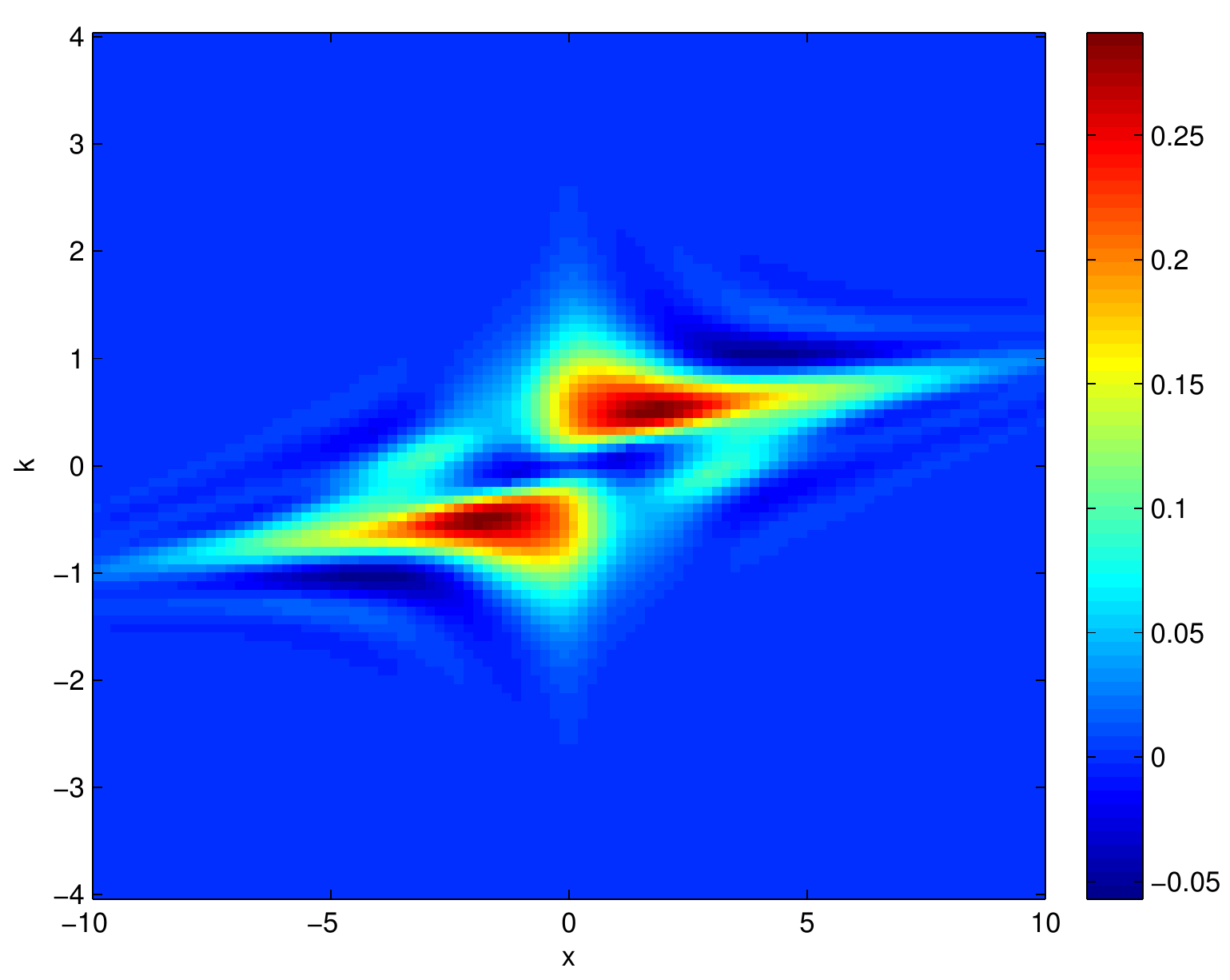}{\includegraphics[width=0.49\textwidth,height=0.27\textwidth]{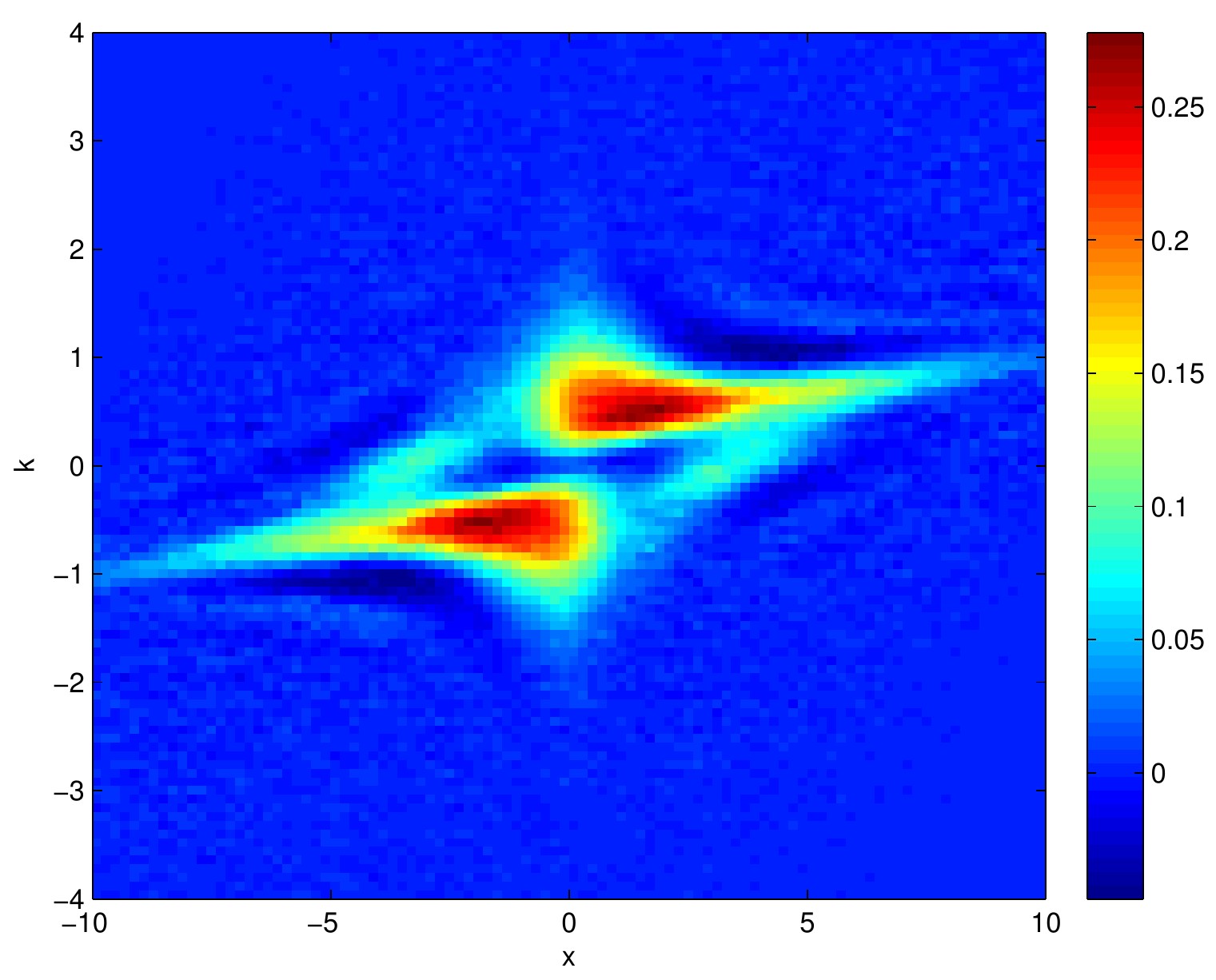}}}
\caption{\small The Helium-like system: Numerical reduced Wigner functions at $t=2.5,5,7.5,10$.
The left column displays the reference solution by ASM,
while the right column shows the numerical solution by WBRW with the
constant auxiliary function $\gamma_0 = 2$, $\Delta t=0.25$ and $T_A=0.5$.}
\label{fig:ex3-wf}
\end{figure}

\subsection{A Helium-like system}

As a typical two-body system, the two-body
Helium-like system has been considered in testing deterministic
Wigner solvers in 4D phase space
\cite{XiongChenShao2015}. Here we utilize again a Helium-like system
in which the electron-nucleus and electron-electron interaction is
given by
\begin{equation}
V(x_1, x_2) = -\frac{2\me^{-\kappa|x_1-x_A|}}{2\kappa}-\frac{2\me^{-\kappa|x_2-x_A|}}{2\kappa}+\frac{\me^{-\kappa|x_1-x_2|}}{2\kappa},
\end{equation}
where the parameter $\kappa$ expresses the screening strength,
$x_A$ denotes the position of the nucleus, and
$x_i (i=1,2)$ is the position of the $i$-th electron.
In fact, $\me^{-\kappa|x_1-x_2|}/2\kappa$ is the Green's function of the 1D screened Poisson equation. The Wigner kernel of the electron-nucleus interaction reads
\begin{equation}
V_w(x_i, k_i) = \frac{2}{\hbar \pi} \cdot \frac{\sin(2(x_i-x_A))}{4k_i^2+\kappa^2}, \;\;\; i=1,2,
\end{equation}
and that of the electron-electron interaction
\begin{equation}
V_w(x_1, x_2, k_1, k_2) = \frac{4}{\hbar \pi}  \cdot \frac{\sin(2k_1 x_1+2k_2 x_2)}{|k_1-k_2|^2+\kappa^2} \cdot \delta(2k_1+2k_2).
\end{equation}
Therefore we can still use a simple rejection method to draw samples from the target distribution $V_w^+(x, k)/\xi(x)$.
Here we use the atomic unit, set $x_A=0$ and $\kappa=0.5$
and adopt the same initial data as used in \cite{XiongChenShao2015}.
The computational domain $\mathcal{X} \times \mathcal{K} = [-10, 10]^2\times [-4,  4]^2$ is divided into $100^4$ cells. The reference solution is obtained by ASM on a uniform grid mesh with $\Delta t=0.05$ and $\Delta x_1 = \Delta x_2 = 0.2$,
while the $\mathcal{K}$-domain is divided into 8 cells and each cell contains 16 collocation points, and the $\mathcal{Y}$-domain is $[-22.5, 22.5]^2$.

\begin{figure}
\vspace{-2cm}
\subfigure[$t=2.5$.]{\includegraphics[width=0.49\textwidth,height=0.35\textwidth]{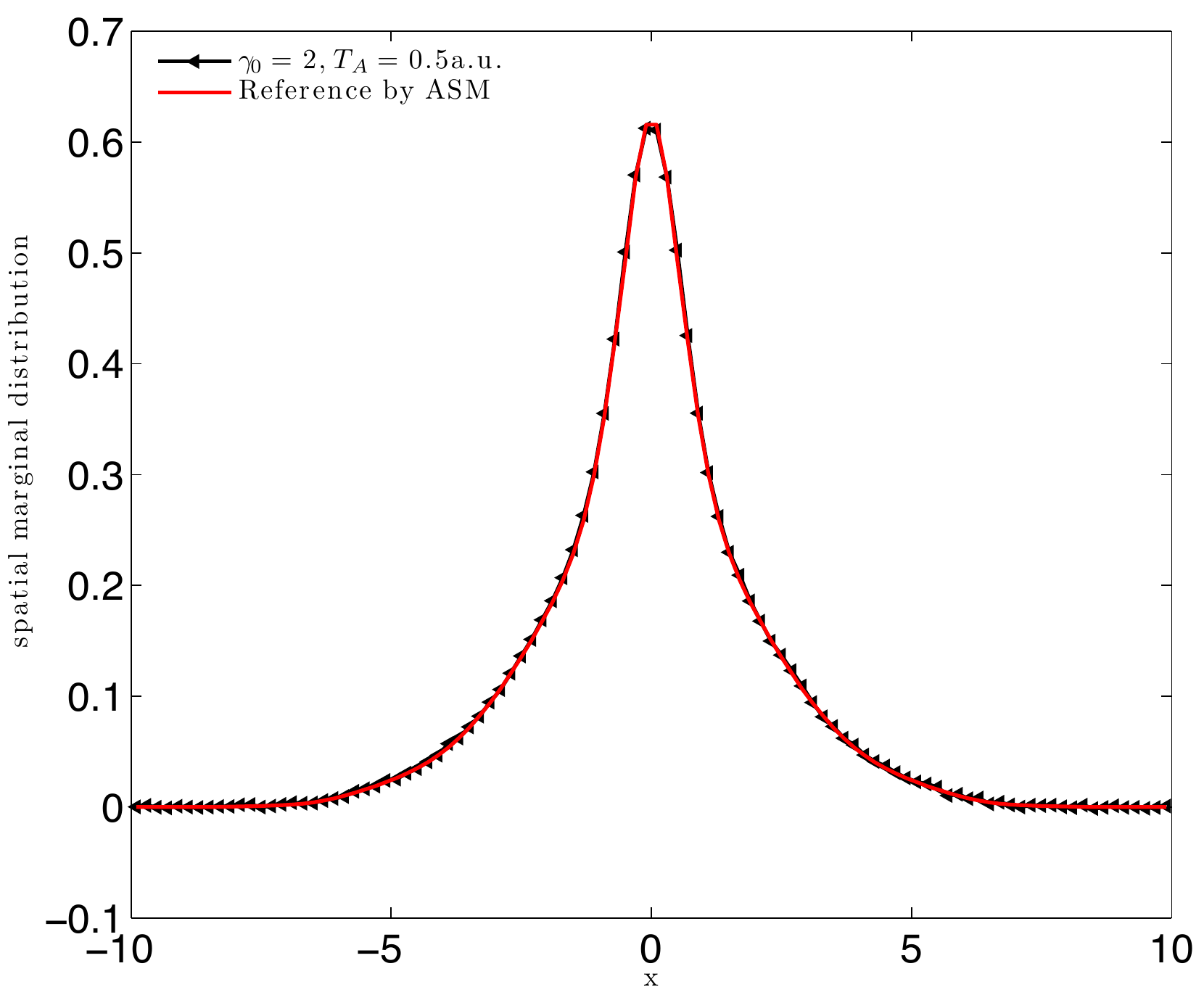}
\includegraphics[width=0.49\textwidth,height=0.35\textwidth]{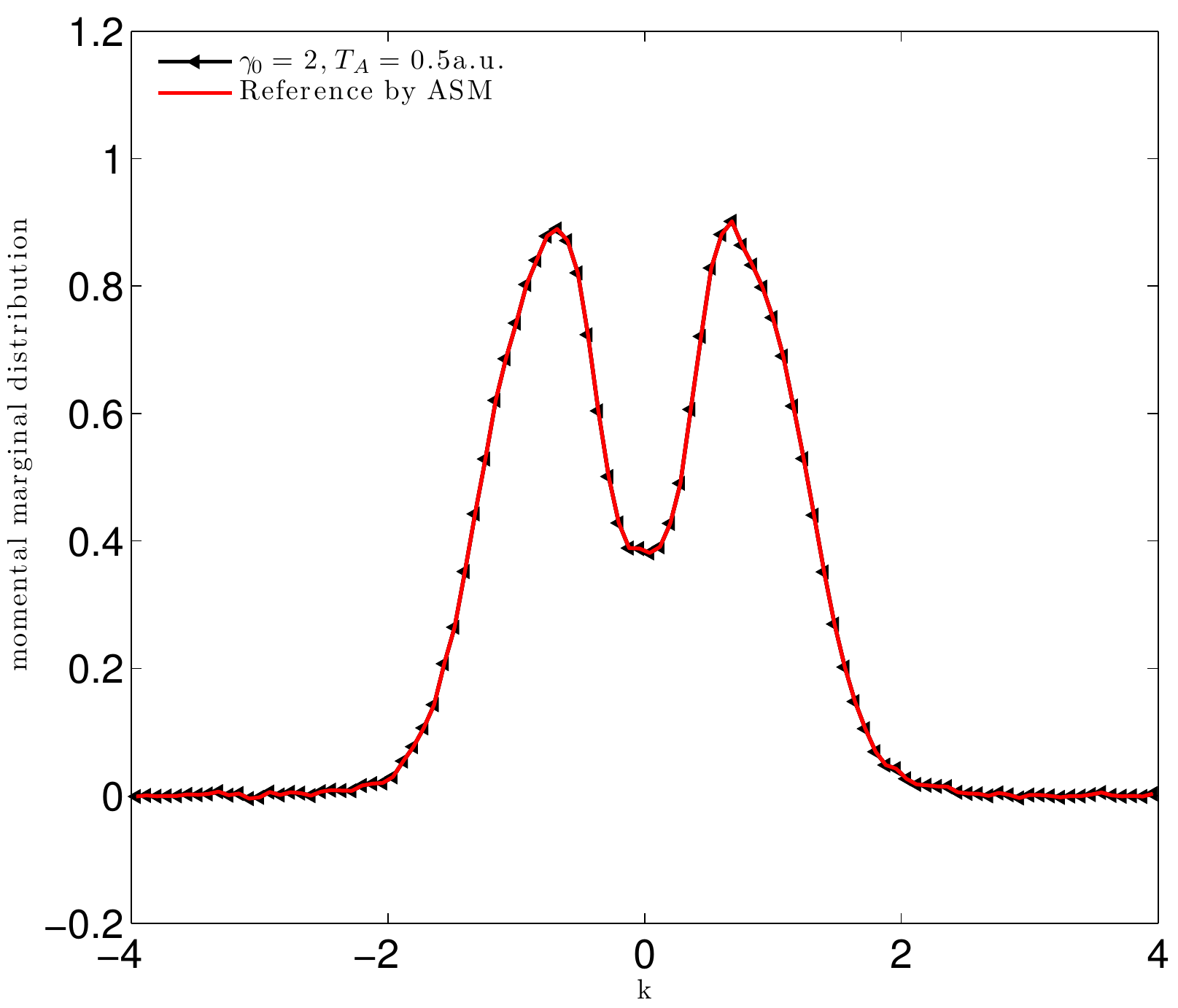}}\\
\subfigure[$t=5$.]{\includegraphics[width=0.49\textwidth,height=0.35\textwidth]{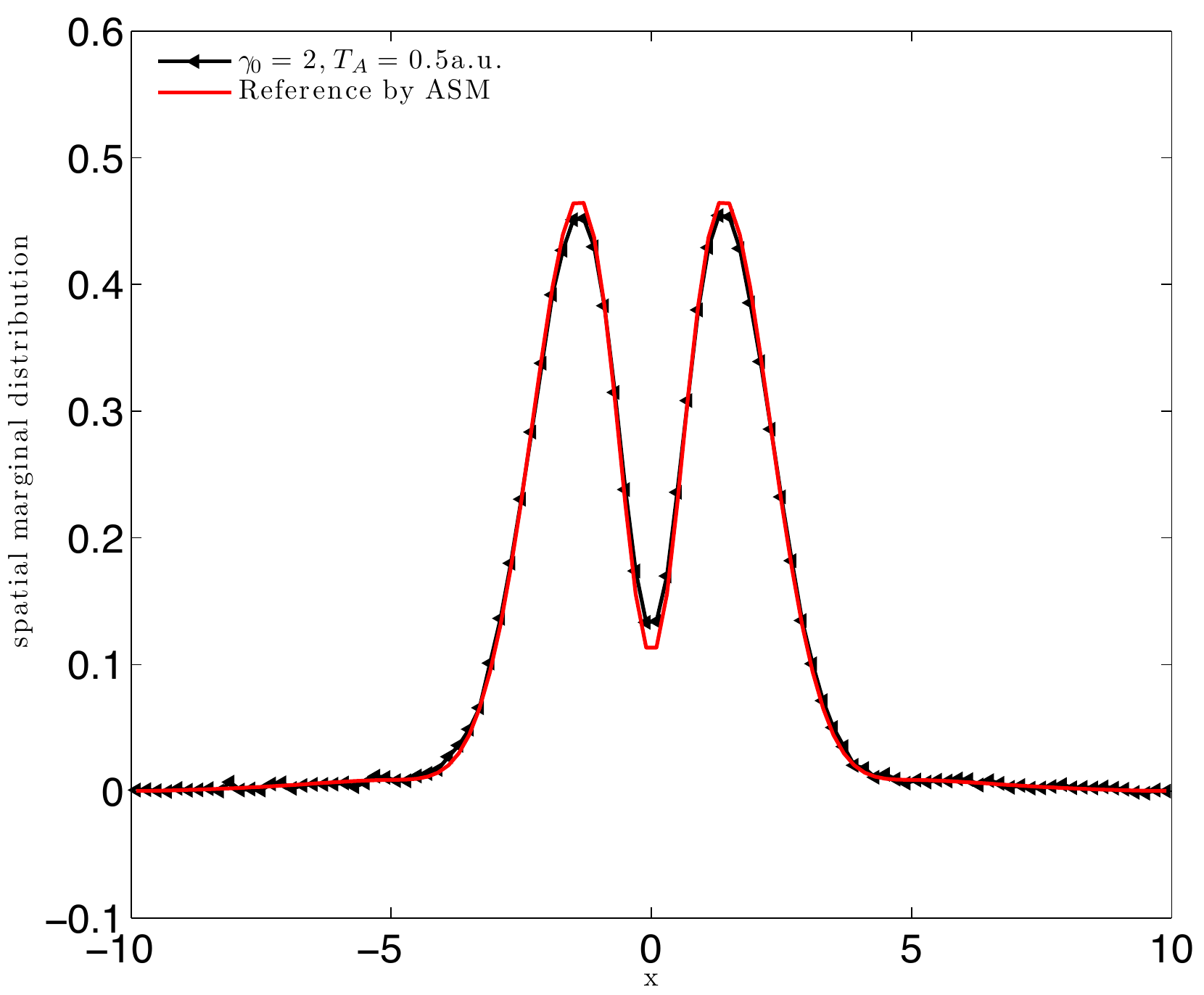}
\includegraphics[width=0.49\textwidth,height=0.35\textwidth]{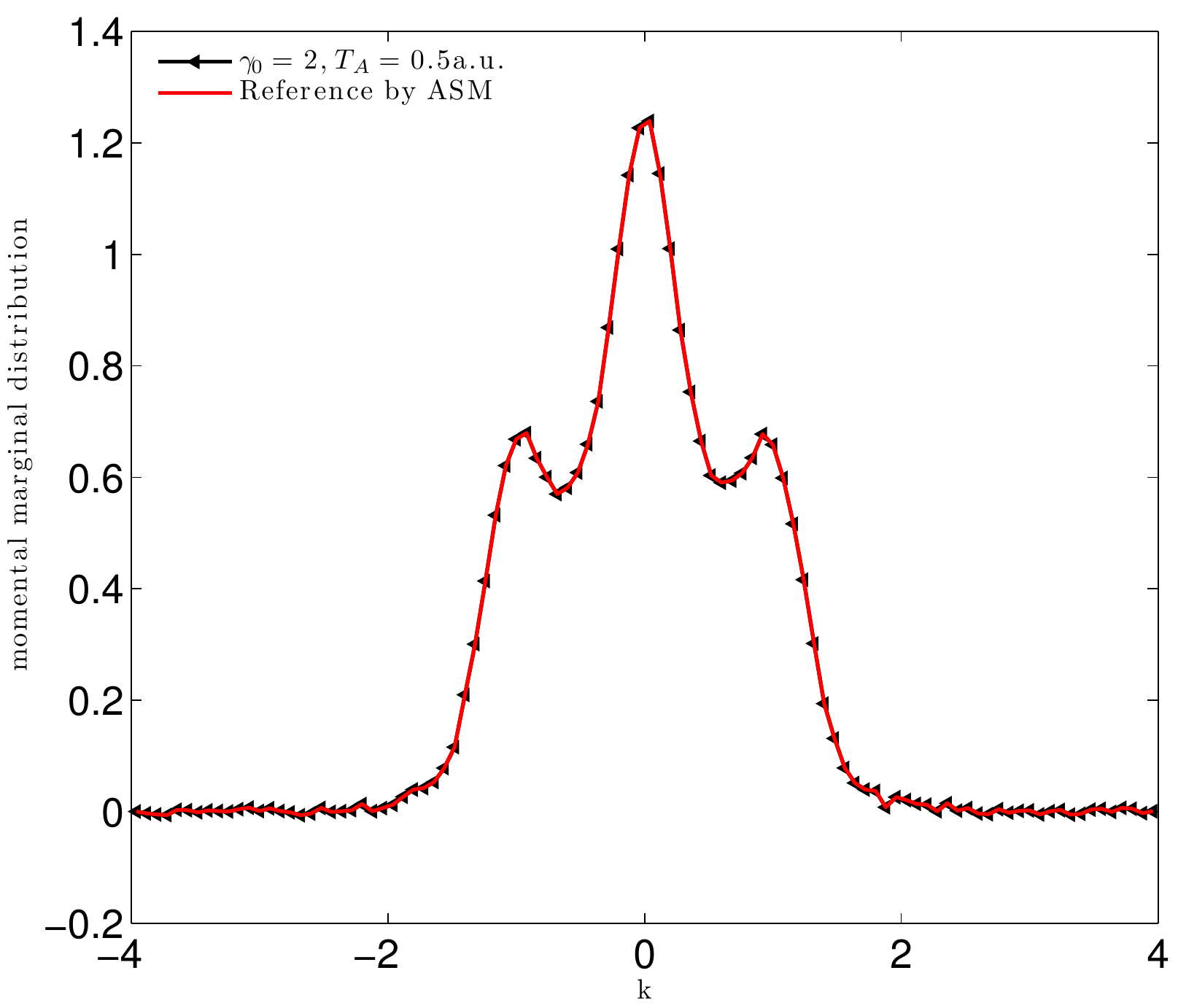}}
\subfigure[$t=7.5$.]{\includegraphics[width=0.49\textwidth,height=0.35\textwidth]{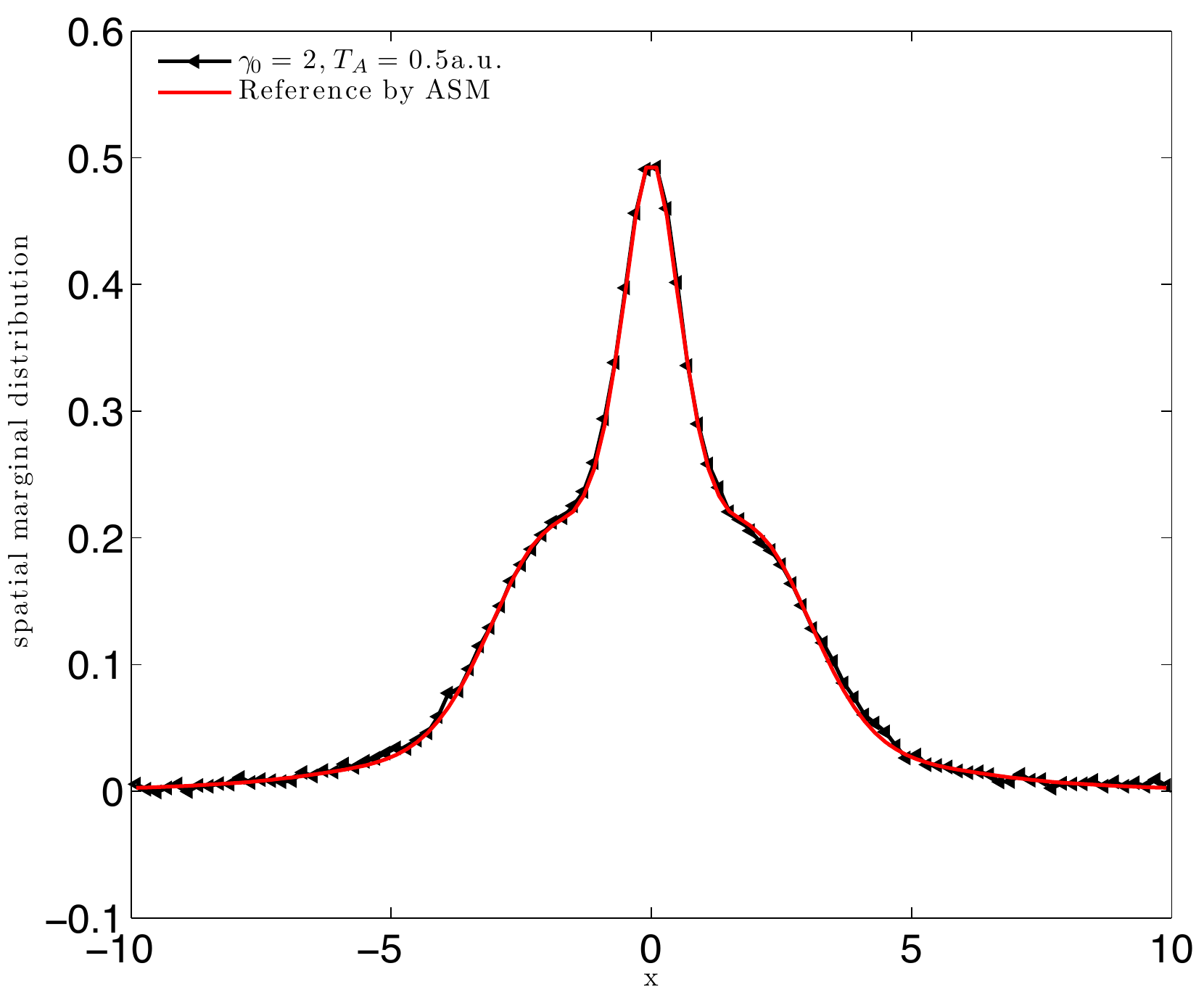}
\includegraphics[width=0.49\textwidth,height=0.35\textwidth]{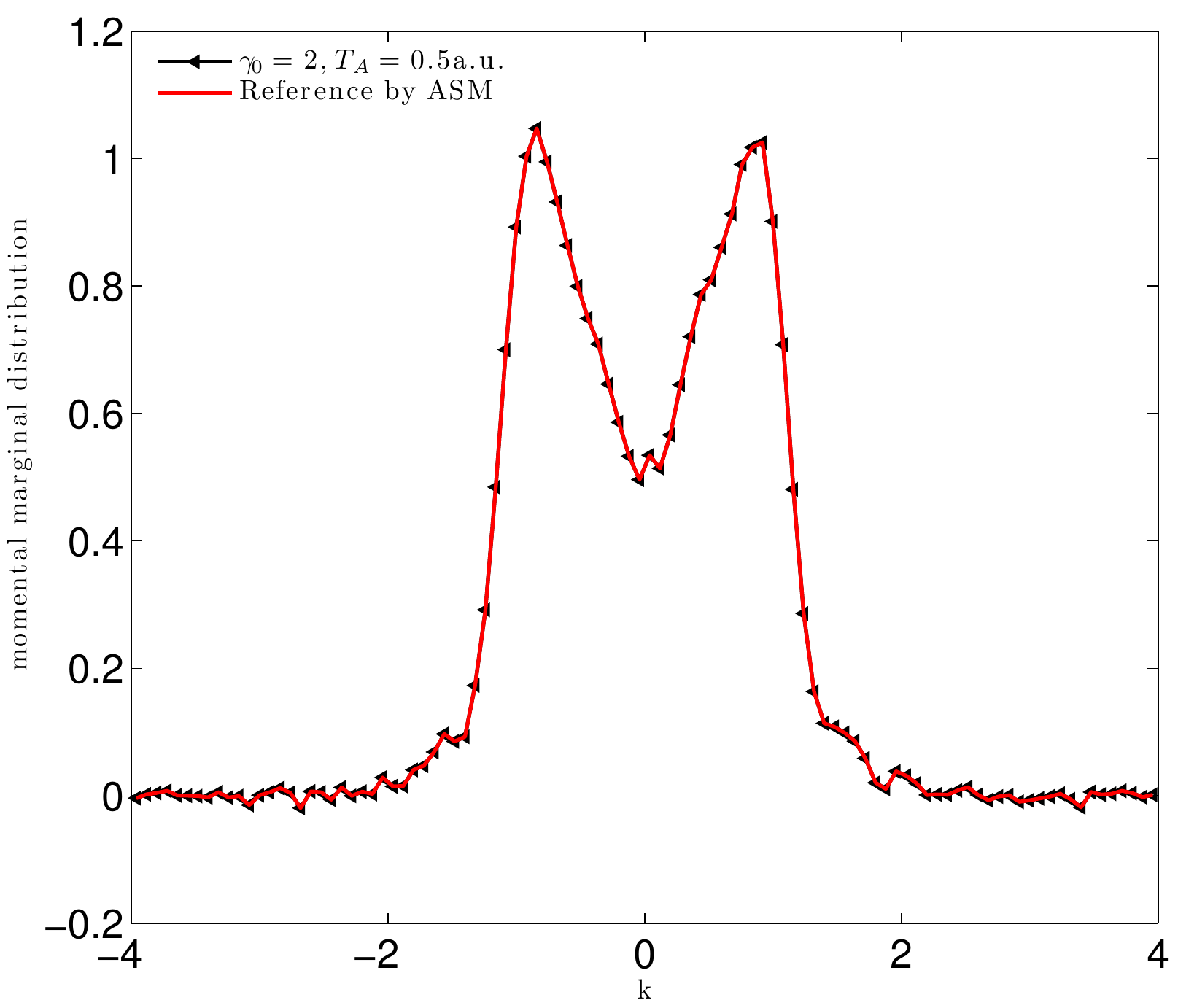}}\\
\subfigure[$t=10$.]{\includegraphics[width=0.49\textwidth,height=0.35\textwidth]{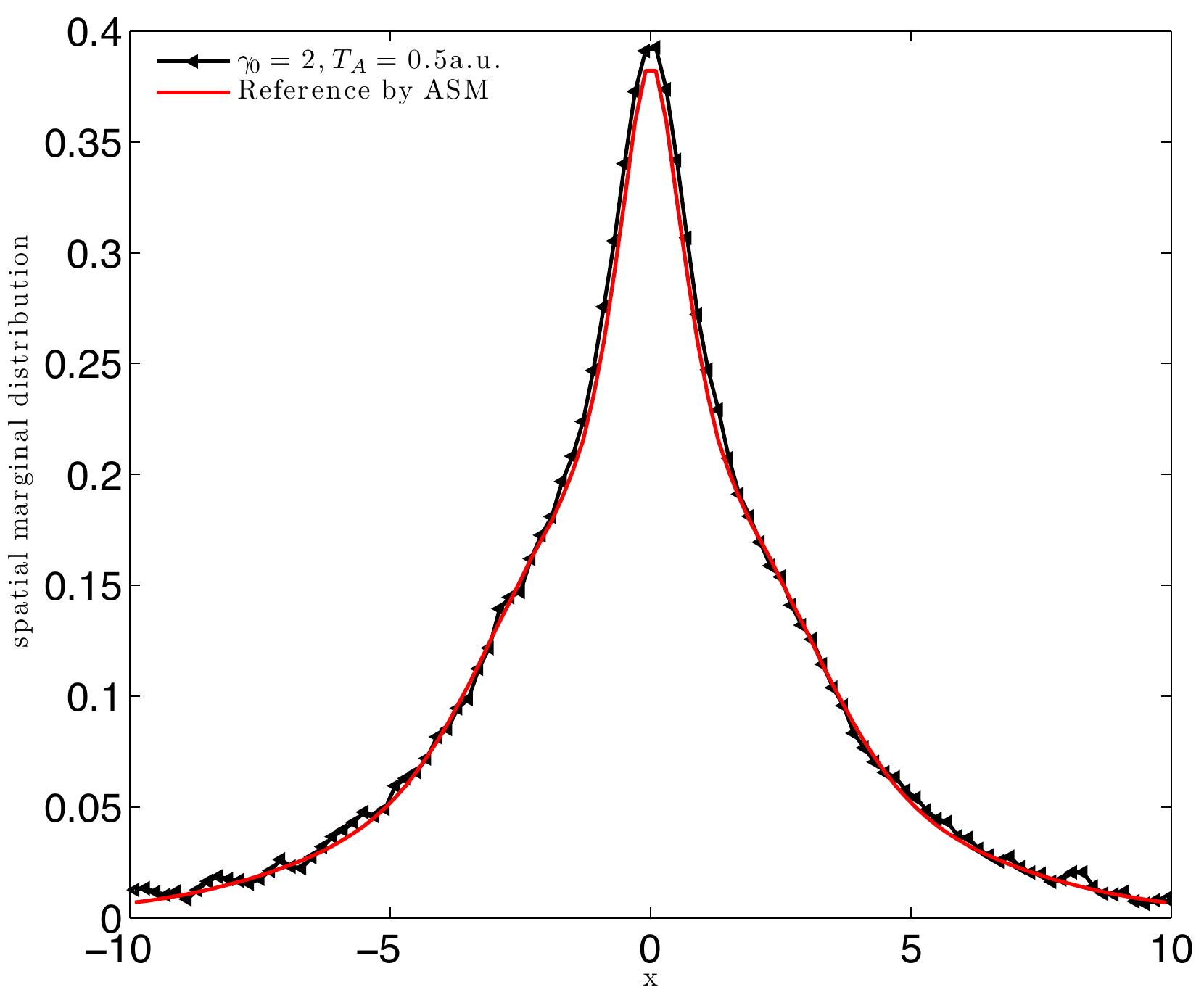}
\includegraphics[width=0.49\textwidth,height=0.35\textwidth]{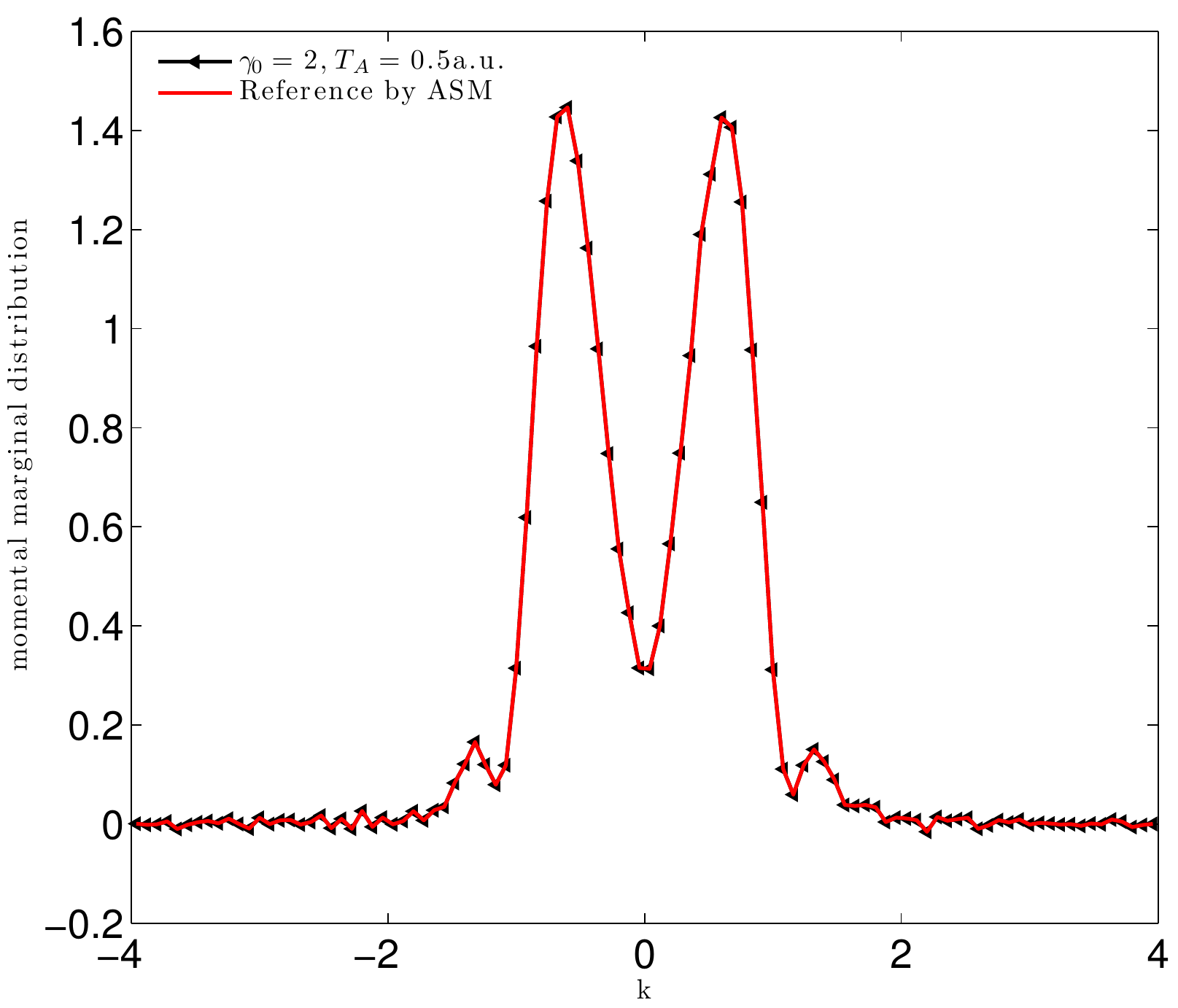}}
\vspace{-0.2cm}
\caption{\small The Helium-like system: Spatial (left column) and momental (right column) marginal probability distributions at $t=2.5,5,7.5,10$.}
\label{fig:ex3-sm}
\end{figure}

To monitor the accuracy,
we record the relative errors of the reduced single-body Wigner function as given in \cite{XiongChenShao2015}, and of corresponding marginal probability distributions. Fig.~\ref{fig:ex3-error}
shows the history of those relative errors. We can see there that, although the reduced Wigner function is comparatively less accurate, it can still yield a more accurate estimation of macroscopically measurable quantities, such as the spatial and momental marginal probability distributions. This also explains why we see more noise in Fig.~\ref{fig:ex3-wf} for the reduced Wigner function than in Fig.~\ref{fig:ex3-sm} for the marginal distributions. The possible reason may lie on the fact that if we wish to be able to estimate a function with the same accuracy as a function in low dimensions, then we need the size of samples to grow exponentially as well.
However, it can be readily observed in Figs.~\ref{fig:ex3-wf} and \ref{fig:ex3-sm} that the main features captured by WBRW are almost identical to those by ASM.

Finally, we would like to mention that the growth of particle number is closely related to the number of cells (dimensionality of feature space). In this example, we use a $100^4=10^8$ uniformly distributed cells for the resampling and set the initial particle number to be about $1.5\times 10^7$ with the total weighted summation being $1\times 10^7$.
It is shown in Fig.~\ref{fig:ex3-pn} that the particle number increases soon to $3\times 10^8$, which is comparable to the cell number, and then approaches a stable value around $3.1\times 10^8$. So if we refine those cells for the resampling, then the particle number will increase to a higher level. Actually, for higher-dimensional problems like $d\geq 3$, the number of cells is much higher than that of samples and such a simple cell based resampling strategy cannot achieve an efficient annihilation.
Hence we have to resort to other advanced techniques to control the sample size in higher-dimensional phase space.

\section{Conclusion and discussion}
\label{sec:con}

This paper is devoted to the mathematical foundation of the branching random walk algorithm for the many-body Wigner quantum dynamics. Although several concepts, such as the signed particle, the adjoint equation and the annihilation procedure, have already been mentioned in previous work, unfortunately related mathematical results are somewhat fragmented or lack of systemic elaboration, and the crucial issues, such as the annihilation of particles and the computational complexity, fall outside the scope of any current available theory. Thus, our original motivation is to provide a framework from the viewpoint of computational mathematics within which all these problems can be fully addressed, 
and interested readers may get a complete view of the Wigner branching random walk algorithm accompanied with both derivation and implementation details in a single reference. Only by this way can we analyze its accuracy, point out the numerical challenge and make further improvements. In fact, we have shown that the signed particle is naturally introduced according to the principle of importance sampling, the motion of particles is described by a probabilistic model, and the annihilation is nothing but the resampling from the instrumental distribution.

%can find all the derivation and implement details in a single reference, but meanwhile can also get a complete view of them

Although the theoretical part of this work is closely related to that shown recently in \cite{Wagner2016}, we adopt a different approach to interpreting the entire story. Actually, both approaches succeed in validating the basis of the spWMC, namely, Eq.~\eqref{eq:inner_product_probabilisitic}. The reason we prefer to the branching random walk model, a mixture of the branching process and the random walk, is that the theory of branching process not only provides a natural interpretation of growth of particles, but also allows us to calculate the particle growth rate exactly and discuss the conservation property. These results are extremely important in real simulations since it gives us a reasonable criterion to control the computational complexity and allocate computational resources efficiently.

We must admit that the numerical challenges 
in higher dimensional phase space are very potent, though the numerical accuracy in a 4D Helium-like system has also been validated in this work. 
The often-used simple cell based resampling technique cannot work even for 6D problems. Therefore it is urgent for us to seek an efficient way to reduce the sample size, and some advanced statistical density estimation methods might be taken into account.

\section*{Acknowledgement}
This research was supported by grants from the National Natural Science Foundation of China (Nos.~11471025, 91330110, 11421101).
The authors are grateful to the useful discussions with Zhenzhu Chen, Paul Ellinghaus, Mihail Nedjalkov and Jean Michel Sellier on the signed particle Monte Carlo method for the $y$-truncated Wigner equation.

%\bibliography{journalname,wigner}
%\bibliography{journalname,wigner}
% extra note for the arxiv paper is added by hand
%\bibliography{journalname,wigner}
%\bibliography{/Users/sihong/Reference/journalname,/Users/sihong/Reference/wigner}

\begin{thebibliography}{10}

\bibitem{bk:Doob2001}
J.~L. Doob.
\newblock {\em {Classical Potential Theory and Its Probabilistic Counterpart}}.
\newblock Springer-Verlag, Berlin, reprint edition, 2001.

\bibitem{bk:Dynkin2002}
E.~B. Dynkin.
\newblock {\em {Diffusions, Superdiffusions and Partial Differential
  Equations}}.
\newblock American Mathematical Society, 2002.

\bibitem{NedjalkovVitanov1990}
M.~Nedjalkov and P.~Vitanov.
\newblock {Application of the iteration approach to the ensemble Monte Carlo
  technique}.
\newblock {\em Solid-State Electron.}, 33:407--410, 1990.

\bibitem{KosinaNedjalkovSelberherr2000}
H.~Kosina, M.~Nedjalkov, and S.~Selberherr.
\newblock Theory of the {Monte Carlo} method for semiconductor device
  simulation.
\newblock {\em IEEE Trans. Electron Devices}, 47:1898--1908, 2000.

\bibitem{KosinaNedjalkovSelberherr2003}
H.~Kosina, M.~Nedjalkov, and S.~Selberherr.
\newblock The stationary {Monte Carlo} method for device simulation. {I}.
  {Theory}.
\newblock {\em J. Appl. Phys.}, 93:3553--3563, 2003.

\bibitem{NedjalkovKosinaSelberherr2003}
M.~Nedjalkov, H.~Kosina, and S.~Selberherr.
\newblock The stationary {Monte Carlo} method for device simulation. {II}.
  {Event biasing and variance estimation}.
\newblock {\em J. Appl. Phys.}, 93:3564--3571, 2003.

\bibitem{YanCaiZeng2013}
C.~Yan, W.~Cai, and X.~Zeng.
\newblock {A parallel method for solving Laplace equations with Dirichlet data
  using local boundary integral equations and random walks}.
\newblock {\em SIAM J. Sci. Comput.}, 35:B868--B889, 2013.

\bibitem{KosztinFaberSchulten1996}
I.~Kosztin, B.~Faber, and K.~Schulten.
\newblock Introduction to the diffusion {Monte Carlo} method.
\newblock {\em Am. J. Phys.}, 64:633--643, 1996.

\bibitem{HairerWeare2014}
M.~Hairer and J.~Weare.
\newblock Improved diffusion {Monte Carlo}.
\newblock {\em Commun. Pure Appl. Math.}, 67:1995--2021, 2014.

\bibitem{Wigner1932}
E.~Wigner.
\newblock On the quantum corrections for thermodynamic equilibrium.
\newblock {\em Phys. Rev.}, 40:749--759, 1932.

\bibitem{tatarskiui1983}
V.~I. Tatarski{\u\i}.
\newblock The {Wigner} representation of quantum mechanics.
\newblock {\em Sov. Phys. Usp}, 26:311--327, 1983.

\bibitem{Zurek1991}
W.~H. Zurek.
\newblock {Decoherence and the transition from quantum to classical}.
\newblock {\em Phys. Today}, October:36--44, 1991.

\bibitem{JacoboniBordone2004}
C.~Jacoboni and P.~Bordone.
\newblock The {W}igner-function approach to non-equilibrium electron transport.
\newblock {\em Rep. Prog. Phys.}, 67:1033--1071, 2004.

\bibitem{DiasPrata2004}
N.~C. Dias and J.~N. Prata.
\newblock Admissible states in quantum phase space.
\newblock {\em Ann. Phys.}, 313:110--146, 2004.

\bibitem{bk:MarkowichRinghoferSchmeiser1990}
P.~A. Markowich, C.~A. Ringhofer, and C.~Schmeiser.
\newblock {\em Semiconductor Equations}.
\newblock Springer-Verlag, Wien-New York, 1990.

\bibitem{th:Biegel1997}
B.~A. Biegel.
\newblock {\em Quantum Electronic Device Simulation}.
\newblock PhD thesis, Stanford University, 1997.

\bibitem{bk:Balescu1975}
R.~Balescu.
\newblock {\em Equilibrium and Nonequilibrium Statistical Mechanics}.
\newblock John Wiley \& Sons, New York, 1975.

\bibitem{bk:Schleich2011}
W.~P. Schleich.
\newblock {\em Quantum Optics in Phase Space}.
\newblock Wiley-VCH, Berlin, 2011.

\bibitem{SellierNedjalkovDimov2015}
J.~M. Sellier, M.~Nedjalkov, and I.~Dimov.
\newblock An introduction to applied quantum mechanics in the {Wigner Monte
  Carlo} formalism.
\newblock {\em Phys. Rep.}, 577:1--34, 2015.

\bibitem{bk:Leonhardt1997}
U.~Leonhardt.
\newblock {\em Measuring the Quantum State of Light}.
\newblock Cambridge University Press, New York, 1997.

\bibitem{LeibfriedPfauMonroe1998}
D.~Leibfried, T.~Pfau, and C.~Monroe.
\newblock {Shadows and mirrors: Reconstructing quantum states of atom motion}.
\newblock {\em Phys. Today}, April:22--28, 1998.

\bibitem{Zachos2002}
C.~Zachos.
\newblock Deformation quantization: quantum mechanics lives and works in
  phase-space.
\newblock {\em Int. J. Mod. Phys. A}, 17:297--316, 2002.

\bibitem{ShaoLuCai2011}
S.~Shao, T.~Lu, and W.~Cai.
\newblock Adaptive conservative cell average spectral element methods for
  transient {Wigner} equation in quantum transport.
\newblock {\em Commun. Comput. Phys.}, 9:711--739, 2011.

\bibitem{XiongChenShao2015}
Y.~Xiong, Z.~Chen, and S.~Shao.
\newblock An advective-spectral-mixed method for time-dependent many-body
  {Wigner} simulations.
\newblock {\em SIAM J. Sci. Comput.}, to appear, 2016 [{\em arXiv:1602.08853}]. 

\bibitem{bk:QuerliozDollfus2010}
D.~Querlioz and P.~Dollfus.
\newblock {\em The {Wigner Monte Carlo} Method for Nanoelectronic Devices: A
  Particle Description of Quantum Transport and Decoherence}.
\newblock Wiley-ISTE, London, 2010.

\bibitem{NedjalkovKosinaSelberherrRinghoferFerry2004}
M.~Nedjalkov, H.~Kosina, S.~Selberherr, C.~Ringhofer, and D.~K. Ferry.
\newblock Unified particle approach to {W}igner-{B}oltzmann transport in small
  semiconductor devices.
\newblock {\em Phys. Rev. B}, 70:115319, 2004.

\bibitem{NedjalkovSchwahaSelberherr2013}
M.~Nedjalkov, P.~Schwaha, S.~Selberherr, J.~M. Sellier, and D.~Vasileska.
\newblock {Wigner quasi-particle attributes -- An asymptotic perspective}.
\newblock {\em Appl. Phys. Lett.}, 102:163113, 2013.

\bibitem{SellierNedjalkovDimov2014}
J.~M. Sellier, M.~Nedjalkov, I.~Dimov, and S.~Selberherr.
\newblock A benchmark study of the {Wigner Monte-Carlo} method.
\newblock {\em Monte Carlo Methods Appl.}, 20:43--51, 2014.

\bibitem{ShaoSellier2015}
S.~Shao and J.~M. Sellier.
\newblock Comparison of deterministic and stochastic methods for time-dependent
  {Wigner} simulations.
\newblock {\em J. Comput. Phys.}, 300:167--185, 2015.

\bibitem{MuscatoWagner2016}
O.~Muscato and W.~Wagner.
\newblock {A class of stochastic algorithms for the Wigner equation}.
\newblock {\em SIAM J. Sci. Comput.}, 38:A1483--A1507, 2016.

\bibitem{th:Ellinghaus2016}
P.~Ellinghaus.
\newblock {\em Two-Dimensional Wigner Monte Carlo Simulation for Time-Resolved
  Quantum Transport with Scattering}.
\newblock PhD thesis, Institute for Microelectronics, TU Vienna, 2016.

\bibitem{Wagner2016}
W.~Wagner.
\newblock {A random cloud model for the Wigner equation}.
\newblock {\em Kinet. Relat. Mod.}, 9:217--235, 2016.

\bibitem{bk:Liu2001}
J.~S. Liu.
\newblock {\em {Monte Carlo Strategies in Scientific Computing}}.
\newblock Springer, New York, 2001.

\bibitem{bk:Harris1963}
T.~E. Harris.
\newblock {\em {The Theory of Branching Processes}}.
\newblock Springer-Verlag, Berlin, 1963.

\bibitem{JiangCaiTsu2011}
H.~Jiang, W.~Cai, and R.~Tsu.
\newblock Accuracy of the {Frensley} inflow boundary condition for {Wigner}
  equations in simulating resonant tunneling diodes.
\newblock {\em J. Comput. Phys.}, 230:2031--2044, 2011.

\bibitem{ArnoldLangeZweifel2000}
A.~Arnold, H.~Lange, and P.~F. Zweifel.
\newblock A discrete-velocity, stationary {Wigner} equation.
\newblock {\em J. Math. Phys.}, 41:7167--7180, 2000.

\bibitem{Goudon2002}
T.~Goudon.
\newblock Analysis of a semidiscrete version of the {Wigner} equation.
\newblock {\em SIAM J. Numer. Anal.}, 40:2007--2025, 2002.

\bibitem{GoudonLohrengel2002}
T.~Goudon and S.~Lohrengel.
\newblock On a discrete model for quantum transport in semi-conductor devices.
\newblock {\em Transport Theory Statist. Phys.}, 31:471--490, 2002.

\bibitem{bk:Pazy1983}
A.~Pazy.
\newblock {\em Semigroups of Linear Operators and Applications to Partial
  Differential Equations}.
\newblock Springer-Verlag, New York, 1983.

\bibitem{DimovNedjalkovSellierSelberherr2015}
I.~Dimov, M.~Nedjalkov, J.~M. Sellier, and S.~Selberherr.
\newblock Boundary conditions and the {Wigner} equation solution.
\newblock {\em J. Comput. Electron.}, 14:859--863, 2015.

\bibitem{bk:Kallenberg2002}
O.~Kallenberg.
\newblock {\em {Foundations of Modern Probability}}.
\newblock Springer, New York, second edition, 2002.

\bibitem{bk:Kress2014}
R.~Kress.
\newblock {\em {Linear Integral Equations}}.
\newblock Springer, New York, third edition, 2014.

\bibitem{KosinaNedjalkovSelberherr2004}
H.~Kosina, M.~Nedjalkov, and S.~Selberherr.
\newblock Solution of the space-dependent {Wigner} equation using a particle
  model.
\newblock {\em Monte Carlo Methods Appl.}, 10:359--368, 2004.

\bibitem{bk:RobertCasella2004}
C.~P. Robert and G.~Casella.
\newblock {\em {Monte Carlo Statistical Methods}}.
\newblock Springer, New York, second edition, 2004.

\bibitem{bk:HastieTibshiraniFriedman2009}
T.~Hastie, R.~Tibshirani, and J.~Friedman.
\newblock {\em {The Elements of Statistical Learning: Data Mining, Inference,
  and Prediction}}.
\newblock Springer, New York, second edition, 2009.

\bibitem{ShaoXiong2016}
S.~Shao and Y.~Xiong.
\newblock A computable branching process for the {Wigner} quantum dynamics.
\newblock {\em arXiv:1603.00159}, 2016.

\end{thebibliography}

\end{document}